\documentclass [12pt]{article}

\usepackage{amsmath,amsthm,amscd,amssymb}

\usepackage[small, centerlast, it]{caption}
\usepackage{epsfig}
\usepackage[titles]{tocloft}
\usepackage[latin1]{inputenc}
\usepackage{verbatim} %for comments via \begin{comment}
\usepackage[active]{srcltx} %inverse search
\usepackage{fancyhdr}
\usepackage{caption}
\usepackage{enumerate}

\def\sp{\hskip -5pt}

\def\b1{{1\!\!1}}

\def\cB{{\mathcal B}}

\def\cD{{\mathcal D}}

\def\cH{{\mathcal H}}

\def\cJ{{\mathcal J}}
\def\cL{{\mathcal L}}

\def\cP{{\mathcal P}}

\def\cK{{\mathcal K}}
\def\cZ{{\mathcal Z}}

\def\sK{{\mathsf K}}
\def\sM{{\mathsf M}}
\def\sN{{\mathsf N}}
\def\sH{{\mathsf H}}

\def\bC{{\mathbb C}}           %%%  complex numbers and so on

\def\bH{{\mathbb H}}

\def\bD{{\mathbb D}}

\def\bN{{\mathbb N}}

\def\bR{{\mathbb R}}

\def\bZ{{\mathbb Z}}

\def\gp{{\mathfrak{g}}}

\def\gB{{\mathfrak B}}

\def\gg{{\mathfrak g}}

\def\gR{{\mathfrak R}}
\def\gS{{\mathfrak S}}
\def\gT{{\mathfrak T}}
\def\gU{{\mathfrak U}}

\def\gZ{{\mathfrak Z}}

\def\beq{\begin{eqnarray}}
\def\eeq{\end{eqnarray}}

\usepackage{sectsty}
\sectionfont{\normalsize}%\normalfont
\subsectionfont{\normalsize\normalfont\itshape}
\newtheoremstyle{plain}
{5pt}% space above
{9pt}% space below
{\itshape}% body font
{}% h indent amount
{\itshape\bfseries}% theorem head font
{}% punctuation after theorem head
{1em}% space after theorem head
{}% theorem head spec (can be left empty, meaning `normal')

\theoremstyle{thm}
\newtheorem{theorem}{\em Theorem}[section]
\newtheorem{lemma}[theorem]{\em Lemma}
\newtheorem{corollary}[theorem]{\em Corollary}
\newtheorem{proposition}[theorem]{\em Proposition}
\newtheorem{definition}[theorem]{\em Definition}
\newtheorem{example}[theorem]{\em Example}
\newtheorem{remark}[theorem]{\em Remark}

\begin{document}
%%%%%%%%%%%%%   Title %%%%%%%%%%%%%%%%%%%%%%%%%%

\par
\bigskip
\large
\noindent
{\bf  Quantum theory in  quaternionic Hilbert space: How Poincar\'e symmetry reduces the theory to the  standard complex one.}
%{\bf  Why isn't quantum theory  represented in real Hilbert spaces? How a complex structure emerges in the Hilbert space  from Poincar\'e invariance}
\bigskip
\par
\rm
\normalsize

%%%%%%%%%%%%%%%%%%%%%%%%%%%%%%%%%%%%%%%%%%%%%
%%%%%%%%%%%% Authors %%%%%%%%%%%%%%%%%%%%%%%%

\noindent  {\bf Valter Moretti$^{a}$}, {\bf Marco Oppio$^{b}$}\\
\par

\noindent 
 $^a$ Department of  Mathematics University of Trento, and INFN-TIFPA \\
 via Sommarive 15, I-38123  Povo (Trento), Italy.\smallskip\\
 valter.moretti@unitn.it
\vspace{0.2cm}

\noindent
$^b$ Faculty of Mathematics, University of Regensburg,\\
Universit\"atstrasse 31, 93053 Regensburg, Germany \smallskip\\
marco.oppio@ur.de
\vspace{0.5cm}

\par

\rm\normalsize

\noindent {\small November 06,  2018}

%\linespread{1.5}
\rm\normalsize

%%%%%%%%%%%% Date %%%%%%%%%%%%%%%%%%%%%%%%%%

\par
\bigskip

\noindent
\small
{\bf Abstract}.
As earlier conjectured by several authors  and much later established by  Sol\`er, from the lattice-theory point of view  Quantum Mechanics  may be formulated in real, complex or quaternionic  Hilbert spaces only. 
 On the other hand no  quantum systems seem to exist that  are naturally described in a real or quaternionic Hilbert space. 
In a previous paper \cite{MO1}, we showed that any quantum system which is elementary from the viewpoint of the Poincar\'e symmetry group and it is initially described in a real  Hilbert space, it 
 can also be described within the standard complex-Hilbert space framework. This complex description is unique and more precise than the real one as for instance,  in the complex description, all self-adjoint 
operators represent observables defined by the symmetry group. The complex picture  fulfils the thesis of Sol\'er theorem and permits the standard formulation of quantum Noether's theorem.  The present 
work is devoted to investigate the remaining case, namely the possibility of a description of a relativistic elementary  quantum system in a quaternionic Hilbert space. Everything  is done  exploiting recent results  
of quaternionic spectral theory independently developed. 
In the initial part of this work,  we extend some results of group representation theory and von Neumann algebra theory from the real and complex case to the quaternionic Hilbert space case.
We prove the double commutant theorem also for quaternionic von Neumann algebras
(whose proof requires a different procedure with respect to the real and complex cases)
 and we extend to the quaternionic case a  result established in the previous paper  concerning the classification of irreducible von Neumann algebras into three categories.
  In the second part of the paper, we  consider  an  elementary relativistic system within  Wigner's approach  defined as a 
locally-faithful  irreducible strongly-continuous unitary representation of the Poincar\'e group in a quaternionic Hilbert space. We prove that, if the squared-mass operator is non-negative, the system admits a
  natural, Poincar\'e invariant and unique  up to sign, complex structure which commutes with the whole algebra of observables generated by the representation itself.  This complex structure
   leads to a physically equivalent reformulation of the theory in a complex Hilbert space. Within this complex  formulation, differently from what happens in the quaternionic one,  all selfadjoint 
   operators represent observables in agreement  with  Sol\`er's thesis,  the standard quantum version of Noether theorem may be formulated and the notion of composite system may be given   in terms of tensor product of elementary systems.
In the third part of the paper,  we focus on the physical hypotheses adopted to define a quantum elementary relativistic system relaxing them on the one hand, and making our model physically
 more general on the other hand.  We use a physically more accurate notion of irreducibility regarding the algebra of observables only, we describe the symmetries in terms of automorphisms of the
  restricted lattice of elementary propositions  of the  quantum system and we adopt a notion of continuity referred to the states viewed as probability measures on the elementary propositions.
 Also in this case, the final result proves that there exist a unique (up to sign) Poincar\'e invariant complex structure making the theory complex and completely fitting into Sol\`er's picture.  
 The overall conclusion is that relativistic elementary systems are naturally and better described  in complex Hilbert spaces even if starting from a real or quaternionic Hilbert space formulation 
 and this complex description is uniquely fixed by physics. 

\normalsize

\newpage
\tableofcontents

\section{Introduction} 

It is known that quantum theories can be formulated  in real, complex and quanternionic Hilbert spaces as we summarize below. A brief account  of basic real Hilbert space spectral theory appears in \cite{MO1} and a summary  on basic 
results on quaternionic Hilbert space theory is included in the next section, whereas for  the general spectral  theory of unbounded normal operators  we address to \cite{GMP1,GMP2}. 

\subsection{Quaternionic Hilbert spaces}
 $\bH := \{a1 + bi + cj + dk\:|\: a,b,c,d \in \bR\}$ henceforth denotes the real unital associative  algebra of quaternions. $i,j,k$ are the standard 
{\bf imaginary units} satisfying $i^2=j^2=k^2 = -1$ and $ij= -ji = k$, $jk=-kj = i$, $ki=-ik=j$ which give rise to the notion of associative,  distributive and non-commutative product in $\bH$ with $1$ as neutral element. $\bH$ is a
 division algebra, i.e., every non zero element admits a multiplicative inverse. The centre of $\bH$ is $\bR$.  
 $\bH$ is assumed to be equipped with the {\bf quaternionic conjugation} $\overline{a1 + bi + cj + dk}= a1 - bi - cj -dk$. Notice that the conjugation satisfies $\overline{qq'}= \overline{q'} \overline{q}$ and $\overline{\overline{q}}=q$ for all $q,q'\in \bH$.
If $q \in \bH$, its {\bf real part} is defined as $Re \: q := \frac{1}{2}(q+\overline{q}) \in \bR$.
The quaternionic conjugation together with the Euclidean {\bf norm} $|q|:=\sqrt{q\overline{q}}$ for $q\in \bH$, makes $\bH$ a real unital $C^*$-algebra which also satisfies the {\bf composition algebra} property $|qq'|=|q|\:|q'|$. 

\begin{definition} {\em A {\bf quaternionic vector  space} is  
an additive Abelian group $(\sH, +)$ denoting the sum operation, equipped with a right-multiplication $\sH \times \bH \ni (x, q) \mapsto xq\in \sH$ such that (a) the right-multiplication  is distributive with respect to $+$,  (b) the sum of 
quaternions  is distributive with respect to the right-multiplication,  (c)  $(xq) q' = x(qq')$ and (d) $x1=x$  for all $x\in \sH$ and $q,q' \in \bH$.}
\end{definition}
\begin{definition}  
{\em A {\bf quaternionic Hilbert space}  is a quaternionic vector space $\sH$ equipped with a {\bf Hermitian quaternionic scalar product}, i.e., a map $\sH \times \sH \ni (x,y) \mapsto \langle x|y \rangle \in \bH$ such
 that  (a) $\langle x|yq+z\rangle = \langle x|y\rangle q+\langle x|z\rangle$ for every $x,y,z\in \bH$ and $q \in \bH$, (b)  $\langle x|y\rangle = \overline{\langle y|x\rangle}$ for every $x,y \in \sH$ and (c) $\langle x| x\rangle \in [0,+\infty)$ 
where (d) $\langle x| x\rangle=0$ implies $x=0$,    and  $\sH$ is complete with respect to the norm $||x|| = \sqrt{\langle x| x \rangle}$. }
\end{definition}
\noindent 
The standard {\bf Cauchy-Schwartz} inequality holds, $|\langle x|y \rangle| \leq ||x||\: ||y||$ for every $x,y \in \sH$ for the above defined quaternionic Hermitian scalar product \cite{GMP1}.
The notion of {\bf Hilbert basis}  is the same as for real and complex Hilbert spaces and properties are the same with obvious changes. A quaternionic Hilbert space turns out to be separable as a metrical space if and only if it admits
 a finite or countable Hilbert basis. The notion of {\bf orthogonal subspace}  $S^\perp$ of a set $S \subset \sH$ is defined with respect to $\langle \cdot | \cdot \rangle$  and enjoys the same standard properties as for the analogue in real and complex Hilbert spaces.
 The notion of operator norm and  bounded operator are the same as for real and complex Hilbert spaces. Since the {\bf Riesz lemma} \cite{GMP1}  holds true also for quaternionic Hilbert spaces, the {\bf adjoint operator}
 $A^* : \sH \to \sH$ 
of a bounded quaternionic linear operator $A : \sH \to \sH$ can be defined as the unique quaternionic linear operator such that $\langle A^*y|x\rangle = \langle y|Ax\rangle$ for every pair $x,y \in \sH$. Notice that if $A : \sH \to \sH$ is quaternionic linear
and $r\in \bR$, we can define the quaternionic linear operator $rA: \sH \to \sH$ such that $rA x := (Ax)r$ for all $x \in \sH$. Replacing $r$ for $q\in \bH$ produces a non-linear map in view of non-commutativity of $\bH$. Therefore only real linear combinations of 
quaternionic linear operators are well defined.
$\gB(\sH)$ denotes the real unital $C^*$-algebra of bounded operators over $\sH$. The notion of {\bf orthogonal projector} $P: \sH \to \sH$ is defined exactly as in the real or complex Hilbert space case,
$P$ is bounded,  $PP=P$ and $P^*=P$. Orthogonal projectors $P$ are one-to-one with the class of 
closed subspaces $P(\sH)$ of $\sH$. $\cL(\sH)$ denotes the orthocomplemented complete lattice (see below) of orthogonal projectors of $\sH$. 
Another important notion is that  of {\bf square root} of positive bounded operators. As for the real and complex case, also for quaternionic Hilbert spaces, if $A$ is bounded and positive, then there exists a unique bounded positive 
operator $\sqrt{A}$ such that $\sqrt{A}\sqrt{A}=A$ \cite{GMP1,GMP2}. In particular, if $A : \sH \to \sH$ is a bounded quaternionic-linear operator $|A| := \sqrt{A^*A}$ is well defined positive and self-adjoint.
For the proofs of the afore-mentioned properties and for more advanced issues, especially concerning spectral theory,  we address the reader to \cite{GMP1} and \cite{GMP2}.

\begin{remark}\label{remquatV2}{\em
In \cite{V2} and \cite{Em} the Quaternionic Hilbert space is defined assuming  a \textit{left}-multiplication $\bH\times\sH\ni (q,u)\mapsto qu\in\sH$ and a Hermitian quaternionic 
scalar product $\sH\times\sH\ni (u,v)\mapsto \langle u|v\rangle\in\bH$ whose only difference resides in point (a): $\langle qx|y\rangle=q\langle x|y\rangle$ for all $x,y\in\sH$ and $q\in\bH$. To 
define a left-multiplication on a space with right-multiplication it suffices to define $qu:=u\overline{q}$ for all $q\in\bH$ and $u\in\sH$, while the scalar product does not need to be changed. A
 map $A:\sH\rightarrow \sH$ is linear, bounded, self-adjoint, idempotent and unitary with respect to the right-multiplication if and only if it has the same properties with respect to the left-multiplication. 
This permits us to exploit indifferently the results in \cite{GMP1,GMP2} and \cite{V2},\cite{Em}.}
\end{remark}

\subsection{The lattice of elementary propositions of a quantum system}
Quantum theory can basically be formulated  as  a non-Boolean probability theory over the partially ordered set of {\bf  elementary propositions} $\cL$ about the given physical quantum system \cite{BeCa,V2,Redei}.  Let us review 
some elementary ideas on this subject restating the discussion already present in the introduction of \cite{MO1}.
{\bf Elementary propositions}, also called {\bf elementary observables}, are the experimentally testable propositions with possible outcomes $0$ and $1$. The partial order relation  $\leq$  in $\cL$ corresponds to the
 logical implication (see \cite{Mackey,BeCa,V2,librone} for the various  interpretations). It is generally  supposed  that the partially ordered set $\cL$ is  a {\bf lattice} (with some noticeable exception as  \cite{Mackey}):  
A  pair of elements $a,b \in \cL$ always admits  $\inf\{a,b\} \in \cL=:a \wedge b$  called 
  {\bf meet},  and   $\sup\{a,b\}\in \c=:a \vee b$ called 
 {\bf join}. It is easy to prove that $a\le b$ if and only if $a=a\wedge b$ and  that   $\vee$ and $\wedge$ are separately {\em symmetric} and {\em associative} in every lattice.
The  $\cL$  is also requested to  be {\bf bounded}: A {\em minimal element} ${\bf 0}$, the always false proposition,  and a {\em maximal element} ${\bf 1}$, the always true proposition,  exist  in $\cL$.  $\cL$ is also 
assumed to be ($\sigma$-){\bf complete}, i.e., $\sup A$ and $\inf A$ exist for every (countable) set $A\subset \cL$.
$\cL$  is finally demanded  to be {\bf orthocomplemented}:  If $a\in \cL$, an {\bf orthogonal complement}  $a^\perp\in \cL$ exists interpreted as the logical negation of $a$.
By definition the orthocomplement  satisfy $a \vee a^\perp = {\bf 1}$,  $a \wedge a^\perp = {\bf 0}$,   $(a^\perp)^\perp = a$,
and  $a\leq b$ implies $b^\perp  \leq a^\perp$ for any $a,b \in \cL$. Now  $a,b \in \cL$ are {\bf orthogonal}, written $a \perp b$, if $a \leq b^\perp$ (equivalently $b \leq a^\perp$).\\
Pairs of mutually {\em compatible} elementary propositions on a quantum system (those which are simultaneously testable by means of experiment) are described by pairs of {\bf commuting} elements $p,q \in \cL$  in the sense of abstract 
orthocomplemented lattices \cite{BeCa}: the sublattice {\bf generated} by $\{p, q\}$, namely the intersection of all orthocomplemented sublattices of  $\cL$ which include $\{p,q\}$  is  {\bf  Boolean}, i.e.,  $\vee$ and $\wedge$ are mutually {\em distributive}.
A maximal set of pairwise compatible propositions is a complete Boolean sublattice and an interpretation in terms of {\em classical logic} is appropriate. Compatibility is not transitive and so the structure of maximal Boolean sublattices of $\cL$ is very complicated.
The whole lattice ${\cal L}$ of elementary propositions of a quantum system  is however  {\em non-Boolean} because  $\wedge$ and $\vee$  are  not mutually distributive. Physically  speaking this is due  to  the existence of  {\em incompatible} 
elementary propositions (e.g., see \cite{BeCa,M2}).
The quantum lattice $\cL$ enjoys a set of peculiar properties that can be  phenomenologically motivated (e.g., see \cite{BeCa}) even if some non-trivial interpretative problems remain \cite{librone}: 
(i) {\bf orthomodularity},  (ii) {\bf $\sigma$-completeness},  (iii) {\bf  atomicity}, (iii)' {\bf atomisticity},  (iv) {\bf covering property}, (v) {\bf separability}, (vi) {\bf irreducibility}  (see the appendix of \cite{MO1} for a brief illustration of these definitions).

\subsection{The coordinatisation problem}
As it will be useful below and in the rest of the paper, we remind  the reader that if $\cL_1$, $\cL_2$ are  orthocomplemented  lattices,  a map $h : \cL_1 \to \cL_2$ is a {\bf lattice homomorphism} if $h(a\vee_1 b) = h(a) \vee_2 h(b)$, 
 $h(a\wedge_1 b) = h(a) \wedge_2 h(b)$, $h(a)^{\perp_2} = h(a^{\perp_1})$ if $a,b \in \cL_1$, $h({\bf 0}_1)= {\bf 0}_2$, $h({\bf 1}_1)= {\bf 1}_2$.
When  the lattices are complete, resp. $\sigma$-complete, the first pair of conditions are made stronger to
 $h(\sup_{a\in A} a) = \sup_{a\in A} h(a)$ and $h(\inf_{a\in A} a) = \inf_{a\in A}  h(a)$ for every 
 infinite, resp. countably infinite, subset $A \subset \cL_1$. 
A straightforward calculation shows that $a\le_1 b$ implies $h(a)\le_2 h(b)$.
A bijective lattice homomorphism is a {\bf lattice isomorphism}. The inverse map of a lattice isomorphism is a lattice isomorphism as well. {\bf Lattice automorphisms} are isomorphisms with $\cL_1=\cL_2$; they form  a group, denoted by $\mbox{Aut}(\cL_1)$.\\
The  long standing {\em coordinatisation problem} \cite{BeCa} consisted of proving (or disproving) that an abstract bounded orthocomplemented lattice ${\cal L}$ satisfying  (i)-(vi) and possibly some added technical requirements,
is necessarily isomorphic to the lattice ${\cal L}(\sH)$ of the orthogonal projectors/closed subspaces of a {\em complex} Hilbert space $\sH$. Here  the partial order relation is the inclusion of the projection closed subaspace of the considered 
orthogonal projectors. This should provide a justification of the standard Hilbert-space formulation of Quantum Theory.
Some intermediate and fundamental results by Piron \cite{Piron} and  Maeda-Maeda \cite{MM} 
established that such ${\cal L}$, if contains  four orthogonal atoms at least, is necessarily isomorphic to the lattice of  the {\em orthoclosed} subspaces ($K=K^{\perp\perp}$)  of a structure generalizing a vector space over a division ring 
$\bD$ equipped with a suitable involution operation, and admitting a generalized non-singular $\bD$-valued Hermitian scalar product (giving rise to the above mentioned notion of orthogonal $^\perp$).
The order relation of  this concrete lattice is the standard inclusion of orthoclosed subspaces.  In 1995 Sol\`er \cite{Soler}  achieved the perhaps conclusive result  (for alternative equivalent statements see \cite{Holland} and \cite{XY}).\\ 

\noindent {\bf [Sol\'er's Theorem].}
{\em Consider an orthocomplemented lattice $\cL$ satisfying (i)-(vi), such that (vii) it
contains at least four orthogonal atoms 
and (viii) $\cL$ includes an infinite orthogonal sets of atoms with unit (generalized) norm. Then  $\cL$ is  isomorphic to the lattice ${\cal L}(\sH)$ of (topologically) closed subspaces  of a separable Hilbert space $\sH$ with
 set of scalars given by either the fields $\bR$, $\bC$ or the real division algebra of quaternions $\bH$}.\\

\noindent In all three cases,  the partial order relation of the lattice is again the standard inclusion of closed subspaces and 
$\sM\vee\sN$ corresponds to the closed span of the union of the closed subspaces $\sM$ and $\sN$, whereas  $\sM\wedge\sN:=\sM \cap \sN$. The minimal element  is the trivial subspace $\{0\}$ and the maximal element is
  $\sH$ itself. Finally, the  orthocomplement of $\sM \in \cL(\sH)$ is described by the standard orthogonal $\sM^\perp$ in $\sH$. All the structure can equivalently be rephrased in terms of orthogonal projectors $P$ in $\sH$, since they are one-to-one associated with 
the closed subspaces   of $\sH$ identified with their images $P(\sH)$. In particular $P\leq Q$ corresponds to
$P(\sH)\subset Q(\sH)$ for $P,Q\in \cL(\sH)$.\\
Dropping  irreducibility requirement  of $\cL$ in Sol\`er's theorem, physically corresponding to absence of {\em superselection rules}, an orthogonal  direct sum of many such Hilbert spaces (even over different set of scalars)  replaces the single Hilbert space $\sH$. \\
Sol\`er's theorem assumes a list of  rigid requirements on the structure of $\cL$ and  the thesis   represents an  equally rigid picture. An evident physical lack in  the hypotheses of the theorem is the absence of any fundamental physical symmetry
 requirement, according to Galileo or Poincar\'e groups.
In complex Hilbert spaces, only type-$I$ factors are permitted by the thesis of Soler's theorem to represent the algebra of observables $\gR$  and no gauge group may enter the game therefore excluding systems of {\em quarks}  where internal
 symmetries (colour $SU(3)$) play a crucial r\^ole.  Sol\`er's   picture is evidently not appropriate also  to describe   non-elementary quantum systems like  {\em pure phases} of extended quantum thermodynamic systems. There the algebra of
 observables is still a factor, but the 
 type-$I$ is not admitted in general due to the presence of a non-trivial commutant $\gR'$.  Also {\em localized} algebras of observables in QFT are not encompassed by Sol\`er's framework. 
However, {\em elementary relativistic systems} like elementary particles in Wigner's view are in agreement with  Sol\`er's picture  in {\em complex} Hilbert spaces. Indeed, these  systems are described as {\em irreducible} unitary 
 representations of Poincar\'e group and, supposing  that the von Neumann algebra of observables is that generated by the representation,  Schur's lemma  demonstrates that  the algebra of observables is the entire $\gB(\sH)$. Therefore the lattice of  
elementary propositions is the entire $\cL(\sH)$ in agreement with  the thesis of Sol\`er's theorem. What happens when changing the set of scalars of the Hilbert space, passing from $\bC$ to $\bR$ or $\bH$ is not obvious.

\subsection{Theoretical notions in common with the three types of Hilbert spaces formulations}\label{secGLEASON}
The following theoretical notions used to axiomatise  quantum theories  are defined in a  separable Hilbert space $\sH$, with scalar product $( \cdot|\cdot)$, over $\bR$, $\bC$ or $\bH$ respectively and referring to  the $\sigma$-complete
 orthocomplemented lattice $\cL(\sH)$ of orthogonal projectors in $\sH$.  However these notions are defined also replacing $\cL(\sH)$ for a smaller  lattice $\cL_1(\sH) \subset \cL(\sH)$, provided it is still  orthocomplemented and $\sigma$-complete 
(and therefore also orthomodular and separable). For future convenience, we shall list these notions below in this generalized case. This list is more or less identical to that  appearing in the introduction of \cite{MO1}.

(1) {\bf Elementary  observables} are represented by the orthogonal projectors in $\cL_1(H)$. Two such projectors are said to be {\bf compatible} if they commute as operators.
Indeed the abstract commutativity notion  of  elementary observables turns out to be equivalent to the standard  commutativity of associated orthogonal projectors \cite{BeCa,M2}.   

(2) {\bf Observables} are the  Projector-Valued Measures  (PVMs) over the real Borel sets (see \cite{MO1} for the real and complex case and \cite{GMP2} for the quaternionic case), taking 
values in ${\cal L}_1(\sH)$ $$ {\cal B}(\mathbb R) \ni E \mapsto P^{(A)}(E)\in {\cal L}_1(\sH)\:.$$ 
Equivalently, \cite{V2} an observable  is a selfadjoint operator $A : D(A) \to \sH$ with $D(A)\subset \sH$ a dense subspace
such that the associated projector-valued measure is made by elements of ${\cal L}_1(\sH)$. 
The link with the previous notion is the  statement of the spectral theorem for selfadjoint operators
 $A = \int_{\sigma(A)} \lambda dP^{(A)}(\lambda)$ (\cite{MO1} for the real and complex case, for the
 quaternionic case see \cite{GMP2}).\\ Obviously the meaning of each  elementary proposition $P^{(A)}(E)$ is {\em the outcome of the measurement of $A$ belongs  to the real Borel set $E$}.\\
Evidently, $\cL_1(\sH)= \cL(\sH)$ if and only if  {\em every} selfadjoint operator in $\sH$ represents an observable.
A selfadjoint operator, in particular an observables, $A$ is said to be {\bf compatible} with another  selfadjoint operator, in particular an observables, $B$ when the respective PVMs are made of pairwise commuting projectors.

(3) {\bf Quantum states} are defined  as {\bf $\sigma$-additive  probability measures} over ${\cal L}_1(\sH)$, that is maps
$\mu : {\cal L}_1(\sH) \to [0,1]$
such that $\mu(I)=1$ and $$\mu\left(s\mbox{-}\sum_k P_k\right) = \sum_k \mu(P_k)\quad \mbox{if $\{P_k\}_{k \in \bN}\subset \cL_1(\sH)$ with $P_kP_h=0$ for $h \neq k$,}$$ 
$s\mbox{-}\sum_k$ denoting the series in the strong operator topology.
$\mu(P)$ has the meaning of {\em the probability that the outcome of $P$ is $1$ if the proposition is tested when the state is $\mu$}. \\
If $\cL_1(\sH)= \cL(\sH)$ for $\sH$ separable with  $+\infty \geq \dim (\sH) \neq  2$ (always assumed henceforth),
these measures are in one-to-one correspondence with all of the selfadjoint positive, unit-trace, trace-class operators $T_\mu: \sH\to \sH$ according to
\beq \mu(P)= tr(PT_\mu P)\quad \forall P \in  {\cal L}(H)\:.\label{GLSTAT}\eeq
$T\in \gB(\sH)$ is {\bf trace-class}  \cite{GleasonQuat} 
if $$\sum_{u\in N}\langle u||T|u\rangle <+\infty \quad \mbox{ for some Hibert basis $N\subset \sH$,}$$ where $\sH$ is a real, complex or quaternionic Hilbert space. The set of trace-class operators turns out to be 
a closed two-sided $^*$-ideal of $\gB(\sH)$ (the unital real $C^*$-algebra of bounded operators $A: \sH \to \sH$) in the three considered cases \cite{GleasonQuat}.\\
This correspondence between $\mu$ and $T_\mu$ exists  for  the three cases as demonstrated by the celebrated {\em Gleason's theorem} valid for $\bR$ and $\bC$ \cite{G}. The quaternionic case is more complicated 
and the extension proposed by Varadarajan in  \cite{V2} was partially incorrect.  A correct statement has been recently obtained by the authors of this paper \cite{GleasonQuat}. 
The problem is that the notion of trace in quaternionic Hilbert space is necessarily {\em basis dependent} unless the argument of the trace is selfadjoint.  Above $PT_\mu P$ is explicitely selfadjoint and the{\em  cyclic property}
 of the trace together with $PP=P$
proves that $tr(PT_\mu P)= tr(PT_\mu)$ in the complex and real cases, finding the standard statement of Gleason theorem in those cases. Cyclicity of the trace does not hold in the quanternionic case \cite{GleasonQuat}.
 An alternative, equivalent,  and much more effective approach \cite{GleasonQuat} 
is to state Gleason's identity in the three cases as 
\beq \mu(P)= Re \left(tr(PT_\mu)\right)\quad \forall P \in  {\cal L}(H)\:,\label{GLSTAT2}\eeq
where, for a quaternionic Hilbert space, $tr(PT_\mu)$ is computed on a Hilbert basis fixed arbitrarily.
In fact, it turns out that the real part of the trace is {\em basis independent} in quaternionic Hilbert space (and also in the remaining two cases).
Finally  (\ref{GLSTAT}) is equivalent to  (\ref{GLSTAT2}) because   $Re \left(tr(PT_\mu)\right)= tr(PT_\mu P)$ by elementary properties of the trace operation in the three considered cases (see \cite{GleasonQuat} for details).

Gleason's result
is valid  (but the correspondence of measures and positive unit-trace thrace-class operators ceases to be one-to-one) for separable complex Hilbert spaces when $\cL_1(\sH)\subsetneq \cL(\sH)$ and $\cL_1(\sH)$ is the projector lattice of a von Neumann algebra whose canonical 
decomposition into definite-type von Neumann algebras does not contain type-$I_2$ algebras \cite{libroGleason}.

 (4) {\bf Pure states} are  extremal elements of the convex body of the afore-mentioned probability measures. If $\cL_1(\sH)= \cL(\sH)$ pure states  are one-to-one with unit vectors of $\sH$ up to {\bf (generalized) phases} $\eta$, i.e., numbers 
of $\bR, \bC, \bH$  respectively,  with $|\eta|=1$. In this case, the notion of {\bf  probability transition} $|\langle \psi|\phi\rangle|^2$ of a pair of pure states defined by unit vectors $\psi,\phi$
can be introduced.
$|\langle\psi|\phi\rangle|^2 = \mu_\psi(P_\phi)$ is the probability that $P_\phi$ is true when the state is $\mu_\psi$, where  $P_\phi = \langle\phi| \cdot \rangle\phi$ and  $\mu_\psi := \langle\psi| \cdot \psi\rangle$.

(5) {\bf L\"uders-von Neumann's post measurement axiom} can be formulated in the standard way in the three cases:
{\em If the outcome of the ideal measurement of $P \in \cL_1(\sH)$ in the state $\mu$ is $1$, the post measurement state is} 
$$\mu_P(\cdot) := \frac{\mu(P \cdot P)}{\mu(P)}\:.$$
If $\cL_1(\sH)= \cL(\sH)$, we may define states in terms of trace class operators and, with obvious notation,
$T_P = \frac{1}{tr(PTP)}PTP$. In terms of probability measures over $\cL(\sH)$, this  is equivalent to say that the post measurement measure $\mu_P$, when the state before the measurement of $P$ is $\mu$,  is 
the {\em unique} probability-measure over $\cL(\sH)$  satisfying  the natural requirement of conditional probability
$\mu_P(Q) = \frac{\mu(Q)}{{\mu(P)}}$, for every $Q\in \cL(\sH)$ with $Q\leq P$.

(6) {\bf Symmetries} are naturally defined  as {\em automorphisms} $h : \cL_1(\sH) \to \cL_1(\sH)$ of the lattice of elementary propositions. A {\em subclass} of symmetries $h_U$ are those induced by 
unitary (or also anti unitary in the complex case) operators $U\in \gB(\sH)$ by means of 
 $h_U(P) := UPU^{-1}$ for every $P \in \cL_i(\sH)$. 
Alternatively, another definition of symmetry is as {\em automorphism} of the {\em Jordan algebra of observables} constructed out of $\cL_1(\sH)$. 
If $\cL_1(\sH)=\cL(\sH)$, following Wigner, symmetries can be defined as {\em bijective}   {\em probability-transition preserving} transformations of pure states to pure states. \\
With the maximality hypothesis on the lattice, the three  notions of symmetry coincide. In this case, {\em all} symmetries turn out to be  described  by  unitary  (or anti unitary in the complex case) operators, 
up to constant phases of $\bR$, $\bC$, $\bH$, respectively  due to well known theorems by Kadison, 
Wigner and Varadarajan \cite{Simon,V2}. 

(7) {\bf Continuous symmetries} are one-parameter groups of lattice automorphisms $\bR \ni s \mapsto h_s$, such that 
$\bR \ni s \mapsto \mu(h_s(P))$ is continuous for every $P\in \cL_1(\sH)$ and every quantum state $\mu$ ($\bR$ may be replaced for a topological group but we  stick here to the simplest
case). The {\bf time evolution} of the system $\bR \ni s \mapsto \tau_s$ is a preferred  continuous symmetry parametrized over $\bR$. 

(8) A {\bf dynamical symmetry} is a continuous symmetry $h$ which  commutes with the time evolution, $h_s \circ \tau_t = \tau_t \circ h_s$ for $s,t \in \bR$.\\
If $\cL_1(\sH)= \cL(\sH)$, every  continuous symmetry $\bR \ni s \mapsto h_s$ is represented by a strongly continuous
one-parameter group of unitary operators $\bR \ni s \mapsto U_s$ such that $h_s(P)= U_sPU_s^{-1}$ for all $P\in \cL(\sH)$ \cite{V2}.
Versions of {\em Stone theorem}  hold in the three considered cases $\bR$, $\bC$ and $\bH$ (the  validity in the quaternionic case easily arises 
form the theory developed in \cite{GMP2} and we will present a proof in Sect. \ref{SECstone}), 
proving that   $U_s = e^{sA}$ for some {\em anti}-selfadjoint operator $A$, uniquely determined by $U$.
 In the complex case, if $\bR \ni s \mapsto e^{sA}$ is also a  dynamical symmetry, the {\em selfadjoint} operator $-iA$, which is an observable the lattice being maximal, is invariant  under the natural adjoint 
action of time evolution $\tau$ unitarily represented by $\bR \ni t \mapsto V_t$,
and thus $-iA$ is a {\bf constant of motion}, $V_t^{-1}(-iA) V_t =-iA$ for every $t\in \bR$. This is the celebrated quantum version of {\em Noether theorem}. In the real Hilbert space case, no such simple result exists, since
we have no general way to construct a selfadjoint operator out of an anti selfadjoint operator $A$ in  absence of $i$. There is no unitary operator  $J$ corresponding to the imaginary unit $iI$ which commutes with 
 the  anti selfadjoint generator $A$ of every possible continuous symmetry (the time evolution in particular), thus producing 
an associated observable $JA$ which is a constant of motion.  Such an operator  however may exist for  one or groups of observables.
 In the quaternionic case, contrarily, there are many, pairwise  non-commuting, imaginary unities as recently established 
 \cite{GMP2}.
An interesting physical discussion on these partially open  issues for the quaternionic formulation appears in \cite{Adler}. 

(9) [Real and complex formulations only] {\bf Composite systems} in real and complex Hilbert space formulations are simply described by taking the (Hilbert) tensor product of the Hilbert spaces of subsystems 
and constructing the corresponding tensor product structures (e.g., see \cite{M2}).
Yet a fundamental obstruction arises with the quaternionic formulation, where a standard  notion of tensor product is forbidden due to non-commutativity of quaternions. As a matter of fact, basic properties of 
the notion of the tensor product lead to the indentities  $((u q)\otimes (v p)) =  (u \otimes (v p))q = (u\otimes v) pq$  and 
$((u q)\otimes (v p)) =  ((uq) \otimes v)p = (u\otimes v) qp$ producing the  contradiction $pq=pq$.  This is a long standing problem with some inconclusive attempt of solutions.

\subsection{Main results of this work}
As soon as one assumes the apparently quite natural hypotheses on the structure of the lattice of elementary propositions of a quantum system as stated in Sol\'er's theorem, the above three  Hilbert-space
 formulations are the only possibilities.  Actually, as already remarked,  Sol\'er hypotheses appear to be quite rigid since they force the theory to be formulated in terms of type-I factors when looking at the
 von Neumann algebra of observables in the complex case.  Though it is the standard  picture in quantum mechanics (in a single superselection sector), it rules out quantum field theory and statistical mechanics  of extend systems.
In \cite{MO1} we focused on a real-Hilbert space formulation without assuming all Sol\'er hypotheses  but supposing  some restrictions concerning symmetries. As a matter of fact,  we considered a weaker and
 more general notion  of Wigner  elementary relativistic systems and we proved that a natural relativistic-invariant complex structure with a very precise physical meaning 
exists in the real Hilbert space formulation,  making complex the theory. This complex formulation also removes some redundancies present in the real formulation. In particular all complex-linear self-adjoint 
operators are observables, whereas many real-linear self-adjoint operators cannot be interpreted as observables in the real formulation and the correspondence between pure-states and rays is one-to-one in the complex formulation but not in the real one.
The passage from the real to the complex theory also permits one to recover the standard relation between continuous symmetries and conserved observables along the evolution of the system.
 Here we focus on  elementary relativistic systems  initially described in quaternionic Hilbert spaces. 

The main goal and result of this work is the proof that, also in the case of a quaternionic formulation with Poincar\'e symmetry,  the theory can be re-formulated into a standard complex Hilbert space picture, where
 all self-adjoint operators represent observables, the standard version of Noether relation between continuous symmetries and conserved observables is restored 
and the notion of composite system can be implemented by the standard tensor product. The complex structure is, as in \cite{MO1}, uniquely imposed by Poicar\'e symmetry and is Poincar\'e invariant.\\ 
We shall establish this result into a pair of distinct theorems. In Theorem \ref{poinccomplexstructure} we prove the thesis referring to a notion of relativistic elementary system very close to Wigner's one: an irreducible 
strongly continuous representation of Poincar\'e group with non negative squared mass, whose observables are determined by the representation itself.
A second version of the result is proved in Theorem \ref{secondmain}. Here a much more sophisticated notion of relativistic elementary system is adopted defined in terms of an irreducible von Neumann algebra of 
observables whose lattice of orthogonal projector admits a continuous (with respect to a physically meaningful topology) irreducible representation of Poincar\'e group and such that the algebra itself is determined by 
the representation in view of a generalised version of Wigner's theorem.\\
 Differently from the real Hilbert space case, here we need an intermediate non-obvious technical result concerning the notion of {\em von Neumann algebra} in quanternionic Hilbert space. Indeed, in section \ref{VN1}, 
relying upon the spectral theory developed in \cite{GMP1} and \cite{GMP2}, we state and prove the quaternionic version of the celebrated {\em double commutant theorem}. 
A new proof is necessary because the standard argument (e.g., see \cite{KR}) does not work due the non-commutativity of the algebra of quaternions.
With the help of the von Neumann algebra machinery and extending some well-known results of the theory or representations of Lie groups  we will achieve our final result after having introduced a (weaker) extension  
of Wigner's  notion of elementary relativistic system\footnote{The idea of the proof that is very similar to that exploited in the real Hilbert space case \cite{MO1} appears in $\:\:\:$ www.mi.infn.it/$\sim$vacchini/talks$\_$bell17/Oppio.pdf $\:\:\:$
 after a talk of one of the authors given at the University of Milan (Department of Physics) on  June 16th 2016.}.

\section{Quaternionic operators and their natural structures}

\subsection{von Neumann algebras in quaternionic Hilbert spaces}\label{VN1}
This section is devoted to extending the notion of von Neumann algebra to quaternionic Hilbert spaces.
 The overall idea is to extend  classical von Neumann's result known as {\em double commutant theorem} to algebras of quaternionic right-linear 
 operators and to define the notion of von Neumann algebra exploiting that result.  The way to do it essentially is a reduction procedure from the quaternionic to the real and complex 
 cases where the theory of von Neumann algebras is well established. This goal needs a few  preliminary elementary results and 
 constructions we shall employ also in other parts of this work.\\
 As a first step we associate a quaternionic Hilbert space $\sH$ with a corresponding real Hilbert space $\sH_\bR$ and a corresponding complex Hilbert space $\sH_{\bC_j}$
   and we study the interplay of these structures  and of the corresponding operator 
 algebras. We remind the reader  that, as previously stressed,  only \textit{real} linear combinations of quaternionic-linear operators can be performed 
 so that for instance $\gB(\sH)$ is  a real $C^*$-algebra and not a quaternionic one.
 \begin{proposition}\label{propHR}
Let  $(\sH,\langle\cdot|\cdot\rangle)$ be a quaternionic Hilbert space and consider the associated real vector space on the set of vectors
$\sH_\bR:=\sH$, whose real linear combinations are the quaternionic ones with real coefficients,  and define the real  scalar product $(\cdot|\cdot):=
Re\langle\cdot|\cdot\rangle$ on $\sH_\bR$.
The following facts are true.
\begin{itemize}

\item[(a)] The norms of $\sH$ and $\sH_\bR$ coincide and $(\sH_\bR,(\cdot|\cdot))$ is a real Hilbert space.

\item[(b)]
Referring to the $\bR$-linear operators  $\cJ,\cK \in \gB(\sH_\bR)$ defined  by $$\cJ u:=uj\:, \quad \cK u:=uk \quad \forall u\in\sH_\bR$$ satisfying  $\cJ\cJ=\cK\cK=-I$ 
and $\cJ\cK=-\cK\cJ$, 
we have
\beq\label{recsp} 
\langle x|y\rangle = (x|y)-(x|\cJ y)j-(x|\cK y)k-(x|\cJ\cK y)jk \quad \forall x,y \in \sH
\eeq

\item[(c)] Every quaternionic-linear operator is  real-linear and in particular   $\gB(\sH)\subset \gB(\sH_\bR)$, where $\gB(\sH)$ is  the unital real
$C^*$-algebra of bounded operators on $\sH$  and the inclusion of algebras  preserves the norms of the operators.
 
 \item[(d)] $A\in\gB(\sH_\bR)$ belongs to $\gB(\sH)$ if and only if it commutes with both $\cJ$ and $\cK$. 

\item[(e)] Let $A : D(A)\rightarrow \sH$ be a quaternionic-linear operator  with $D(A) \subset \sH$ a (quaternionic) subspace and consider 
$A$ as a real-linear operator on $\sH_\bR$ with domain $D(A)$ viewed as real subspace of $\sH_\bR$.
Then

(i)   $\cJ A\subset A\cJ$ and $\cK A\subset A\cK$,

(ii)  $D(A)$ is a dense real-linear subspace of $\sH_\bR$ if it is dense in  $\sH$,
 
 (iii) the adjoint of $A$ with respect to $\sH$ and with respect to $\sH_\bR$ coincide,
 
 (iv)  $A$  is symmetric, antisymmetric, essentially selfadjoint, essentially antiselfadjont, selfadjoint, antiselfadjoint, isometric, unitary, idempotent in $\sH_\bR$ if   it is respectively such in $\sH$. 

 (v) $A$ is closable in $\sH_\bR$ if it is such in $\sH$. In this case the closures defined in $\sH$ and $\sH_\bR$, respectively coincide,
 
 (vi) $A$ is positive on $\sH_\bR$ if it is such on $\sH$.
\end{itemize}

 \item[(f)] Let $A : D(A)\rightarrow \sH_\bR$ be  
 a $\bR$-linear operator   with $D(A) \subset \sH_\bR$ a (real) subspace and suppose  that $\cJ A\subset A\cJ$ and $\cK A\subset A\cK$.
 Then
 
 (i)  $D(A)$ is a quaternionic-linear subspace of $\sH$ which is dense if it is dense in $\sH_\bR$ and $A$ is a quaternionic linear operator.
 
 (ii)  the adjoint of $A$ with respect to $\sH_\bR$ and with respect to $\sH$ coincide,
 
 (iii)  $A$  is symmetric, antisymmetric, essentially selfadjoint, essentially antiselfadjont, selfadjoint, antiselfadjoint, isometric, unitary, idempotent in $\sH$ if   it is respecively such in $\sH_\bR$,
 
 (iv) $A$ is closable in $\sH$ if it is such in $\sH_\bR$. In this case the closure defined on $\sH$ and $\sH_\bR$, respectively coincide,
 
 (v) $A$ is positive in $\sH$ if it is positive and symmetric in $\sH_\bR$.

\end{proposition}

\begin{proof}All statements immediately follow from the corresponding definitions. \end{proof}
\noindent We are now in a position to focus on the properties of the commutant of sets of quaternionic linear operators in $\gB(\sH)$,
 viewed as a unital {\em real} $C^*$-algebra, with the final aim  to prove 
the quaternionic version of the {\em double commutant theorem}.  \\
As in the real and complex Hilbert space cases, if $\gS \subset \gB(\sH)$ with $\sH$ a quaternionic Hilbert space, that the {\bf commutant} $\gS'$ of $\gS$ is
$$\gS' := \{A \in \gB(\sH) \:|\: AS=SA \quad \forall S \in \gS\}\:.$$
 If $\gS\subset\gT\subset\gB(\sH)$, it is easy to see that $\gT'\subset\gS'$ and $\gS\subset\gS''$. 
This immediately leads to $\gS'''=\gS'$. Moreover, if $\gS$ is closed under the $^*$-operation, then its commutant $\gS'$ turns out to be a unital $^*$-subalgebra of the (real)
 unital $C^*$-algebra $\gB(\sH)$. 
All that is identical to the known results in real and complex Hilbert spaces.
 In view of the definition of $\sH_\bR$, however, 
$\gS \subset \gB(\sH)$ can also be interpreted as a subset of $\gB(\sH_\bR)$ and thus a corresponding notion of commutant $\gS'^{\null_\bR}\subset\gB(\sH_\bR)$, 
$$\gS'^{\null_\bR} := \{A \in \gB(\sH_\bR) \:|\: AS=SA \quad \forall S \in \gS\}$$
can be defined.
  Notice that $\cJ,\cK,\cJ\cK\in\gS'^{\null_\bR}$ whatever $\gS\subset \gB(\sH)$ we choose, it being made of quaternionic-linear 
  operators due to (e) in Prop.\ref{propHR}. An  elementary result 
  relating  the two notions of commutant  is the following lemma.
  \begin{lemma}\label{lemmaM}
For every $\gS \subset \gB(\sH)$, it holds $$\gS' \subset \gS'^{\null_\bR}\subset \gS'+\cJ\gS'+\cK\gS'+\cJ\cK\gS'\quad (\subset \gB(\sH_\bR) )\:.$$
\end{lemma}
\begin{proof} The first inclusion is trivial since  all $\bH$-linear operators commuting with the elements of $\gS$ are $\bR$-linear.
Let us pass to the second inclusion. For $A\in \gS'^{\null_\bR}$ define the operator 
$B:=A-\cJ A\cJ-\cK A\cK-\cJ\cK A\cJ\cK$. 
This clearly belongs to $\gS'^{\null_\bR}$. By direct inspection one immediately sees that $B$ also commutes with both $\cJ$ and $\cK$ and thus 
it is quaternionic linear ((f) Prop.\ref{propHR}) and belongs to $\gS'$.
Similarly, if defining $B_\cJ:=A-\cJ A\cJ+\cK A\cK+\cJ\cK A\cJ\cK$, then 
$\cJ B_\cJ$ turns out to belong to $\gS'$.
We can repeat a similar  argument twice eventually  establishing that
\begin{equation*}
\begin{cases}
B&:=A-\cJ A\cJ-\cK A\cK-\cJ\cK A\cJ\cK\ \in \gS'\\
B_\cJ&:=A-\cJ A\cJ+\cK A\cK+\cJ\cK A\cJ\cK\ \in \cJ\gS'\\
B_\cK&:=A-\cK A\cK+\cJ A\cJ+\cJ\cK A\cJ\cK\ \in \cK\gS'\\
B_{\cJ\cK}&:=A-\cJ \cK A \cJ \cK+\cJ A \cJ+\cK A \cK \in \cJ \cK\gS'
\end{cases}
\end{equation*}
This concludes the proof because $4A=B+B_\cJ+B_\cK+B_{\cJ\cK}\in \gS'+\cJ\gS'+\cK\gS'+\cJ\cK\gS'$.
\end{proof}
  
 \noindent  We can now state and prove the quaternionic version of the double commutant theorem. 
The proof of this key tool is different to the known one of real and complex $^*$-algebras \cite{KR} in view of the 
non-commutativity of quaternions. 
A crucial intermediate result in the proof of Theorem 5.31 \cite{KR} to get $\gR'' \subset \overline{\gR}^s$ uses the fact that $\sH_\psi :=\overline{\{A\psi \:|\: A \in \gR\}}$
for $\psi \in \sH$ is a closed complex subspace of $\sH$.  If $\sH$ is quaternionic the corresponding property is false since $\sH_\psi$ is only  a {\em real subspace} instead of {\em quaternionic} in view of non-commutativity of quaternions.
This is the reason why we followed an alternative and a bit elaborate route.\\
  To do it, observe that a unital sub $^*$-algebra $\gR\subset\gB(\sH)$ is also a unital sub $^*$-algebra $\gR\subset\gB(\sH_\bR)$.

\begin{theorem}[Double commutant theorem]\label{DCT}
Let $\sH$ be a quaternionic Hilbert space and $\gR\subset\gB(\sH)$ a unital sub $^*$-algebra, then $$\gR''=\overline{\gR}^s=\overline{\gR}^w = \gR'{\null^{_\bR}}'{\null^{_\bR}}\:,$$
where the strong and weak closure above can indifferently be taken in $\gB(\sH)$ or $\gB(\sH_\bR)$.
\end{theorem}

\begin{proof} The idea of the proof is to reduce to the analogous result for algebras of operators in real Hilbert spaces (e.g., Thm 2.26 in \cite{MO1}).
As a first step, we prove that the weak and strong closures of a set 
$\gS \subset \gB(\sH) \subset \gB(\sH_\bR)$ do not depend on the choice of either  the real or the quaternionic Hilbert space structure. Let us prove it for the strong closures first.
If an operator $A\in\gB(\sH_\bR)$ is a strong-limit of elements of $\gS$ then it commutes with $\cJ,\cK$ because every element of the sequence does and so it belongs also to $\gB(\sH)$. 
To conclude, since the norms of $\sH$ and $\sH_\bR$ coincide, we easily see that $A$ is also a strong limit of elements of $\gS$ within $\gB(\sH)$. The opposite inclusion is similar, just remember that $\gB(\sH)\subset\gB(\sH_\bR)$. 
Let us pass to the weak closures. Take $\{A_n\}_{n\in\bN}\subset\gS$ and first suppose that it weakly converges to some $A\in \gB(\sH)$. This means that
 $\langle x|A_n y\rangle\to \langle x|Ay\rangle$ for every $x,y\in\sH$ and thus 
$(x|A_ny) = Re(\langle x|A_n y\rangle) \to Re(\langle x|A y\rangle) = (x|Ay)$. $\bR$-linearity of $A$ finally implies that $A_n$ weakly converges to $A\in \gB(\sH_\bR)$, too. 
Suppose conversely that $A_n$ weakly converges to  $A\in \gB(\sH_\bR)$, that is  $(x|A_n y)\to (x|A y)$ for every $x,y\in\sH_\bR$. $A$ is quaternionic-linear as every $A_n$ is. Indeed,
\begin{equation*}
\begin{split}
(x|\cJ Ay)&=(\cJ^*x|Ay)=\lim_{n\to\infty}(\cJ^*x|A_ny)=\lim_{n\to\infty}(x|\cJ A_ny)=\\
&=\lim_{n\to\infty}(x|A_n\cJ y)=(x|A\cJ y)
\end{split}
\end{equation*}
Arbitrainess of $x,y$ yields $\cJ A=A\cJ$. The same arguments works for $\cK$ and thus  $A$ is $\bH$-linear in view of (f) Prop.\ref{propHR}. Finally $A_n$ converges weakly
 to $A$ also on $\sH$ because (\ref{recsp}) gives
\begin{equation*}
\begin{split}
\langle x|A_n y\rangle &= (x|A_n y)-(x|\cJ A_ny)j-(x|\cK A_n y)k-(x|\cJ\cK A_n y)jk=\\
&= (x|A_n y)-(x| A_n\cJ y)j-(x|A_n \cK y)k-(x|A_n \cJ\cK y)jk\\
& \to (x|A y)-(x| A\cJ y)j-(x|A \cK y)k-(x|A \cJ\cK y)jk = \langle x|A y\rangle
\end{split}
\end{equation*}
We have in particular established that the weak and strong closures of a unital $\null^*$-algebra $\gR\subset\gB(\sH)\subset\gB(\sH_\bR)$ do not depend on the real or quaternionic Hilbert space structure. The double commutant theorem
  for real unital $^*$-algebras  guarantees that \beq\label{idcent} (\gR'^{\null_\bR})'^{\null_\bR}=\overline{\gR}^s=\overline{\gR}^w\eeq  where now the closures can 
  indifferently be intepreted in $\gB(\sH)$ or in $\gB(\sH_\bR)$. \\
To conclude the proof,  take a $\bH$-linear operator $A\in (\gR')'$. $A$ is also  $\bR$-linear and thus  $A \in (\gR')'^{\null_\bR}$ (where $\gR'$ is first defined within
 $\gB(\sH)$ and next is viewed as a subset of $\gB(\sH_\bR)$).
So $A$ is a $\bR$-linear operator which  commutes with $\cJ,\cK$ and the elements of $\gR'$. The second inclusion in the statement of  Lemma \ref{lemmaM}, implies that 
$A$ also commutes with the elements of $\gR'^{\null_\bR}$. Summing up,
$A  \in (\gR'^{\null_\bR})'^{\null_\bR}$.
 This proves  that $(\gR')'\subset (\gR'^{\null_\bR})'^{\null_\bR} = \overline{\gR}^s$, the identity arising from (\ref{idcent}). Since $\gR\subset (\gR')'$, we also have that 
$\gR \subset (\gR')'\subset (\gR'^{\null_\bR})'^{\null_\bR} = \overline{\gR}^s $. Taking the strong closure of every space and noticing that $\gS'$ and $\gS'^{\null_\bR}$ 
are always strongly closed, we have
$\overline{\gR}^s \subset (\gR')'\subset (\gR'^{\null_\bR})'^{\null_\bR} = \overline{\gR}^s$ so that
$\overline{\gR}^s =(\gR')'= (\gR'^{\null_\bR})'^{\null_\bR}$ and (\ref{idcent}) concludes the proof of 
$\gR''=\overline{\gR}^s=\overline{\gR}^w = \gR'{\null^{_\bR}}'{\null^{_\bR}}$. 
\end{proof}
\begin{corollary}\label{cor1}
Let $\sH$ be a quaternionic Hilbert space and $\gR\subset\gB(\sH)$ a unital sub $^*$-algebra. The following statements are equivalent:
\begin{enumerate}[(a)]
\item $\gR=\gR''$,
\item $\gR$ is weakly closed,
\item $\gR$ is strongly closed.
\end{enumerate}
\end{corollary}

\begin{proof} (a) implies (b) from Thm.\ref{DCT}, (b) implies (c) trivially. If (c) holds true then $\gR = \overline{\gR}^s = \gR''$
 from Thm.\ref{DCT} again.
\end{proof}
\noindent This result suggests to extend  the usual definition of von Neumann algebra to the quaternionic case.
\begin{definition}\label{defVN} If $\sH$ is a quaternionic Hilbert space, 
a {\bf von Neumann algebra} in $\gB(\sH)$ is a unital $^*$-subalgebra of $\gB(\sH)$ satisfying the three equivalent statements of Corollary \ref{cor1}. 
\end{definition}

\begin{remark}{\em $\null$\\
{\bf (a)}  In view of Thm.\ref{DCT}, if $\sH$ is a quaternionic Hilbert space and $\gR\subset\gB(\sH)$ a unital $^*$-algebra, then $\gR$ is a von Neumann 
algebra on $\sH$ if and only if it is a von Neumann algebra on $\sH_\bR$.\\
{\bf (b)} Exactly as in the real Hilbert space case, in quaternionic Hilbert spaces, a von Neumann algebra is a unital sub $C^*$-algebra of $\gB(\sH)$ (and $\gB(\sH_\bR)$) 
since strong closure implies uniform closure and $||A^*A||=||A||^2$ for $A\in \gB(\sH)$ is valid also in quaternionic Hilbert spaces \cite{GMP1}.}
\end{remark}
\noindent If  $\gS\subset\gB(\sH)$ is $^*$-closed, $\gS'$ is a unital $^*$-algebra which is von Neumann  because $\gS'= (\gS')''$, in particular, $\gS''= (\gS')'$ is 
also a von Neumann algebra. If $\gT \supset \gS$
is another von Neumann algebra, we have  $\gT' \subset \gS'$ and so $\gT'' \supset \gS''$. Therefore $\gS''$ is the smallest von Neumann algebra including $\gS$.
\begin{definition}
Let $\gS\subset\gB(\sH)$ be closed under the $^*$-operation, then  $\gS''$ is called the {\bf von Neumann generated} by $\gS$.
\end{definition}

\subsection{The lattice of orthogonal projectors of quaternionic von Neumann algebras}
  An important structure associated with a von Neumann algebra  $\gR$, in a quaternionic Hilbert space $\sH$ in particular,  is the lattice of orthogonal projectors 
  in $\gR$, denoted by $\cL_\gR(\sH)$ (where  $\cL(\sH):= \cL_{\gB(\sH)}(\sH)$). As in the real case \cite{MO1} and differently from the complex case it does not contain 
  all the information about the algebra $\gR$. Indeed differently from the case of a complex Hilbert space,  it may  hold $\cL_\gR(\sH)''\subsetneq\gR$, as the following elementary example shows.
\begin{example}
{\em
Let $\sH$ be any quaternionic  Hilbert space and $J\in\gB(\sH)$ such that $J^*=-J$ and $JJ=-I$. Consider the unital sub $^*$-algebra $\gR:=\{aI+bJ\:|\:a,b\in\bR\}\subset\gB(\sH)$. 
If we manage to prove that $\gR$ is weakly closed, Corollary \ref{cor1} guarantees that it is a von Neumann algebra. 
So, suppose there exists some $A\in\gB(\sH)$ such that $a_nI+b_nJ\rightarrow A$ weakly. It is clear that $a_nI-b_nJ=(a_nI+b_nJ)^*$ weakly converges to $A^*$.  So that 
 $a_n\rightarrow a$ and $b_n\rightarrow b\in\bR$ 
for some $a, b\in\bR$. This gives $A=aI+bJ$ so that $\gR$ is weakly closed and therefore is a von Neumann algebra.
On the other hand,  $\cL_\gR(\sH)=\{0,I\}$, since every self-adjoint element of $\gR$ is necessarily of the form $aI$ for some $a\in\bR$. The smallest algebra containing
 it is obviously $\cL_\gR(\sH)''=\{aI\:|\:a\in\bR\}\subsetneq\gR$.
}
\end{example}
\noindent To investigate the elementary properties of $\cL_\gR(\sH)$ we need some further preliminary technical results.  Similarly to Prop.\ref{propHR}, given a
 quaternionic Hilbert space $\sH$, and using the real Hilbert space $\sH_\bR$,  we can define an associated  complex Hilbert space $\sH_\bC$ depending on  the choice of a preferred 
imaginary unit. As usual $\bC_j = \{a+jb \:|\: a,b \in \bR\}$ is a quaternionic realization of the field of complex numbers  relying on the choice of the imaginary unit 
$j \in \bH$.

 \begin{proposition}\label{propHC}
Let  $(\sH,\langle\cdot|\cdot\rangle)$ be a quaternionic Hilbert space, fix the imaginary unit $j \in \bH$ and a corresponding  $\bR$-linear operator $\cJ$ in $\sH_\bR$
as in (b) of Prop.\ref{propHR} and consider the associated complex  vector space on the set of vectors
$\sH_{\bC_j}:=\sH$, whose complex linear combinations are the quaternionic ones with  coefficients in $\bC_j$,  and define the Hermitian scalar 
product $$(x |y)_j:=(x|y)-(x|\cJ y)j \quad \forall x,y \in \sH\:,$$
where $x,y$ are viewed as elements of $\sH_\bR$ and the corresponding real scalar product takes place in the right-hand side. 
The following facts are true.
\begin{itemize}
\item[(a)]  $(x|y) = Re (x|y)_j$ for every $x,y \in \sH$  viewed as elements of $\sH_\bR$.

\item[(b)] $\langle x|y\rangle = (x|y)_j-k(x|Ky)_j$ for every $x,y\in\sH$

\item[(c)] The norms of $\sH$, $\sH_\bR$ and $\sH_{\bC_j}$ coincide and $(\sH_{\bC_j},(\cdot|\cdot)_j)$ is a complex Hilbert space.  ($\sH_{\bC_j} = (\sH_\bR)_\cJ$
where the right-hand side is the internal complexification of $\sH_\bR$ referred to the complex structure $\cJ$ as in Sect.2.5  in \cite{MO1}).

\item[(d)] Every quaternionic-linear operator is  complex-linear and in particular   $\gB(\sH)\subset \gB(\sH_{\bC_j})$ and the inclusion of algebras
  preserves the norms of the operators.
 
 \item[(e)] $A\in\gB(\sH_{\bC_j})$ belongs to $\gB(\sH)$ if and only if it commutes with $\cK$. 

\item[(f)] Let $A : D(A)\rightarrow \sH$ be a quaternionic-linear operator  with $D(A) \subset \sH$ a (quaternionic) subspace and consider 
$A$ as a complex-linear operator on $\sH_{\bC_j}$ with domain $D(A)$ viewed as complex subspace of $\sH_{\bC_j}$.
Then

(i)   $\cK A\subset A \cK$,

(ii)  $D(A)$ is a dense real-linear subspace of $\sH_{\bC_j}$ if it is dense in  $\sH$,
 
 (iii) the adjoint of $A$ with respect to $\sH$ and with respect to $\sH_{\bC_j}$ coincide,
 
 (iv)  $A$  is symmetric, antisymmetric, essentially selfadjoint, essentially antiselfadjont, selfadjoint, antiselfadjoint, isometric, unitary, idempotent in $\sH_{\bC_j}$ if   it is respectively such in $\sH$,
 
 (v) $A$ is closable in $\sH_{\bC_j}$ if it is such in $\sH$. In this case the closures defined in $\sH$ and $\sH_{\bC_j}$, respectively coincide
 
 (vi) $A$ is positive in $\sH_{\bC_j}$ if it is such in $\sH$

 \item[(g)] Let $A : D(A)\rightarrow \sH_{\bC_j}$ be  
 a $\bC$-linear operator   with $D(A) \subset \sH_{\bC_j}$ a (complex) subspace and suppose  that $\cK A\subset A\cK$.
 Then
 
 (i)  $D(A)$ is a quaternionic-linear subspace of $\sH$ which is dense if it is dense in $\sH_{\bC_j}$ and $A$ is a quaternionic linear operator.
 
 (ii)  the adjoint of $A$ with respect to $\sH_{\bC_j}$ and with respect to $\sH$ coincide,
 
 (iii)  $A$  is symmetric, antisymmetric, essentially selfadjoint, essentially antiselfadjont, selfadjoint, antiselfadjoint, isometric, unitary, idempotent in $\sH$ if   it is respecively such in $\sH_{\bC_j}$,
 
 (iv) $A$ is closable on $\sH$ if it is such in $\sH_{\bC_j}$. In this case the closures defined in $\sH$ and $\sH_{\bC_j}$, respectively coincide
 
 (v) $A$ is positive in $\sH$ if it is positive and symmetric on $\sH_{\bC_j}$
\end{itemize}
\end{proposition}

\begin{proof}All statements immediately follows from the corresponding definitions. \end{proof}
\noindent Let $(\sK, (\cdot|\cdot))$ be a real, complex or quaternionic Hilbert space. Consider a, generally unbounded, selfadjoint operator $A$. Thanks to the 
Spectral Theorem (see, e.g., \cite{MO1} for the real case and  \cite{GMP2} for the quaternionic case\footnote{Concerning self-adjoint operators, we use the notation $dP$
 instead of $d\cP$ preferred in  \cite{GMP2} in the integral spectral formula.}) it can be spectrally decomposed as
\begin{equation}\label{st}
A=\int_{\bR}\lambda\, dP,\ \ \ D(A)=\left\{x\in\sH\:\bigg|\: \int_\bR|\lambda|^2\, d\mu^P_x<\infty\right\}
\end{equation}
where $P:\mathcal{B}(\bR)\rightarrow\cL(\sH)$ is the projection-valued measure (PVM) uniquely associated to $A$ and $\mu^P_x(E)=(x|P_E x)$ is the positive valued
 finite measure associated with $P$. The spectral theorem for self-adjoint operators can be restated as follows.
\begin{lemma}\label{lemmast}
Let $(\sK, (\cdot|\cdot))$ be a real, complex or quaternionic Hilbert space and  consider a, generally unbounded, selfadjoint operator $A: D(A) \to \sK$ with $D(A)$ 
dense in $\sK$. There exists a unique PVM $P$ on $\bR$ such that, if $\mu^P_x(E):=(x|P_E x)$, then  $id\in L^2(\bR,\mu^P_x)$ and
$
(x|Ax) =\int_\bR \lambda\, d\mu_x^P
$
for all $x\in D(A)$ 
\begin{proof} The PVM $P$ associated to $A$ by the spectral theorem satisfies the thesis.
Suppose $Q$ is another PVM satisfying the above requirements and take $x\in D(A)$. In particular this means that $x\in D\left(A'\right)$ where $A'$ is the
 selfadjoint operator $\int_\bR\lambda\, dQ$ and $(x|Ax)=(x|A'x)$, i.e. $(x|(A-A')x) =0$. Since $A-A'$ is symmetric on the dense domain $D(A)$, a straightforward 
 calculation shows that $A-A'=0$ on $D(A)$, i.e. $A\subset A'$. Since $A$ is selfadjoint it does not admit proper selfadjoint extensions (see \cite{GMP1,GMP2} for 
 the quaternionic case and \cite{MO1} for the real case), that is $A=A'$. Finally, the spectral Theorem guarantees that $Q=P$.
\end{proof}
\end{lemma}
Once the PVM associated with $A$ is available, if $f:\bR\rightarrow\bR$ is Borel measurable, the following operator can be defined
\begin{equation}\label{stf}
f(A):=\int_\bR f(\lambda)\, dP,\ \ \ D(f(A)):=\left\{x\in\sH\:\bigg|\: \int_\bR|\lambda|^2\, d\mu_x^P <\infty\right\} 
\end{equation}
In particular it satisfies $(x|f(A)x)=\int_\bR f(\lambda)\, d\mu_x^P$ for all $x\in D(f(A))$.
\vspace{0.1cm}

\noindent Another  technical lemma is in order.
\begin{lemma}\label{commst}
Let $A$ be a, generally unbounded, self-adjoint operator on a quaternionic Hilbert space $\sH$ with $PVM$ $P^{(A)}$ and $U\in\gB(\sH)$. The following facts hold:
\begin{enumerate}[(a)]
\item $P^{(A)}$ equals the PVMs associated with $A$ on $\sH_\bR$ and $\sH_{\bC_j}$ via spectral theorem.
\item $UA\subset AU$ if and only if $UP^{(A)}_E=P^{(A)}_EU$ for every $E\in\cB(\bR)$.
\item  If $f:\bR\rightarrow \bR$ is Borel measurable, then the definitions of $f(A)$ referred to $\sH,\sH_\bR$ or $\sH_{\bC_j}$ coincide.
\end{enumerate}
\end{lemma}
\begin{proof}
Consider the unique real PVM $Q$ associated with $A$ on the real Hilbert space $(\sH_\bR,(\cdot|\cdot))$, while $P$  denotes the quaternionic analogue  on $\sH$. 
 Take $u\in D(A)$. Since $A$ selfadjoint, it holds $\langle u|Au\rangle\in\bR$, hence
\begin{equation*}
\int_\bR\lambda\, d\mu^{Q}_u=(u|Au)=\langle u|Au\rangle =\int_\bR\lambda\, d\mu^{P}_u
\end{equation*}
It is easy to see that $P$ is still a PVM if understood on $\sH_\bR$. Moreover $\mu_u^P(E)=\langle u| P_E u\rangle=(u|P_E u)$, again from  
selfadjointness of $P_E$. Lemma \ref{lemmast} for the real Hilbert space  case proves that $P=Q$. 
We have established  (a) for $\sH_\bR$. Next (b) follows immediately, since the statement is a well-known property in real Hilbert spaces \cite{MO1}. 
Analogously, property (c) for $\sH_\bR$ follows easily from (a)
using the fact that the  integral operator in (\ref{st}) can be computed as a strong limit (even uniform if $A$ is bounded) of bounded operators of the form
$
\sum_{l=1}^{n}c_lP({F}_l)
$
where $c_l\in\bR$ and $F_l\in\mathcal{B}(\bR)$ depend only on the chosen function $f$ and not on the scalar field of the Hilbert space.
Finally, the proofs for the case of $\sH_{\bC_j}$ is immediate by noticing that $\sH_{\bC_j}=(\sH_\bR)_\cJ$ and using the theory of \cite{MO1}.
\end{proof}

\noindent We are eventually in a position to prove an elementary though relevant proposition concerning the interplay of $\gR$ and $\cL_\gR(\sH)$. This result extends Thm 2.29  we established in \cite{MO1}  to quaternionic Hilbert spaces.

\begin{proposition}\label{2.12}
Let $\gR$ a von Neumann algebra, define the set
$$\cJ_\gR:=\{J\in\gR\:|\: J^*=-J,\ -J^2\in \cL_\gR(\sH)\}.$$ The following facts are valid.\\
{\bf (a)} $A^*=A\in\gR$ iff the orthogonal projectors of its PVM belong to $\cL_\gR(\sH)$.\\
{\bf (b)} $\cL_\gR(\sH)$ is a complete orthomodular bounded sublattice of $\cL(\sH)$.\\
{\bf (c)} $\cL_\gR(\sH)''$ contains all the selfadjoint operators of $\gR$.\\
{\bf (d)} $\cL_\gR(\sH)''+\cJ_\gR\cL_\gR(\sH)''=\gR$.\\
%%%%%%  CORREZIONE mancava '' in \cL_\gR(\sH)'' sotto 
{\bf (e)} $\cL_\gR(\sH)''\subsetneq\gR$ iff there exists $J\in \cJ_\gR\setminus\cL_\gR(\sH)''$.
\end{proposition}
\begin{proof}
(a) and (c). Suppose $A^*=A \in \gR$ and $B\in\gR'$. In particular $B$ commutes with $A$ and Lemma \ref{commst} guarantees that $B$ commutes 
with the PVM $P^{(A)}$ of $A$. Arbitrariness of  $B$ implies  $P^{(A)}_E\in\gR''=\gR$ for every Borel set $E$. Suppose conversely that $A^*=A\in\gB(\sH)$ is 
 such that $P^{(A)}_E \in \cL_\gR(\sH)$ for every Borel set $E$, we want to prove that $A \in \gR$. If $B\in\mathcal{L}_\gR(\sH)'$, in particular $B$ commutes 
 with the PVM of $A$ and so, thanks again to Lemma \ref{commst}, it commutes also with $A$ and  $A\in\cL_\gR(\sH)''\subset \gR$. This concludes the proof of (a).  
 The used argument proves also (c). Indeed, if $A^*=A\in\gR$, then, thanks to the first implication of (a) its PVM belongs to $\cL_\gR(\sH)$ and so the argument above applies. \\
(b) The properties of $\cL_\gR(\sH)$ listed in (b) are inherited from the same properties of $\cL(\sH)$. The proof can be obtained by mimicking the
 one developed in \cite{Redei} for real and complex von Neumann algebras.\\
(d) First suppose $A\in\gR$ is anti selfadjoint. The polar decomposition theorem for bounded operators in quaternionic Hilbert spaces (see \cite{GMP1}) 
guarantees that $A=W|A|$ where $W$ is an antiselfadjoint partial isometry and $|A|:=\sqrt{A^*A}$. Moreover $W$ and $|A|$ commute with each other 
and with every operator commuting with $A$. This guarantees that $W,|A|\in\gR''=\gR$. As $|A|$ is self-adjoint, $|A|\in\cL_\gR(\sH)''$ for (c). Since $W$ 
is a partial isometry, $W^2$ is an orthogonal projector which clearly belongs to $\gR$, hence $W\in \cJ_\gR$.
To conclude the proof of (d), observe that a generic operator $A\in\gR$ can always be decomposed as $\frac{1}{2}(A+A^*)+\frac{1}{2}W_0|A-A^*|$ 
where $(W_0,|A-A^*|)$ is the polar decomposition of $A-A^*$.  (c) and the previous discussion prove that $A+A^*,|A-A^*|\in\cL_\gR(\sH)''$ and $W_0\in\cJ_\gR$ 
concluding the proof of (d).\\
(e) If there exits $J\in\cJ_\gR\setminus\cL_\gR(\sH)''$ then $\cL_\gR(\sH)''\subsetneq\gR$ evidently. So, suppose  $A\in\gR\setminus\cL_\gR(\sH)''$, 
then  $A-A^*\in\gR\setminus\cL_\gR(\sH)''$, otherwise $A\in\cL_\gR(\sH)''$ since  $A+A^*$ is self-adjoint and thus belongs to $\cL_\gR(\sH)''$ thanks
 to (c). Referring to the polar decomposition $A-A^*=W_0|A-A^*|$, the above discussion guarantees that $W_0\in\cJ_\gR$. Since $A-A^*\not\in\cL_\gR(\sH)''$ 
 and $|A-A^*|\in\cL_\gR(\sH)''$, $W_0$ cannot belong to $\cL_\gR(\sH)''$. 
\end{proof}

\subsection{Stone theorem}\label{SECstone}
We need to extend to quaternionic Hilbert spaces the crucial technical result known as Stone Theorem. 
Consider an antiselfadjoint operator $A$ on $\sH$. As $A$ is closed and normal \cite{GMP2}, thanks to the spectral theorem for closed 
normal (generally unbounded) operators in quanternionic Hilbert spaces Theorem 6.6 in  \cite{GMP2}, fixed an imaginary unit $i$, there 
exists a unique PVM $P$ on $\bC^+_i$ and a (non-unique) left-multiplication $\cL$ commuting with $P$, encapsulated in $\cP := (\cP, \cL)$, such that 
\begin{equation}
A=\int_{\bC^+_i} id \, d\cP(z),\ \ \ D(A):=\left\{x\in\sH\:\bigg|\:\int_{\bC_i}|z|^2\, d\mu_x^{(A)}<\infty\right\}\:.
\end{equation}
Notice that $\mbox{supp} P\subset i\bR^+ \cup \{0\}$ since $A=-A^*$ (from (b) Thm 4.1 in \cite{GMP2} and (c) Thm 4.8 in \cite{GMP1}). If $f:\bC^+_i\rightarrow \bC_i$ is Borel measurable we can define the operator
\begin{equation}\label{functionOP}
f(A):=\int_{\bC^+_i} f(z)\, d\cP(z),\ \ \ D(f(A)):=\left\{x\in\sH \bigg|\int_{\bC^+_i}|f(z)|^2\, d\mu_x^{(A)}<\infty\right\}
\end{equation}
In particular, consider the operator 
\begin{equation}
e^{tA}:=\int_{\bC^+_i} e^{tz}\, d\cP(z)\:.
\end{equation}
The map $\bR\ni t\mapsto e^{tA}$ is a one-parameter group of unitary operators as straightforward consequence of Thm 3.13 in \cite{GMP2}. 
Following essentially the same proof as in the complex Hilbert space case, one easily establishes that
that map is strongly continuous and
\begin{equation}\label{stoneproperty}
D(A)=\left\{ x\in\sH\:\bigg|\: \exists\ \dfrac{d}{dt}\big|_0 U_t x \in \sH \right\}\:, \quad Ax=\left.\dfrac{d}{dt}\right|_{t=0} U_t x\:.
\end{equation}
Exactly as in the complex and real cases, these one-parameter groups are the \textit{only} strongly-continuous one-parameter unitary 
groups on $\sH$. It is the content of the well-known Stone's theorem now extended to the quaternionic Hilbert space case. 
\begin{theorem}[Stone's Theorem]\label{stonetheorem}
Let $\sH$ be a quaternionic Hilbert space and $\bR\ni t\mapsto U_t\in\gB(\sH)$ be a strongly-continuous one-parameter unitary group, 
then there exists a unique anti-selfadjoint operator $A$ called the {\bf (anti-selfadjoint) generator of} $U$ such that 
$$
U_t=e^{tA} \quad \mbox{for all $t \in \bR$}.
$$
$A$ turns out to be defined as in (\ref{stoneproperty}) where $D(A)$ is dense and invariant under the action of $U$.\\
Finally $A$ coincides with the unique anti-selfadjoint operator in $\sH_\bR$  $(\sH_{\bC_j})$ that generates $U$ when understood as 
strongly-continuous unitary group of operators in $\sH_\bR$ $(\mbox{resp. }\sH_{\bC_j})$.
\end{theorem}
\begin{proof} See Appendix \ref{AppProof}. \end{proof}

\subsection{Polar decomposition of unbounded operators in quaternionic Hilbert spaces}
Another technical tool is the so-called polar decomposition of closed generally unbounded operators. We state the version in quaternionic Hilbert spaces together with some direct application. 
\begin{theorem}[Polar decomposition theorem]\label{PDT}
Let $\sH$ be a quaternionic Hilbert space and  $A:D(A)\rightarrow \sH$ a $\bH$-linear closed operator with $D(A)$ dense $\sH$. The following facts hold.
\begin{enumerate}[(a)]
\item $A^*A$ is densely defined, positive and selfadjoint.
\item There exists a unique pair  of operators $U,P$ in $\sH$ such that,
\begin{enumerate}[(i)]
\item $A=UP$ where in particular $D(P)=D(A)$ 
\item $P$ is selfadjoint and $P\geq 0$
\item $U\in\gB(\sH)$ is isometric on $Ran(P)$  (and thus on $\overline{Ran(P)}$ by continuity),
\item $Ker(U)\supset Ker(P)$.
\end{enumerate}
The right-hand side of (i)  is called the {\bf polar decomposition} of $A$.  It turns out that, in particular,
\begin{enumerate}[(1)]
\item $P=|A|:=\sqrt{A^*A}$ (defined by means of (\ref{stf}) with $\sqrt{z}:=0$ if $z \not \in [0,+\infty)$.)
\item $Ker(U)= Ker(A) = Ker(P)$,
\item  $Ran(U) = \overline{Ran(U)}$
\end{enumerate} 
\item The polar decompositions of $A$ carried out on $\sH,\sH_\bR,\sH_{\bC_j}$ are made of the same pair of operators $U$, $P$ arising in (i)(a) above.
\end{enumerate}
\end{theorem}
\begin{proof} See Appendix \ref{AppProof}. \end{proof}

\noindent A pair of related results, in part involving one-parameter groups of unitaries, are stated below.

\begin{proposition}\label{LemmaCOMM}
Let $\sH$ be quaternionic Hilbert space. Consider an, either  selfadjoint or anti selfadjoint,  operator $A:D(A)\rightarrow \sH$ with polar decomposition  $A=UP$. The following facts hold.
\begin{enumerate}[(a)]
\item If $A^*=-A$, then $B\in\gB(\sH)$ satisfies $Be^{tA}=e^{tA}B$ if and only if $BA\subset AB$
\item If $B\in\gB(\sH)$ satisfies  $BA\subset AB$, then $BU=UB$ and $BP \subset UP$.
\item The  commutation relations are true
 $$UA \subset  AU \quad \mbox{and}\quad U^*A\subset AU^*\:.$$
Moreover, for every measurable function $f: [0,+\infty) \to \bR$:
 $$Uf(P)\subset f(P)U\quad \mbox{and}\quad U^*f(P)\subset f(P)U^*\:.$$
\item $U$ is respectively selfadjoint or anti selfadjoint.
\item If $A$ is injective (equivalently if either $P$ or $U$ is injective), then $U$ and $U^*$ are unitary. In this case all the inclusions in (c) are identies.
\end{enumerate} 
\end{proposition}
\begin{proof}
The thesis easily follows by working in $\sH_\bR$ and exploiting the analogous result for real Hilbert spaces, Thm 2.19 in \cite{MO1}. Just notice that
 the definitions of $e^{tA}$ on $\sH$ and $\sH_\bR$ coincide thanks to Theorem \ref{stonetheorem}.
\end{proof}

\begin{proposition}\label{polarCOMM}
Let $\sH$ be a quaternionic Hilbert space and  $A$ and $B$ anti-selfadjoint
operators in $\sH$ with  polar decompositions $A=U|A|$ and $B=V|B|$.
If the strongly-continuous one-parameter groups generated by $A$ and $B$ commute, i.e.,
$$e^{tA}e^{sB}=e^{sB}e^{tA} \quad \mbox{ for every }s,t\in\bR$$ then the following facts hold
\begin{enumerate}[(a)]
\item $UB\subset BU$ and $U^*B\subset BU^*$;
\item $Uf(|B|)\subset f(|B|)U$ and $U^*f(|B|)\subset f(|B|)U^*$ for\quad $f: [0,+\infty) \to \bR$ every measurable function 
\item (iii) $UV=VU$ and $U^*V=VU^*$.
\end{enumerate}
If  any of $A$, $|A|$,  $U$ is injective, then  the inclusions in (i) and (ii) can be replaced by identities. 
\end{proposition}
\begin{proof}
Again the thesis easily arises by working in $\sH_\bR$ and using the analogous result for real Hilbert spaces, Thm 2.20 in  \cite{MO1}. Again notice 
that the definitions of $e^{tA}$ on $\sH$ and $\sH_\bR$ coincide thanks to Theorem \ref{stonetheorem}.
\end{proof}

\section{Restriction to subspaces induced by complex structures}
 To go on, we need to introduce some fundamental technical tools of quaternionic Hilbert spaces theory established in \cite{GMP1}.
A preliminary definition is necessary.

\begin{definition}\label{defCSJ} {\em If $\sH$ is a quaternionic  Hilbert space, an operator $S \in \gB(\sH)$ such that $S^2=-I$ and $S^*=-S$ is called {\bf complex structure} on $\sH$. }
\end{definition}
\begin{remark}
{\em Notice that the previously introduced operators $\cJ$ and $\cK$ are {\em not} complex structures because they are only $\bR$-linear
 while complex structures $S$ are required to be $\bH$-linear.}
\end{remark}

\subsection{Restriction to a complex Hilbert subspace induced by a complex structure}

\noindent If $\sH$ is a quaternionic Hilbert space with scalar product $\langle\cdot|\cdot\rangle$, $j,k$ any couple of anticommuting imaginary units and $J\in\gB(\sH)$ a complex structure,
define the subsets (which should be indicated by $\sH_+^{Jj}$ and $\sH_-^{Jj}$ according to notation of \cite{GMP1})
\begin{equation}\label{defHJ}
\sH_{J}:=\{u\in\sH\:|\:Ju=uj\}\:,\quad \sH^{(-)}_{J} :=\{u\in\sH\:|\:Ju=-uj\}
\end{equation}
$\sH_J$ and $\sH^{(-)}_J$ are evidently closed by (right) scalar multiplication with quaternions $a+bj \in \bC_j$ and thus are complex vector
 spaces. It is easily proved that the restriction of 
$\langle\cdot|\cdot\rangle$
 to $\sH_{J}$, resp., $\sH^{(-)}_{J}$ is a Hermitian complex scalar product. Since these sets are evidently closed because $J$ is unitary, $\sH_J$ 
 and $\sH_J^{(-)}$ equipped with the relevant restrictions of $\langle\cdot|\cdot\rangle$ are complex Hilbert spaces.  Since $jk=-kj$ the following identity holds 
\begin{equation}\label{HJk}
\sH^{(-)}_{J} =  \sH_{J}k :=\{vk\in\sH\:|\:v \in \sH_J\}\:.
\end{equation}
Due to (\ref{HJk}), we will refer to $\sH_J$ only in the rest of this paper.

\noindent Notice that the elements of $\sH_J$ can be interpreted as \textit{complex vectors} of $\sH$, as they commute with every element of $\bC_j$.

\noindent The following proposition collects some results that can be found in Prop.3.8, Lemmata 3.9-3.10 and the proof of (a) Prop. 3.11 of \cite{GMP1}. In particular 
\begin{proposition}\label{propHJ} Let $\sH$ be a quaternionic Hilbert space and $J\in\gB(\sH)$ a complex structure. The  complex vector 
space $\sH_J$ in Def.(\ref{defHJ}) equipped with the complex-linear structure $\bC_j$ and the restriction $\langle\cdot|\cdot\rangle_J$ of 
the scalar product of $\sH$ to $\sH_J$ satisfies the following properties.\\
{\bf (a)} $\sH_{J}$ is a  $\bC_j$-complex Hilbert space, non-trivial unless $\sH= \{0\}$.\\
{\bf (b)} The direct (generally not orthogonal) decomposition is valid $\sH=\sH_{J}\oplus \sH_{J}k$
such that
\begin{equation*}
\|x\|^2=\|x_1\|^2+\|x_2\|^2\quad \mbox{if $\sH \ni x=x_1+x_2k$ with $x_1,x_2 \in \sH_J$.}
\end{equation*}
{\bf (c)} The map $\sH_J \ni u\mapsto uk \in \sH_Jk$ is $\bC_j$-antilinear, isometric and bijective. \\
{\bf (d)} If $N\subset \sH_J$ is a Hilbert basis of $\sH_J$, then $N$ is also a Hilbert basis of the whole $\sH$.
\end{proposition}
\noindent The complex space $\sH_J$ has an important interplay with operators as established in Prop.3.11 in \cite{GMP1} we summarize below into a form adapted to our work.
The statements (c)-(iii), (c)-(ix) and (d)  do not appear in \cite{GMP1}, however their proofs are trivial. In particular the proof of (c)-(iii) is analogous to the one of (c)-(iv).
\begin{proposition}\label{propHJop}
With the same hypotheses as for Prop.\ref{propHJ} the following facts hold.\\

\noindent {\bf (a)} If $A:D(A)\to \sH$, with $D(A) \subset \sH$,
is $\bH$-linear and $JA \subset AJ$, then $A_J:=A|_{\sH_J\cap D(A)}$ is a well-defined $\bC_j$-linear operator on $\sH_J$. \\

\noindent {\bf (b)} If $X:D(X)\to \sH_J$, with $D(X) \subset \sH_J$, is $\bC_j$-linear, there exists a \textit{unique} $\bH$-linear operator 
$\tilde{X}:D(\widetilde{X})\rightarrow \sH$ such that $J\widetilde{X} \subset \widetilde{X}J$ and $\widetilde{X}_J=X$ where  in particular
 $D(X)=D(\widetilde{X})\cap \sH_J$. \\

\noindent {\bf (c)} If $X$ is as in (b) and $S:D(S)\to \sH_J$, with $D(S) \subset \sH_J$, is  another $\bC_j$-linear operator, the following facts are true on the relevant natural domains.

(i)  $\widetilde{X}J=J\widetilde{X}$.

(ii) $\widetilde{aX}= a \widetilde{X}$ for every $a \in \bR$.

 (iii) $\widetilde{S+T}=\widetilde{S}+\widetilde{T}$.

(iv)  $\widetilde{ST}=\widetilde{S}\widetilde{T}$.

(v)   $\widetilde{S} \subset \widetilde{T}$ iff $S\subset T$.

(vi) $X \in \gB(\sH_J)$ if and only if $\widetilde{X} \in \gB(\sH)$. In this case $\rVert\tilde{X}\rVert=\rVert X\rVert$.

(vii)  $D(\widetilde{X})$ is dense in $\sH$  if and only if  $D(X)$ is dense in $\sH_J$.

(viii) $(\widetilde{X})^*=\widetilde{X^*}$ if $X$ is densely defined.\\
As a consequence $X$ 
is symmetric, essentially selfadjoint,  selfadjoint, antiselfadjoint, isometric, unitary, idempotent if and only if  $\widetilde{X}$ is, respectively, 
symmetric, essentially selfadjoint, selfadjoint, antiselfadjoint, isometric, unitary, idempotent on $\sH$.

(ix)  $X$ is closable if and only if $\widetilde{X}$ is closable. In this case $\widetilde{\overline{X}}= \overline{\widetilde{X}}$.\\

\noindent {\bf (d)} If $\gB(\sH) \ni A_n \to A \in \gB(\sH)$ weakly, resp., strongly for $n \to +\infty$ and $JA_n=A_nJ$, then $JA=AJ$ and $\gB(\sH_J) \ni (A_n)_J \to A_J \in \gB(\sH_J)$
 weakly, resp., strongly for $n \to +\infty$.

\end{proposition}
\noindent The stated results extends into the following result which consider constructions related to spectral measures. 
\begin{proposition}\label{operatorfunctioncompl}
Let $\sH$ be a quaternionic Hilbert space, $J\in\gB(\sH)$ a complex structure and
 $A:D(A)\rightarrow \sH$, where $D(A)\subset \sH$, a selfadjoint $\bH$-linear operator such that  $JA\subset AJ$. Then the following facts hold referring to the notation in (b) Prop.\ref{propHJop}.\\

\noindent {\bf (a)} If $P^{(A)}$ is the PVM of the spectral decomposition of $A$, the operators 
$(P_E^{(A)})_J$ with $E\subset \bR$ Borel set, form  the PVM $P^{(A_J)}$ of  $A_J$ in $\sH_J$. Furthermore the map 
$$P^{(A)} \ni P^{(A)}_E \mapsto (P_E^{(A)})_J \in P^{(A_J)}$$
is an isomorphism of $\sigma$-complete Boolean lattices.
\\

\noindent {\bf (b)} If $f:\bR\rightarrow \bR$ is Borel measurable, then $f(A)J=Jf(A)$ and $f(A)_J=f(A_J)$.\\

\noindent {\bf (c)} $A\geq 0$ if and only if $A_J\geq 0$.\\

\noindent {\bf (d)} $\sigma_S(A)= \sigma(A_{J})$, $\sigma_{pS}(A)= \sigma_p(A_{J})$, $\sigma_{cS}(A)= \sigma_c(A_{J})$, where $\sigma_S$ denotes the {\em spherical spectrum} of the 
quaternionic linear operator $A$  \cite{GMP1,GMP2}.
\end{proposition}
\begin{proof} See Appendix \ref{AppProof}. \end{proof}

\noindent To conclude, let us examine the interplay of the von Neumann algebra structure in $\sH$
and that in $\sH_J$.
 
\begin{proposition}\label{complexidentif}
Let $\gR$ be a von Neumann algebra over the quaternionic Hilbert space $\sH$ and $J\in\gR'\cap \gR$ be a complex structure. Then $\gR_J:=\{A_J\:|\: A\in\gR\}$ 
is a  von Neumann algebra over the $\bC_j$-complex Hilbert space $\sH_J$. 
\end{proposition}
\begin{proof}
$\gR_J$ is a unital $^*$-subalgebra of $\gB(\sH_J)$, indeed $I\in \gR_J$ and if  $A_J,B_J\in\gR_J$, then  $A_J+B_J = (A+B)_J$ belongs to $\gR_J$. If 
furthermore $a+ jb \in \bC_j$, since $J \in \gR$, we have  $J_J \in \gR_J$ and  $(a+jb) A_J :=aA_J + bJ_JA_J = (aA+ bJA)_J$ is a well defined $\bC_j$-linear
  operator in $\gR_J$. Finally, if $A_J\in\gR_J$ then $(A_J)^*=(A^*)_J\in\gR_J$.  These properties ensures that $\gR_J$ is a unital $^*$-subalgebra of $\gB(\sH_J)$ as said above.
To conclude it is enough establishing that $\gR_J$ is strongly closed.
 Let $(C_n)_{n\in\bN}\subset \gR_J$ which  strongly converges to some $C\in\gB(\sH_J)$, we want to prove that $C\in \gR_J$. By definition,  
 $C_n=(K_n)_J$ for $K_n\in\gR$, $C=K_J$ for some $K\in\gB(\sH)$, and $\|K_nu -Ku\|\to 0$ for every $u\in\sH_J$. 
So, take $x\in\sH$, we know from Prop.\ref{propHJ} that $x=x_1+x_2k$ for unique $x_1,x_2\in\sH_{J}$. Hence, again for Prop.\ref{propHJ}, 
$\|K_nx-Kx\|^2=\|(K_nx_1-Kx_1)+(K_nx_2-Kx_2)k\|^2=\|(K_nx_1-Kx_1)\|^2+\|K_nx_2-Kx_2\|^2\to 0$. All that proves that $\gR \ni K_n \to K$
 strongly and thus $K\in \gR$ because it is a von Neumann algebra. Summing up, $\gR_J \ni C_n \to C= K_J\in \gR_J$ s that $\gR_J$ is strongly closed and therefore is a von Neumann algebra.
\end{proof}

\subsection{Restriction to a real Hilbert subspace induced by two anticommuting complex structures}
If $\sH$ is a quaternionic Hilbert space, suppose we have a pair $J,K$ of anticommuting ($JK=-KJ$) complex structures. 
We want to prove that  the set
\begin{equation}\label{defHJK}
\sH_{JK}:=\{u\in\sH\:|\: Ju=uj,\ Ku=uk\}
\end{equation}
is a real Hilbert space. $\sH_{JK}$ is evidently a real subspace of $\sH$
which also satisfies $JKu = ui$ if $u \in H_{JK}$ from $i=jk$.
However it is not obvious that the scalar product of $\sH$ makes $\sH_{JK}$ a real Hilbert space.
To prove it, we introduce an important technical tool given by  the map \begin{equation}\label{defL}L: \bH \ni q   \mapsto L_{q}:=aI+bJK+cJ+dK \in \gB(\sH) \quad \mbox{where $q=a+ib+jc+kd$\:.}
\end{equation}
$L$ is a {\bf left multiplication} (\cite{GMP1,GMP2}) on $\sH$ --
an injective unital real $^*$-algebra homomorphism -- associated to $J$ and $K$.
%Within this definition we can equip $H_{JK}$ with a left-real Hilbert space strucure , since the vectors of $\sH_{JK}$ satisfy 
Notice that, exploiting definitions (\ref{defHJK}) and (\ref{defL}), it holds
\begin{equation}\label{realitycondition}
L_q u=uq\ \ \forall q\in\bH
\end{equation}
In this sense the elements of $\sH_{JK}$ can be interpreted as \textit{real vectors} of $\sH$, as they \textit{commute} with every element of $\bH$.

\begin{proposition}\label{propEL} Referring to (\ref{defHJK}) a7nd (\ref{defL}), there exists a Hilbert basis $N$ of $\sH$
such that $N \subset \sH_{JK}$ and
\begin{equation*}
L_q=\sum_{z\in N}zq\langle z|\cdot\rangle\:, \quad \forall q \in \bH\:,
\end{equation*}
so that $L_qz=zq$ for every $z\in N$ in particular.
\end{proposition}
\begin{proof}
The proof immediately arises from  Thm 4.3 in \cite{GMP2}.
\end{proof}
\noindent We are in a position to state and prove the result afore mentioned. Notice that, with obvious notation,
 the sets $\sH_{JK}i$, $\sH_{JK}j$, and $\sH_{JK}k$  are real subspaces of $\sH.$

\begin{proposition}\label{realdecomp}
Let $\sH$ be a quaternionic Hilbert space and $J,K\in\gB(\sH)$ two anticommuting complex structures. The  real vector
 subspace $\sH_{JK}$ defined in(\ref{defHJK}) equipped with  the restriction $\langle\cdot|\cdot\rangle_{JK}$ of the scalar product of $\sH$ to $\sH_{JK}$ satisfies the following properties.\\
{\bf (a)} $\sH_{JK}$ is a  real Hilbert space, non-trivial unless $\sH=\{0\}$.\\
{\bf (b)} The direct (generally not orthogonal) decomposition is valid $$\sH=\sH_{JK}\oplus \sH_{JK}i \oplus \sH_{JK}j \oplus \sH_{JK}k$$
such that
\begin{equation}\label{realdecompnorm}
\|x\|^2=\sum_{i=1}^4 \|x_i\|^2\quad \mbox{if }\sH \ni x=x_1+x_2i+ x_3j + x_ak \mbox{ with } x_1,x_2,x_3,x_4 \in \sH_{JK}
\end{equation}
{\bf (c)} The maps 

$\sH_{JK} \ni u\mapsto ui \in \sH_{JK}i$,

$\sH_{JK} \ni u\mapsto uj \in \sH_{JK}j$, 

$\sH_{JK} \ni u\mapsto uk \in \sH_{JK}k$

 \noindent are $\bR$-linear, isometric and bijective. \\
  {\bf (d)} If $N \subset \sH_{JK}$ is a Hilbert basis of $\sH_{JK}$, then $N$ is also a Hilbert basis of $\sH$.
\end{proposition}
\begin{proof}
(a) Let us prove that $\langle \cdot|\cdot\rangle$ takes values in $\bR$ if restricted to $\sH_{JK}$. Let $u,v \in \sH_{JK}$  and $q\in\bH$, exploiting (\ref{realitycondition}) we get
\begin{equation*}
\langle u|v\rangle q=\langle u|vq\rangle=\langle u|L_qv\rangle=
\langle L_q^*u|v\rangle=\langle L_{\overline{q}}u|v\rangle=\langle u\overline{q}|v\rangle=q\langle u|v\rangle
\end{equation*} 
Since $q$ is generic, $\langle u|v\rangle$ must be real. $\langle\cdot|\cdot\rangle$ restricted to $\sH_{JK}$ is therefore a real scalar product
 denoted by $\langle\cdot|\cdot\rangle_{JK}$, because the remaining properties immediately arises from those valid in $\sH$. As aconsequence,
 the norm induced by $\langle\cdot|\cdot\rangle_{JK}$  is equal to the norm of $\sH$ restricted to $\sH_{JK}$.
As $\sH_{JK}$ is closed because $J$ and $K$ are continuous, it is complete as a normed space, since $\sH$ is, and thus it is a real Hilbert space.
 It is also obvious that $N$ of Prop. \ref{propEL} is a Hilbert basis for $\sH_{JK}$ because it is an orthonormal subset of $\sH_{JK}$  whose orthogonal is $\{0\}$.
  Since $N$ is non empty if $\sH$ is non-trivial, $\sH_{JK}\neq \{0\}$.\\
(b) The last statement in (b) arises immediately bearing in mind that when computing $||x||^2 = \langle x|x \rangle$
exploiting the decomposition $x= x_1 + x_2 i+ x_3j + x_4k$ it holds
 $\langle x_a|x_b\rangle=\langle x_b|x_a\rangle\in\bR$, because $x_a,x_b \in \sH_{JK}$, and this impleis that the mixed terms in the expansion  of the scalar product $\langle x|x \rangle$ cancel pairwise.
 Now, let us prove the first part. Let $u\in\sH$, then a direct inspection shows that $u_1:=u -JKui-Juj-Kuk\in\sH_{JK}$,
$u_{jk}:=u-JKui+Juj+Kuk\in(\sH_{JK})i$, $u_j:=u+JKui -Juj+Kuk\in(\sH_{JK})j$, and  $u_k:=u+JKui+Juj-Kuk\in(\sH_{JK})k$. Summing together the four identities we obtain $4u=u_1+ u_{i} + u_j+u_k$. 
The found decomposition is unique, and thus the decomposition of $\sH$ is direct. Indeed suppose that $x_1+x_2i+x_3j+x_4jk$ and $x'_1+x'_2i+x'_3j+x'_4jk$ are two decompositions of the same vector
 $x\in\sH$, then $(x_1-x'_1)+(x_2-x_2')i+(x_3-x'_3)j+(x_4-x'_4)jk=0$. Exploiting (\ref{realdecompnorm}) we immediately get $\|x_\alpha-x'_\alpha\|=0$, hence $x_\alpha=x'_\alpha$, for all $\alpha=1,2,3,4$ concluding the proof.
\\
(c) The proof is elementary using (b) and the basic properties of $J$ and $K$.  The proof of (d) easily arises from the first part of  (b) 
and the fact that $\langle x|y\rangle\in\bR$ if $x,y \in \sH_{JK}$.
\end{proof}
\noindent We can now pass to deal with operators and their restrictions to $\sH_{JK}$.

\begin{proposition}\label{lemmaextensreal}
With the same hypotheses as in Proposition \ref{realdecomp} and with $L$ as in (\ref{defL}), the following facts hold.\\

\noindent {\bf (a)} Consider a $\bH$-linear operator  $X:D(X)\to \sH$ where $D(X)\subset \sH$, such that $JX \subset XJ$ and
 $KX \subset XK$. Then $X_{JK} := X|_{D(X)\cap\sH_{JK}}$ is a well defined  $\bR$-linear operator in $\sH_{JK}$. \\

\noindent {\bf (b)} Consider an $\bR$-linear operator $A:D(A)\rightarrow \sH_{JK}$, where 
$D(A) \subset \sH_{JK}$. Then there exists a unique $\bH$-linear operator $\tilde{A}:D(\widetilde{A})\rightarrow \sH$, 
with $D(\widetilde{A}) \subset \sH$, such that  $J\widetilde{A}\subset \widetilde{A}J$, $K\widetilde{A}\subset \widetilde{A}K$, and 
 $A=\widetilde{A}_{JK}$,  in particular $D(A)=D(\widetilde{A})\cap \sH_{JK}$. \\

\noindent {\bf (c)} The following facts hold if $A$ is as in (b) and  $B:D(B)\rightarrow \sH_{JK}$ is another $\bR$-linear operator where  $D(B) \subset \sH_{JK}$.

(i) $L_q\widetilde{A}=\widetilde{A}L_q$ for all $q\in\bH$.

(ii)   $\widetilde{aA} = a \widetilde{A}$ for every $a\in \bR$.

(iii)  $\widetilde{A+B}=\widetilde{A}+\widetilde{B}$. 

(iv)  $\widetilde{AB}=\widetilde{A}\widetilde{B}$.
 
(v)  $A\subset B$ iff $\tilde{A}\subset\widetilde{B}$.

(vi)  $A\in\gB(\sH_{JK})$ iff  $\widetilde{A}\in\gB(\sH)$ and $\|\widetilde{A}\|=\|A\|$.

(vii) $D(A)$ is dense if and only if $D(\tilde{A})$ is dense

(viii) $\widetilde{A^*}=(\widetilde{A})^*$ if $A$ is densely defined\\
As a consequence,  $A$ is symmetric, antisymmetric, (essentially) selfadjoint,  antiselfadjoint, unitary, idempotent iff $\widetilde{A}$
 is symmetric, antisymmetric, (essentially) selfadjoint,  antiselfadjoint, unitary, idempotent

(ix) $A$ is closable iff $\widetilde{A}$ is closable. In this case  $\overline{\widetilde{A}}=\widetilde{\overline{A}}$.\\

\noindent {\bf (d)} If $\gB(\sH) \ni X_n \to X \in \gB(\sH)$ weakly, resp., strongly for $n \to +\infty$ and $JX_n=X_nJ$, and $KX_n=X_nK$,
 then $JX=XJ$, $KX=XK$ and $\gB(\sH_{JK}) \ni (A_n)_{JK} \to A_{JK} \in \gB(\sH_{JK})$ weakly, resp., strongly for $n \to +\infty$.
\end{proposition}
\begin{proof} See Appendix \ref{AppProof} \end{proof}

\begin{proposition}\label{operatorfunctionquat} Let $\sH$ be a quaternionic Hilbert space and $J,K\in\gB(\sH)$ two anticommuting complex structures. 
If  $A:D(A)\rightarrow \sH$, where $D(A) \subset \sH$, is  selfadjoint and satisfy $JA\subset AJ$ and $KA \subset AK$, then the following facts hold true.
\begin{enumerate}[(a)]
\item 
If $P^{(A)}$ is the PVM of the spectral decomposition of $A$, the operators 
$(P_E^{(A)})_{JK}$ with $E\subset \bR$ Borel set, form  the PVM $P^{(A_{JK})}$ of  $A_{JK}$ in $\sH_{JK}$. Furthermore the map 
$$P^{(A)} \ni P^{(A)}_E \mapsto (P_E^{(A)})_{JK} \in P^{(A_{JK})}$$
is an isomorphism of $\sigma$-complete Boolean lattices.

\item if $f:\bR\rightarrow \bR$ is Borel measurable function, then $f(A)J=Jf(A)$ and $f(A)K=Kf(A)$ and $f(A)_{JK}=f(A_{JK})$.
\item $A$ is positive if and only if $A_{JK}$ is positive.

\item $\sigma_S(A)= \sigma(A_{JK})$, $\sigma_{pS}(A)= \sigma_p(A_{JK})$, $\sigma_{cS}(A)= \sigma_c(A_{JK})$.
\end{enumerate}
\end{proposition}
\begin{proof}
The proof is analogous to the one carried out for Proposition \ref{operatorfunctioncompl}.
\end{proof}
\noindent As we did for the {\em complex} case,  we can now prove the following result. 

\begin{proposition}\label{quaternidentif}
Let $\gR$ a von Neumann algebra over a quaternionic Hilbert space $\sH$ and $J,K\in\gR'$ be complex structures such that $JK=-KJ$. Then $\gR_{JK}:=\{A_{JK}\:|\: A\in\gR\}$ 
is a von Neumann algebra over $\sH_{JK}$.
\end{proposition}
\begin{proof}
$\gR_{JK}$ is clearly a unital sub $^*$-algebra of $\gB(\sH_{JK})$ thanks to Prop. \ref{lemmaextensreal} and the fact that only real combinations are to be taken into account 
in the algebra $\gB(\sH_{JK})$ as well as in the algebra $\gB(\sH)$. The algebra $\gR_{JK}$ turns out to be closed in the strong topology on $\sH_{JK}$, the proof
 being identical to the one carried out in Prop. \ref{complexidentif}. 
\end{proof}

\section{Lie-Group representations in quaternionic Hilbert spaces}
This section is devoted to generalize to the quaternionic Hilbert space case some well-known results of representation theory especially of Lie groups in terms of unitary operators.
\subsection{G\r{a}rding domain and Lie algebra representation}
As usual,  $C_0^\infty(G)$ denotes the vector space  of real-valued infinitely differentiable compactly supported functions on a given (real finite-dimensional) Lie group $G$ and we consider 
 a strongly-continuous unitary representation $U$ of $G$ on a quaternionic Hilbert space $\sH$.
If $x\in\sH$ and $f\in C_0^\infty(G)$ we define 
\begin{equation}\label{gardvector}
x[f]=\int_G f(g) U_gx\, dg
\end{equation}
where $dg$ is the left-invariant Haar measure on $G$, as the unique vector in $\sH$ such that $\langle y|x[f]\rangle = \int_G f(g)\langle y|U_gx\rangle\, dg$ for all $y\in\sH$ - the existence and uniqueness of such a 
vector is the content of the Riesz representation theorem which holds true also on quaternionic Hilbert space, as described in \cite{GMP1}.

\noindent Next we may extend to the quaternionic Hilbert space case a well-known definition of the real and complex cases (see, e.g., \cite{MO1} for the real case).
\begin{definition}\label{defGarding}{\em Given a Lie group $G$ and a  strongly-continuous unitary representation $U$ of $G$ on a quaternionic Hilbert space $\sH$,
the $\bH$-linear subspace of $\sH$ generated by all the vectors $x[f]$ (\ref{gardvector}) is called the \textbf{G\r{a}rding Domain} associated with $U$ and denoted by $\cD_G^{(U)}$.}
\end{definition}
\begin{remark}\label{gardingrealcombin}{\em
Actually $\cD_G^{(U)}$ coincides with the {\em real} span of vectors $x[f]$. Indeed, 
\begin{equation*}
\begin{split}
&\langle y|x[f]q\rangle=\langle y|x[f]\rangle q=\left(\int_G f(g) \langle y|U_g x\rangle\, dg \right)q= \int_G f(g) \langle y|U_g x\rangle q\, dg=\\
&=\int_G f(g) \langle y|(U_g x)q\rangle \, dg=\int_G f(g) \langle y|U_g(xq)\rangle q\, dg=\langle y|(xq)[f]\rangle
\end{split}
\end{equation*}
Since this holds for any $y\in\sH$ we have that $x[f]q=(xq)[f]$.
}
\end{remark}
\noindent We can immediately prove the following result.  Due to (d) below, the elements of $D_G^{(U)}$ are also called the \textbf{smooth} vectors of $U$.
\begin{proposition}\label{propgarding}
Referring to Def.\ref{defGarding}, the following facts hold true. 
\begin{enumerate}[(a)]
\item Denoting by $U_\bR$ and $U_{\bC_j}$ the map $U$ respectively viewed as strongly-continuous unitary representations on $\sH_\bR$ and $\sH_{\bC_j}$, it holds 
$\cD_G^{(U)}=\cD_G^{(U_\bR)}=\cD_G^{(U_{\bC_j})}$
\item $\cD_G^{(U)}$ is dense in $\sH$.
\item $U_g(\cD_G^{(U)})\subset \cD_G^{(U)}$ for all $g\in G$.
\item $x\in D_G^{(U)}$ if and only if the map $G \ni g\mapsto U_gx$ is smooth at every point $g\in G$ with respect to the smooth atlas of $G$.
\end{enumerate}
\end{proposition}
\begin{proof}
We prove (a) and (d). Items (b) and (c) immediately follow from (a) and the G\r{a}rding theory on real Hilbert spaces (e.g., see \cite{MO1} Thm 3.6). Let $x\in\sH_\bR=\sH$ and $f\in C_0^\infty(G)$,  
 $x(f)$ is the vector (\ref{gardvector}) but defined with respect to the real Hilbert space structure of $\sH_\bR$ while we keep the notation $x[f]$ for the quaternionic case. We
  intend to prove that $x(f)=x[f]$.
 $x(f)$ is the only vector such that $(y|x(f))=\int_G f(g) (y|U_gx)\, dg$ for all $y\in\sH_\bR$, where $(\cdot|\cdot)=Re\langle\cdot|\cdot\rangle$.  $x[f]$, instead, is  the only vector such that
$\langle y|x[f]\rangle = \int_G f(g)\langle y|U_gx\rangle\, dg$ for all $y\in\sH$.
Taking the real part of both sides we get
%\begin{equation*}
%\begin{split}
%&Re\langle y|x[f]\rangle+Im\langle y|x[f]\rangle=\langle y|x[f]\rangle=\int_G f(g) \langle y|U_gx\rangle\, dg=\\
%&=Re\int_G f(g) \langle y|U_gx\rangle\, dg+Im\int_G f(g) \langle y|U_gx\rangle\, dg=\\
%&=\int_G f(g) Re\langle y|U_gx\rangle\, dg+\int_G f(g) Im\langle y|U_gx\rangle\, dg
%\end{split}
%\end{equation*}
%This gives 
$Re\langle y|x[f]\rangle=\int_G f(g) Re\langle y|U_gx\rangle\, dg$ for all $y\in\sH=\sH_\bR$. By definition of $x(f)$ it follows $x[f]=x(f)$. This clearly proves 
that $\cD_G^{(U_\bR)}\subset\cD_G^{(U)}$. The converse inclusion holds because 
$(x[f])q=(xq)[f]$ for every $q\in\bH$ and thus $x[f]j,x[f]k,x[f]jk$ belongs to $\cD_G^{(U_\bR)}$. We have established that $\cD_G^{(U_\bR)}=\cD_G^{(U)}$.
Since  $\sH_{\bC_j}=(\sH_\bR)_\cJ$, (b) Thm 3.5 in \cite{MO1} proves $\cD_G^{(U_\bR)}=\cD_G^{(U_{\bC_j})}$.\\
(d) The notion of differentiability only uses the norm and $\bR$-linearity of $\sH$, hence $G \ni g\mapsto U_gx$ is smooth with respect to $\sH$ if 
and only if is smooth  with respect to the underlying real structure $\sH_\bR$. Since the thesis holds for real Hilbert spaces, point (a) concludes the proof.
\end{proof}
 \noindent
Let us pass to the representation theory of the   the Lie algebra $\gg$ associated to a Lie group $G$. Take ${\bf A}\in\gg$ and consider the 
one-parameter subgroup $\bR\ni t\mapsto \exp(t{\bf A})\in G$. Thanks to Theorem \ref{stonetheorem} there exists a unique antiselfadjoint 
operator $A$ on $\sH$ such that $U_{\exp(t{\bf A})}=e^{tA}$ and $A$ equals the antiselfadjoint generators of $t\mapsto U_{\exp(t{\bf A})}$ 
when interpreted as acting on $\sH_\bR$ or $\sH_{\bC_j}$. Thanks to Prop.\ref{propgarding} (a) and the properties of G\r{a}rding domain in real
 Hilbert spaces, we easily get the following result.
\begin{proposition}\label{algebrarepr} Consider a Lie group $G$ with Lie algebra $\gg$, and a  strongly-continuous unitary representation $U$ of $G$ 
on a quaternionic Hilbert space $\sH$. The following facts are valid.
\begin{enumerate}[(a)]
\item If ${\bf A}\in\gg$ and $U_{\exp(t{\bf A})}=e^{tA}$, then $A\left(\cD_G^{(U)}\right)\subset\cD_G^{(U)}$ and $\cD_G^{(U)}$ is a core for $A$, i.e., $\overline{A|_{\cD_G^{(U)}}}=A$. 
\item The map
$\gg\ni {\bf A}\mapsto u({\bf A}):=A|_{\cD_G^{(U)}}\in\mathcal{L}(\cD_G^{(U)})$
is a Lie algebra homomorphism.
\item If $u_\bR$ and $u_{\bC_j}$ denotes the Lie algebra homomorphism associated with $U$ as above in $\sH_\bR$ and $\sH_{\bC_j}$  respectively, 
it holds $u({\bf A})=u_\bR({\bf A})=u_{\bC_j}({\bf A})$. In particular the antiselfadjoint generators related to ${\bf A}$ defined on $\sH,\sH_\bR,\sH_{\bC_j}$ are equal. 
\end{enumerate}
\end{proposition}
\begin{proof} The antiselfadjoint generators related to ${\bf A}$ defined on $\sH,\sH_\bR,\sH_\bC$ are equal due to Thm\ref{stonetheorem}. Since also $\cD_G^{(U)}$ is
 independent from the Hilbert space (Prop.\ref{propgarding}), exploiting  Th 3.6 (c) in \cite{MO1} we get $A(\cD_G^{(U)})\subset \cD_G^{(U)}$ and thus $u$ is independent as well. We have proved (c).
Now, let $\mathcal{L}_\bR(\cD_G^{(U)})$ be the Lie algebra of real-linear operators defined on $\cD_G^{(U)}$, of course 
$\mathcal{L}(\cD_G^{(U)})\subset \mathcal{L}_\bR(\cD_G^{(U)})$. The map \beq \gp\ni {\bf A}\mapsto u({\bf A}):=A|_{\cD_G^{(U)}}\in\mathcal{L}_\bR(\cD_G^{(U)})\label{umap}\eeq
is a Lie algebra homomorphism thanks to 
Thm 3.6 (d) in \cite{MO1} applied to $\sH_\bR$. Since actually $A|_{\cD_G^{(U)}}\in\mathcal{L}(\cD_G^{(U)})$, (b) is true. Finally, exploiting Thm (f) 3.6 in \cite{MO1} and (e)-(v) Prop. \ref{propHR} we conclude the proof of (a).
\end{proof}
\noindent As we did in the real theory, we can consider the real associative unital {\bf universal enveloping algebra} $E_\gp$ of the Lie algebra $\gp$
(e.g., see \cite{V} and the quick account in Appendix E of \cite{MO1}). By construction  $\gp \subset E_\gp$ as Lie subalgebra and, if $\circ$ denotes 
the product of the algebra whch extends that of $\gg$, a generic element of $E_\gp$ takes the form
\begin{equation}\label{MO1dev}
{\bf M}=c_01+\sum_{k=1}^N \sum_{j=1}^{N_k}c_{jk}{\bf A}_{j1}\circ...\circ{\bf A}_{jk}
\end{equation}
for some $N, N_k\in\bN$, $c_0,c_{jk}\in\bR$ and ${\bf A}_{jm}\in\gg$.
We also assume to  endow $E_\gp$  with the standard real involution $\null^+$ completely defined by requiring
\begin{equation}\label{Ptilde}
\left(c_01+\sum_{k=1}^N \sum_{j=1}^{N_k}c_{jk}{\bf A}_{j1}\circ...\circ{\bf A}_{jk}\right)^+=c_01+\sum_{k=1}^N \sum_{j=1}^{N_k}c_{jk}(-1)^k{\bf A}_{jk}\circ...\circ {\bf A}_{j1}\:.
\end{equation}
An element ${\bf M}\in\gp$ is said to be {\bf symmetric} if ${\bf M}^+={\bf M}$, so that $\gp$ is made of antisymmetric elements.
An important technical role is played by  the {\bf Nelson elements} of  $E_\gg$ which are those of the form \beq{\bf N}:=\sum_{i=1}^n{\bf X}_i\circ{\bf X}_i\:,\label{NelsonE}\eeq
where $\{{\bf X}_1,\dots{\bf X}_n\}$ is any basis of $\gg$. Notice that they are symmetric by construction.\\
Exactly as above we can now take advantage of the real theory and state the following theorem which has the same proof as the corresponding items of Thm 3.6 and Prop.3.8 in \cite{MO1}.
\begin{theorem}\label{thmuext}
Referring to Proposition \ref{algebrarepr} the following facts hold
\begin{enumerate}
\item The Lie algebra homomorphism $u$ defined in (\ref{umap}) uniquely extends to a real unital associative algebra representation of the universal enveloping algebra $E_\gg$ defined on $\cD_G^{(U)}$. More precisely, if ${\bf M}$ is taken as in (\ref{MO1dev}), then it holds
\begin{equation}\label{defugen}
u({\bf M}) =c_0I|_{\cD_G^{(U)}}+\sum_{k=1}^N\sum_{j=1}^{N_k}c_{jk}u({\bf A}_{j1})\cdots u({\bf A}_{jk})
\end{equation}
It also holds $u({\bf M}^+)\subset u({\bf M})^*$, in particular $u({\bf M})$ is symmetric if ${\bf M}^+={\bf M}$.
\item Suppose that ${\bf M}\in E_\gg$ satisfies both ${\bf M}={\bf M}^+$ and $[{\bf M},{\bf N}]_\gg=0$ for some Nelson element ${\bf N}$, then $u({\bf M})$ 
is essentially self-adjoint. In particular $u({\bf N})$ is always essentially self-adjoint.
\item if $B\in\gB(\sH)$ the following facts are equivalent
\begin{enumerate}[(i)]
\item $Bu({\bf A})\subset u({\bf A})B$ for every ${\bf A}\in\gg$
\item $B\overline{u({\bf A})}\subset \overline{u({\bf A})}B$ for every ${\bf A}\in\gg$
\item $BU_g=U_gB$ for every $g\in G$
\end{enumerate}
If one of these conditions is satisfies, then $B(\cD_G^{(U)})\subset \cD_G^{(U)}$
\end{enumerate}
\end{theorem}

\subsection{Analytic vectors of unitary representations in quaternionic Hilbert spaces}

There is another important subspace of $\sH$ made of ``good vectors'' for a strongly-continuous unitary representation $U$ of a Lie group $G$, even better than  $\cD_G^{(U)}$. 
To introduce this space we need a definition.
A function $f:\bR^n\supset U\rightarrow \sH$ is called \textbf{real analytic} at $x_0\in U$ if there exists a neighbourhood $V\subset U$ of $x_0$ where
 the function $f$ can be expanded in power series
\begin{equation}
f(x)=\sum_{|\alpha|= n, n=0}^\infty (x-x_0)^\alpha v_\alpha,\ x\in V
\end{equation}
for suitable $v_\alpha\in\sH$ for every multi-index $\alpha\in\bN^n$.

\begin{definition}{\em 
Let $\sH$ be a quaternionic Hilbert space and $G\ni g\mapsto U_g$ a unitary strongly-continuous representation on $\sH$ of the Lie group $G$. 
A vector $x\in \sH$ is said to be \textbf{analytic} for $U$ if the function $g\mapsto U_gx$ is real analytic at every point $g\in G$, referring to the 
analytic atlas of $G$. The linear subspace of $H$ made of these vectors is called the \textbf{Nelson space} of the representation and denoted by $\cD_N^{(U)}$.}
\end{definition}

\begin{remark}{\em Adopting notations as in Prop.\ref{propgarding}, 
it is evident from the definition that $\cD_N^{(U)}=\cD_N^{(U_\bR)}=\cD_N^{(U)}$
}
\end{remark}
\noindent There exists another related definition of \textit{analyticity} for vectors.  Let $A: D(A) \to \sH$ an operator in a quaternionic Hilbert space $\sH$, we say that a vector
 $x \in \bigcap_{n=0}^{+\infty} D(A^n)$
is {\bf analytic} for $A$, if there exists $t_{x}>0$ such that 
\begin{equation}
 \sum_{n=0}^{+\infty} \frac{t_x^n}{n!}||A^nx|| <+\infty \label{defanalyticvector}
\end{equation}
From the elementary theory of series of powers, we know that  $t$ above can be replaced for every $z\in \bC$ with $|z|< t_x$  obtaining an absolutely convergent series.

\begin{remark}\label{remanalitic}
{\em It should be evident  that the analytic vectors for $A$ form a subspace of $D(A)$. Moreover a vector $x$ is analytic for $A$ on $\sH$ if and only if it is analytic for $A$ on $\sH_\bR$ if and only if it is analytic for $A$ on $\sH_{\bC_j}$
}
\end{remark}
\noindent One of remarkable Nelson's results, here  extended to quaternionic Hilbert spaces, states that
\begin{proposition}\label{propNseries} Consider an operator  $A : D(A) \to \sH$ on a quaternionic Hilbert space $\sH$.
\begin{enumerate}[(a)]
\item If $A$ is anti selfadjoint  and $x \in D(A)$ is analytic with $t_x >0$ as in (\ref{defanalyticvector}), then
$$e^{tA}x = \sum_{n=0}^{+\infty} \frac{t^n}{n!}A^nx\quad \mbox{if $t\in \bR$ satisfies $|t|\leq t_x$.}$$
\item If $A$ is (anti) symmetric and $D(A)$ includes a set of analytic vectors whose finite span is dense in $\sH$, then $\overline{A}$ is (anti)  selfadjoint and $D(A)$.
\end{enumerate}
\end{proposition} 
\begin{proof} 
Point (a). Since the thesis holds for real Hilbert spaces (Thm 3.13 in \cite{MO1}) and the construction of $e^{tA}$ with respect to the Hilbert structures 
of $\sH$ and $\sH_\bR$ are equal to each other we immediately get the thesis.  Point (b). Again, the thesis hold for real Hilbert spaces and an 
operator over $\sH$ is (anti) symmetric or (anti) self-adjoint over $\sH$ if and only if it is so over $\sH_\bR$. The thesis follows immediately.
\end{proof}
\begin{theorem}\label{teonelson}
The Nelson subspace $D_N^{(U)}$ satisfies the following properties.
\begin{enumerate}[(a)]
\item $D_N^{(U)}\subset D_G^{(U)}$
\item $U_g(D_N^{(U)})\subset D_N^{(U)}$ for any $g\in G$
\item $D_N^{(U)}$ is dense in $\sH$
\item $D_N^{(U)}$ consists of analytic vectors for every operator $u({\bf A})$ with ${\bf A}\in \gg$
\item $u({\bf A})(D_N^{(U)})\subset D_N^{(U)}$ for any ${\bf A}\in\gg$.
\item Let $p : \bR \to \bR$ be a real polynomial such that either
$$ \mbox{$p(-x) = p(x)$ for every $x\in \bR$ or $p(-x) = -p(x)$ for every $x\in \bR$\:.}$$
If  ${\bf A} \in \gg$ then
 $\overline{u(p({\bf A}))}$ is, respectively,   selfadjoint or anti selfadjoint. 
\end{enumerate}

\end{theorem}
\begin{proof}
The thesis holds for real Hilbert spaces (Thm 3.14 in \cite{MO1}). So, taking into account the equalities $D_G^{(U)}=D_G^{(U_\bR)}$ and $D_N^{(U)}=D_N^{(U_\bR)}$
 and Remark \ref{remanalitic} we immediately get the result.
\end{proof}

\subsection{Group representations restricted to a complex Hilbert subspace  induced by a complex structure}
To conclude this section we study the case  of a group representation $U$ in a quaternionic Hilbert space equipped with a complex structure commuting with $U$.
 We start with the most elmentary situation.

\begin{lemma}\label{Lemmagenertorsrestrictoin}
Let $J$ be a complex structure in the quaternionic Hilbert space $\sH$ and
and $U:\bR\ni t\mapsto U_t\in\gB(\sH)$ a strongly-continuous one-parameter group of unitary operators such that $JU_t=U_tJ$ for every $t\in\bR$. If $A$ is the anti-selfadjoint generator of $U$, then \\

\noindent {\bf (a)} $AJ=JA$,\\

\noindent {\bf (b)} $A_J$ is the anti-selfadjoint generator of $U_J:\bR\ni t\mapsto (U_t)_J\in\gB(\sH_J)$.
\end{lemma}
\begin{proof}
 We know from Stone Theorem that $U_t=e^{tA}$ for some anti-selfadjoint operator $A:D(A)\rightarrow \sH$. Suppose that $U_tJ=JU_t$ for 
 every $t\in\bR$, then Lemma \ref{LemmaCOMM} yields  $JA\subset AJ$. Since $J$ is unitary and antiselfadjoint this immediately leads to $JA=AJ$.  $A_J$ 
 is anti-selfadjoint and defined on $D(A_J)=D(A)\cap\sH_J$ by Prop.\ref{propHJop}  and $x\in D(A_J)$ if and only if $x\in\sH_J$ and there exists
$
\lim_{t\to 0}\frac{U_t x-x}{t}=\lim_{t\to 0}\frac{(U_t)_J x-x}{t}=A_Jx
$
due to Stone theorem.
This is exactly the definition of the generator of $U_J$ proving the thesis.
\end{proof}
\noindent We pass to the main result.
\begin{proposition}\label{gardintcompl}
Let $U:G\ni g\mapsto U_g\in\gB(\sH)$ a strongly-continuous unitary Lie-group representation over the quaternionic Hilbert space $\sH$ and $J$ a complex
 structure such that $JU_g=U_gJ$ for every $g\in G$. Then the following facts hold
\begin{enumerate}[(a)]
\item $J(D_G^{(U)})= D_G^{(U)}$,
\item $U_J: G\ni g\mapsto (U_g)_J\in\gB(\sH_J)$ is a strongly-continuous unitary representation over $\sH_J$,
\item $D_G^{(U_J)}=D_G^{(U)}\cap \sH_J$,
\item $Ju({\bf M}) = u({\bf M})J$ for all ${\bf M}\in E_\gg$. Moreover $u_J({\bf M})=u({\bf M})_J$. 
\end{enumerate}
\end{proposition}
\begin{proof} See Appendix \ref{AppProof} \end{proof}

\subsection{Group representations restricted to a real Hilbert subspace  induced by two anti commuting complex structures}

\begin{lemma}\label{Lemmagenertorsrestrictoinrea}
Let $J,K$ be anticommuting complex structures in the quaternionic Hilbert space $\sH$ and
and $U:\bR\ni t\mapsto U_t\in\gB(\sH)$ a strongly-continuous one-parameter group of unitary operators such that $JU_t=U_tJ$ and $KU_t=U_tK$ for every $t\in\bR$. If 
$A$ is the antiselfadjoint generator of $U$, then \\

\noindent {\bf (a)} $AJ=JA$ and $AK=KA$,\\

\noindent {\bf (b)} $A_{JK}$ is the anti-selfadjoint generator of $U_{JK}:\bR\ni t\mapsto (U_t)_{JK}\in\gB(\sH_{JK})$.
\end{lemma}
\begin{proof}
Analogous to the proof of Lemma \ref{Lemmagenertorsrestrictoin}
\end{proof}
\begin{proposition}\label{gardingtreal}
Let $U:G\ni g\mapsto U_g\in\gB(\sH)$ a strongly-continuous unitary Lie-group representation over the quaternionic Hilbert space $\sH$ and Let $J,K$ be anticommuting 
complex structures such that $JU_g=U_gJ$ and $KU_g=U_gK$ for every $g\in G$. Then the following facts hold
\begin{enumerate}[(a)]
\item $J(D_G^{(U)})\subset D_G^{(U)}$ and $K(D_G^{(U)})\subset D_G^{(U)}$
\item $U_{JK}: G\ni g\mapsto (U_g)_{JK}\in\gB(\sH_{JK})$ is a strongly-continuous unitary representation over $\sH_{JK}$
\item $D_G^{(U_{JK})}=D_G^{(U)}\cap \sH_{JK}$
\item $u({\bf M})J=Ju({\bf M})$ and $u({\bf M})K=Ku({\bf M})$  for all ${\bf M}\in E_\gg$. Moreover $u_{JK}({\bf M})=u({\bf M})_{JK}$.
\end{enumerate}
\end{proposition}
\begin{proof}
The proof is analogous to  Proposition \ref{gardintcompl}'s one. 
\end{proof}

\section{Irreducibility and Schur's Lemma on Quaternionic Hilbert spaces}
The notion of irreducibility will play crucial role in our work.

\begin{definition}
{\em Let $\sH$ be  a real, complex or quaternionic Hilbert space.
A family  of operators  $\gS\subset \gB(\sH)$ is said to be {\bf  irreducible} if $U(\sK)\subset \sK$ for all $U\in \gR$ and  
 a closed subspace $\sK \subset \sH$ implies $\sK=\{0\}$ or $\sK =\sH$. 
$\gS$ is said to be {\bf reducible} if it is not irreducible.}
\end{definition}
\noindent Since the definition refers to {\em closed} subspaces,  our notion of irreducibility is sometimes called {\em topological} irreducibility.

\begin{remark}\label{remirr} $\null$\\
{\em
{\bf (a)} If $\gS$ is irreducible, then $\gS'\cap \cL(\sH)=\{0,I\}$. (This is because if an orthogonal projector $P$ commutes with all elements of $\gS$, 
then the closed subspace  $\sK := P(\sH)$ is invariant under $\gS$).\\
{\bf (b)}  $\gS'\cap \cL(\sH)=\{0,I\}$ is {\em equivalent} to irreducibility of $\gS$  when $\gS$ is closed under Hermitian conjugation. (This is because
 $U(\sK)\subset \sK$ implies $PUP = UP$ if $P$ is the orthogonal projector onto $\sK$. If this holds for every $U \in \gR$ which is closed under 
 Hermitian conjugation, $PU^*P=U^*P$ holds as well.  Taking the Hermitian conjugate, $PUP=PU$, so that $UP=PU$.)}
\end{remark}

\subsection{Schur's lemma in quaternionic Hilbert spaces}
Let us pass to formulate the quaternionic version of Schur's Lemma which, as in the real Hilbert space case, has a formulation more complicated than in the standard complex Hilbert space case.

\begin{proposition} [Schur's lemma  for essentially selfadjoint operators] \label{SL}
Let $\sH$  be a quaternionic Hilbert space and let $\gS \subset \gB(\sH)$ be irreducible.\\ If the operator $A : D(A) \to \sH$, with $D(A) \subset \sH$ dense, is essentially selfadjoint and 
\begin{equation}\label{irrcomm}
 UA \subset  AU\quad \mbox{for all}\quad U \in\gS
\end{equation}
then  $\overline{A} \in \gB(\sH)$ (the bar denoting the closure of $A$) and 
\begin{equation*}
\overline{A}= a I\:,\quad \mbox{for some $a \in \bR$.}
\end{equation*}
If $A$ satisfying (\ref{irrcomm}) is selfadjoint, then $A\in \gB(\sH)$ with  $A = a I$
for some $a \in \bR$.
\end{proposition}

\begin{proof} The proof, based on the spectral theorem, is essentially identical to that of Prop. 2.13 in \cite{MO1} which is valid for real and complex Hilbert spaces. 
\end{proof}

\noindent A different and more precise result can be obtained when the class $\gR$ consists of a $^*$-closed subset of $\gB(\sH)$. In this case the statement is 
different from the one valid in complex Hilbert spaces but the same statement holds in real Hibert spaces. 

\begin{proposition}\label{SL2}
Let $\sH$ be a quaternionic Hilbert space, $\gS\subset\gB(\sH)$ a $^*$-closed subset and consider a densely-defined closed operator in $\sH$ $A:D(A)\rightarrow \sH$ such that
\begin{equation}\label{SL2'}
UA= AU,\ \ UA^*= A^*U\:, \quad \forall U \in\gS\:.
\end{equation}
If $\gS$ is irreducible then
$
A=aI+bJ
$
for some $a,b\in\bR$ and $J$ a complex structure. In particular  $D(A)=\sH$ and $A\in\gB(\sH)$ in both cases.
\end{proposition}
\begin{proof}
$AU= UA$ and $A^*U= UA^*$ imply
$A^*AU = A^*UA= UA^*A$ on natural domains.
Since  $A$ is closed, the operator $A^*A$ is densely defined and selfadjoint (Thm \ref{PDT}). 
Proposition \ref{SL} for the selfadjoint operator 
$A^*A$ implies  $A^*A= aI$ for some real $a$. In particular $D(A^*A)=D(aI)=\sH$ so that $D(A)=\sH$  and thus, since $A$ is closed, the closed graph 
theorem (\cite{GMP1} Prop.2.11) gives $A\in \gB(\sH)$. To go on,
decompose $A=\frac{A+A^*}{2}+\frac{A^*-A}{2}$ where the two addends 
denoted by $A_S$ and $A_A$
are, respectively, selfadjoint and anti selfadjoint belong to $\gB(\sH)$  and commute with the elements of $\gS$. In particular, $U A_S=A_S U$ for any 
$U\in \gS$ gives $A_S=aI$ for some $a\in\bR$, thanks to Prop.\ref{SL}. Similarly,  $A_A^2\in \gB(\sH)$ is selfadjoint and commutes with the operators in $\gS$, 
hence $A_A^2=cI$ for some $c\in\bR$, thanks again to Proposition \ref{SL}. It must be  $c\le 0$ because, if $v\in\sH$ has unit 
norm,  $c=\langle v|cv\rangle=\langle v|A_AA_Av\rangle=-\langle A_Av|A_Av\rangle=-\|A_Av\|^2\le 0$. In particular, $c=0$ if
 and only if $A_A=0$, that is if $A$ is selfadjoint and in this case the proof ends. In the case  $c\neq 0$, define $J:=\frac{A_A}{\sqrt{-c}}$. 
 With this definition we find  $J\in \gB(\sH)$, $J^*=-J$ and $J^*J=-I$ so that $J$ is a complex structure  and
$A= aI + bJ$ for $a,b \in \bR$ ending the proof again.
\end{proof} 
\begin{remark}\label{remarkshurrep}{\em$\null$\\
{\bf (a)} If $\gS:=\{U_g\:|\:g\in G\}$ for a unitary group representation $G\ni g\mapsto U_g\in\gB(\sH)$,  the hypothesis (\ref{SL2'}) of Proposition \ref{SL2} can be weakened to
$$
U_g A\subset A U_g\:,\quad \forall g\in G\:.
$$ 
Indeed, multiplying both side  by $U_{g^{-1}}$ on the left and  by $U_g$ on the right, we get  $AU_g\subset U_gA$ so that $U_gA=AU_g$.
Taking the adjoint of this identity we also have $U_g^*A^*\subset A^*U_g^*$, because $U_g$ is bounded. Since $U_g^*=U_{g^{-1}}$
 and $g$ varies on the whole set $G$ we have found $U_gA^*\subset A^*U_g$ and thus $U_gA^*=A^*U_g$ with the same reasoning as above, recovering (\ref{SL2'}).\\
{\bf (b)} If $A$ is bounded, the hypothesis (\ref{SL2'}) of Proposition \ref{SL2} can be weakened to  $UA=AU$ for every $U\in\gS$. Indeed the second identity in (\ref{SL2'})
 immediately follows  from $UA=AU$ and $^*$-closure of $\gS$.\\
{\bf (c)} In general, irreducibility of a unitary representation of a group on $\sH$ is lost when moving from the quaternionic Hilbert space structure to the underlying
 real Hilbert space one $\sH_\bR$. Indeed consider the following example. Let $\sH=\bH$ and $G=SO(3)$ and define the representation $G \ni R\mapsto U_R$ 
 defined by $U_R(a,{\bf b}):=(a,R{\bf b})$ for all $(a,{\bf b})\in\bH$. This is clearly unitary and irreducible (we are working on a one-dimensional Hilbert space). 
 Of course $\sH_\bR=\bR^4$ and $U_R(a,{\bf 0})=(a,{\bf 0})$ for all $a\in\bR$, hence the one-dimensional subspace $\{(a,{\bf 0})\:|\:a\in\bR\}$ is invariant 
 under the action of $U$. This make the representation reducible on $\bR^4$.}
\end{remark}
\noindent To conclude this general part let us consider the case of  $^*$-closed subset of $\gB(\sH)$ equipped with one or two complex structures.

\begin{proposition}\label{complexidentif2} 
Let $\gS$ a $^*$-closed subset of $\gB(\sH)$ for a quaternionic Hilbert space $\sH$.  The following facts hold.\\

\noindent {\bf (a)} If there is a  complex structure $J\in\gS'$, then the set of complex-linear operators  $\gS_J := \{A_J\:|\: A \in \gS\} \subset \gB(\sH_J)$ is irreducible if $\gS$ is. \\

\noindent {\bf (b)} If there is a pair of complex structures $J, K \in\gS'$
with $JK=-KJ$,
 then the set of real-linear operators $\gS_{JK}  := \{A_{JK}\:|\: A \in \gS\} \subset \gB(\sH_{JK})$ is irreducible if $\gS$ is. 
\end{proposition}
\begin{proof} 
Dealing with $^*$-closed sets, irreducibility is equivalent to the non-existence of non-trivial projectors commuting with the algebra for (b) Remark \ref{remirr}.\\
(a) Suppose that $\gS$ is irreducible and let $P\in\cL(\sH_J)\cap (\gS_J)'$. We have $P=\widetilde{P}|_{\sH_J}$ for some $\widetilde{P}\in\gB(\sH)$, which is clearly an
 orthogonal  projector on $\sH$ for Prop.\ref{propHJop} because $P$ is. That proposition also implies that $P$ commutes with every element of $\gS_J$ if and only 
 if $\widetilde{P}$ commutes with every element of $\gS$. Since $\gS$ is irreducible we have the thesis.  
The proof of (b) is essentially identical.
\end{proof}

\subsection{Application to Lie-group representations}
A remarkable consequence of the properties of Nelson's technology  and our version of Schur's lemma for Lie group representations is the following proposition.

\begin{proposition}\label{propCSAI}
Let $g\mapsto U_g$ be an irreducible strongly-continuous unitary representation of a connected Lie group over a quaternionic Hilbert space $\sH$ and ${\bf M}\in E_\gg$ such that 
\begin{enumerate}[(a)]
\item $u({\bf M})$ is essentially self-adjoint,
\item $[{\bf M},{\bf A}]_\gg=0\ \ \forall {\bf A}\in\gg$
\end{enumerate}
then it holds $u({\bf M})=cI|_{D_G^{(U)}}$ for some $c\in\bR$.\\
In particular the thesis holds if ${\bf M}= {\bf M}^+$ and  (b) is valid.
\end{proposition}
\begin{proof} 
Let $x\in D_N^{(U)}$ and ${\bf A}\in\gg$. Thanks to Theorem \ref{teonelson} it holds $x\in D_G^{(U)}$ and $x$ is analytic for $u({\bf A})$, in particular it is analytic for $\overline{u(\bf A)}$. 
Exploiting Prop.\ref{propNseries}, we have that  there exists $t_{{\bf A},x}>0$ such that
$$
U_{\exp(t{\bf A})}x=e^{t\overline{u({\bf A})}}x=\sum_{n=0}^\infty\frac{t^n}{n!}u({\bf A})^nx,\ \ |t|\le t_{{\bf A},x}\:.
$$
Moreover $D_N^{(U)}$ is invariant under the action of $u$, hence $u({\bf M})x\in D_N^{(U)}$. Then there exits $t_{{\bf A},u({\bf M})x}>0$ such that
$$
U_{\exp(t{\bf A})}u({\bf M})x=e^{t\overline{u({\bf A})}}u({\bf M})x=\sum_{n=0}^\infty\frac{t^n}{n!}u({\bf A})^nu({\bf M})x,\ \ |t|\le t_{{\bf A},u({\bf M})x}\:.
$$
Now take a positive  real  $t_x<\min\{t_{{\bf A},x},t_{{\bf A},u({\bf M})x}\}$. Using $[u({\bf M}),u({\bf A})]=0$ we have
$$
U_{\exp(t{\bf A})}u({\bf M})x=\sum_{n=0}^\infty\frac{t^n}{n!}u({\bf M})u({\bf A})^nx,\ \ |t|\le t_{x}\:.
$$
Since   $u({\bf M})$ is closable, it follows directly from the equations above and the invariance of $D_G^{(U)}$ under the action of
 $U$ that $$U_{\exp(t{\bf A})}u({\bf M})x = \sum_{n=0}^\infty\frac{t^n}{n!}u({\bf M})u({\bf A})^nx = u({\bf M})U_{\exp(t{\bf A})}x$$ for every $|t|\le 
t_x$. Actually this equality holds for every $t\in\bR$. Indeed define 
$\cZ:=\{z>0|u({\bf M})U_{\exp(t{\bf A})}x=U_{\exp(t{\bf A})}u({\bf M})x,\ |t|\le z\}$ and let $t_0:=\sup\cZ$. Suppose that $t_0<\infty$, then it is
 easy to see that the fact that  $u({\bf M})$ is closable ensures that $u({\bf M})U_{\exp(t_0{\bf A})}x=U_{\exp(t_0{\bf A})}u({\bf M})x$, hence
 $t_0\in\cZ$. We know that $y:=U_{\exp(t_0{\bf A})}x\in D_N^{(U)}$,  we can therefore  repeat the above reasoning  finding a real
 $t_y>0$ such that $u({\bf M})U_{\exp(t{\bf A})}y=U_{\exp(t{\bf A})}u({\bf M})y$ for every $|t|\le t_y$.  Noticing that  
 $\exp((t+t_0){\bf A})=\exp(t{\bf A})\exp(t_0{\bf A})$, it straightforwardly follows  that
 $u({\bf M})U_{\exp(t+t_0){\bf A}}x=U_{\exp(t+t_0){\bf A}}u({\bf M})x$ for $|t|\le t_y$, hence $t_0+t_y\in\cZ$, which is in contradiction with  the
 definition of $t_0$. This proves that $t_0=\infty$. As is well known from the elementary theory of Lie-group theory, since the
 $G$ is connected, every element is the product of a finite number of elements belonging to one parameter subgroups generated
 by $\gg$, so that  we have actually demonstrated  that $u({\bf M})U_g=U_gu({\bf M})$ on $D_N^{(U)}$ for every $g\in G$. This identity  
implies  $U_gu({\bf M})|_{D_N^{(U)}}=u({\bf M})|_{D_N^{(U)}}U_g$ on the natural domains thanks to the invariance of the Nelson space under the action of the 
group representation. In our hypotheses, $u({\bf M})|_{D_N^{(U)}}$ is  the restriction of a closable operators and thus  it is closable as 
well and so $U_g\overline{u({\bf M})|_{D_N^{(U)}}}=\overline{u({\bf M})|_{D_N^{(U)}}}U_g$ for 
every $g$. Using Proposition \ref{SL2} and Remark \ref{remarkshurrep} (a) we find $D(\overline{u({\bf M})|_{D_N^{(U)}}})=\sH$ and 
$\overline{u({\bf M})|_{D_N^{(U)}}}\in\gB(\sH)$, more precisely $\overline{u({\bf M})|_{D_N^{(U)}}}=aI+bJ$ for some $a,b\in\bR$, where 
 $J$ is some complex structure.
Since $\overline{u({\bf M})|_{D_N^{(U)}}}\subset \overline{u({\bf M})}$, the maximality of the domain gives 
$\overline{u({\bf M})|_{D_N^{(U)}}}=\overline{u({\bf M})}$. As the  latter is selfadjoint, it follows that $b=0$ and $\overline{u({\bf M})}=aI$ 
with $a\in\bR$ ending the proof of the first statement. If ${\bf M}={\bf M}^+$ then $u({\bf M})$ is symmetric for Thm \ref{thmuext} and, for the same theorem, 
 it is also essentially self-adjoint because (b) implies that ${\bf M}$ commutes with a Nelson element. 
\end{proof}

\section{Irreducible quaternionic von Neumann algebras}
This section is devoted to focus on the basic properties of irreducible quaternionic von Neumann algebras.

\subsection{The commutant of an irreducible quaternionic von Neumann algebra}
An irreducible $^*$-closed subset $\gS\subset \gB(\sH)$ in a complex Hilbert space is trivial in view of the complex version of Schur lemma. 
For real and quaternionic Hilbert spaces the picture is more complicated.
In principle, the commutant of a generic $^*$-closed irreducible subset $\gS\subset \gB(\sH)$ may contain infinitely different complex structures as suggested by Proposition \ref{SL2}. We now examine the
 quaternionic Hilbert space case  when $\gS$ is a von Neumann algebra, finding a result similar Thm 5.3 established  in \cite{MO1} dealing 
with  von Neumann algebras in real Hilbert spaces.

\begin{theorem}\label{threecommutant} Let $\gR$ be a von Neumann algebra on the quaternionic Hilbert space $\sH$. If $\gR$ is irreducible, 
then $\gR'$ is of three possible mutually exclusive types:
\begin{itemize}
\item $\gR'= \{a I\:|\: a\in \bR\}$ ({\bf quaternionic-real type}). 
\item $\gR'= \{aI + bJ\:|\:a,b \in \bR\}$  where $J$ is a complex structure determined up to its sign. Furthermore $J \in \gR$  ({\bf quaternionic-complex type}).
\item $\gR' = \{aI + bJK + cJ + dK\:|\:a,b,c, d \in \bR\}$ 
where $J,K$ and  $JK=-KJ$ are complex structures. Furthermore $J,K, JK \not \in \gR$ ({\bf quaternionic-quaternionic type})
\end{itemize}
In both the quaternionic-real and quaternionic-quaternionic cases the centre $\gZ_\gR$ of $\gR$ is $\gZ_\gR=\{aI\:|\:a\in\bR\}$, while in the quaternionic-complex case $\gZ_\gR=\{aI+bJ\:|\:a,b\in\bR\}$
\end{theorem}

\begin{proof} If $A \in \gR'$,  Prop. \ref{SL2} implies 
 $A= aI+bL$ for some $a,b \in \bR$ and some complex structure $L$. As a consequence, $\gR'$ is a real associative unital normed algebra with
  the further property that $||AB||= ||A||\:||B||$. (Indeed, by direct computation $||(aI+ bL)x||^2= (a^2+b^2) ||x||^2$ so that $||aI+bL||^2= a^2+b^2$. 
  Furthermore, iterating the procedure, where $L'$ is another complex structure,  $||(aI+bL)(a'I+b'L')x||^2 = (a^2+b^2)(a'^2+b'^2) ||x||^2 = ||aI+bL||^2\:||a'I+b'L'||^2||x||^2$ and thus
$||(aI+bL)(a'I+b'L')||= ||aI+bL||\:||a'I+b'L'||$.) Due to  \cite{UW} there exists  a  real associative unital normed algebra isomorphism $h$ from  $\gR'$ to $\bR$, $\bC$ or $\bH$. 
In the first case, $\gR'=h^{-1}(\bR)=\{aI\:|\:a\in\bR\}$.
In the second case, $\gR' =h^{-1}(\bC) = \{aI + bJ\:|\:a,b \in \bR\}$ where 
$J := h^{-1}(i)$. Furthermore, as $h^{-1}$ is an isomorphism, $JJ= h^{-1}(jj) = h^{-1}(-1) =-I$. In the third case, $\gR' = h^{-1}(\bH) = \{aI + bJK + cJ + dK\:|\:a,b,c, d \in \bR\}$ with
$J := h^{-1}(j)$,  $K := h^{-1}(k)$,  $JK := h^{-1}(i)$ where  $i,j,k \in \bH$ (with $i=jk=-kj$) are the three imaginary units. Again, as in the real-complex case, we get $JJ= h^{-1}(jj) = h^{-1}(-1) =-I$
and $KK= h^{-1}(kk)= h^{-1}(-1) =-I$. 
Let us prove that $J$ in the quaternionic-complex case and $J,K$ in the quaternionic-quaternionic one are antiselfadjoint concluding that they are complex structures. 
The proof being the same in both cases, we deal with  $J$ only. Since $\gR'$ is a $\null^*$-algebra, it holds $J^*\in\gR'$, in particular $J^*J\in\gR'$ which is clearly
 self-adjoint and positive. Since $\gR$ is irreducible, Proposition \ref{SL} guarantees that $J^*J=aI$ for some $a\ge 0$. Multiplying both sides  by $-J$ on the right, 
 using $JJ=-I$, we get $J^*=-aJ$. Taking the adjoint on both sides yields  $J=-aJ^*$ which, in particular, assures that $a\neq 0$, $J$ being bijective since $JJ=-I$.
 So, $J^*=-\frac{1}{a}J$. Summing up, $0=J^*-J^*=\left(a-\frac{1}{a}\right)J$.  As $JJ=-I$, it must be $a-\frac{1}{a}=0$, hence $a=1$ and $J^*= -aJ= -J$ as wanted. $J$ is a complex structure.
$JK$ turns out to be a complex structure as well, since $J$ and $K$ are complex structures and $JK=-KJ$.\\
To conclude,  let us establish the form of the centers $\gZ_\gR$.
The real case is obvious. In the complex case,  $J$ commutes with $ \{aI + bJ\:|\:a,b \in \bR\} =\gR'$, so it belongs to $\gR''=\gR$ and thus $\gZ_\gR= \{aI+bJ\:|\:a,b\in\bR\}$.
 This result also implies that, in the complex case, $J$ is unique up to its sign. Indeed, let $J'$ be another complex structure in $\gR'$, 
then it commutes with $J$ (as it belongs to $\gR$). Therefore $JJ' \in \gR'$  is self adjoint and thus $JJ' = aI$, namely $J'= -aJ$, because $\gR$ is irreducible. Since $JJ=J'J' =-1$ we must have $a=\pm 1$.
 $\gZ_\gR$ for the quaternionic case is easy. Suppose that $U=a+bJ+cK+dJK \in \gR$ for some $a,b,c,d\in\bR$, then, since $\gR=\gR''$ it must be $UJ=JU$, that is $aJ+bJ^2+cKJ+dJKJ=aJ+bJ^2+cJK+dJ^2K$.
 A straightforward calculcation shows that $2cJ-2d=0$ which, taking the Hermitean conjugate, is equivalent to $2d+2cJ=0$. Combining the two equations we get $d=c=0$. Finally, since it must also be 
$UK=KU$ we get $aK+bJK=aK+bKJ$ which simplifies as $bJK=bKJ$. Since $J,K$ anticommute with each other it must be $b=0$. This concludes the proof.
\end{proof}

\subsection{The structure of irreducible von Neumann algebras in quaternionc Hilbert spaces.}
We are in a position to prove that the structure of irreducible von Neumann algebras and associated lattice of orthogonal projectors is isomorphic  the one of
 $\gB(\sH')$ where $\sH'$ is suitable. This result will play a central and crucial role in the rest of this work.

\begin{theorem}\label{quaternioni3theorem}
Let $\gR$ be an irreducible von Neumann algebra in the quaternionic Hilbert space $\sH$.
Referring to the three cases listed in Theorem \ref{threecommutant},  the following facts hold.\\
{\bf (a)} If $\gR'= \{a I\:|\: a\in \bR\}$, then \\

(i) $\gR=\gB(\sH)$,\\

(ii) $\cL_\gR(\sH)=\cL(\sH)$.\\
{\bf (b)} If $\gR'= \{aI + bJ\:|\:a,b \in \bR\}$, then \\

(i)  $\gR_J = \gB(\sH_J)$ and the map $\gR \ni A \mapsto A_J \in \gB(\sH_J)$ is a norm-preserving 

weakly-continuous (thus strongly-continuous)    $^*$-isomorphism of real unital $^*$-algebras 

which, in particular, maps $J$ to $jI$.\\

(ii) the map  $\cL_{\gR}(\sH)\ni P \mapsto P_J \in \cL(\sH_J)$ is an isomorphism of complete 

orthocomplemented lattices.\\
 {\bf (c)} If $\gR' = \{aI + bJK + cJ + dK\:|\:a,b,c, d \in \bR\}$, then \\

(i) $\gR_{JK}= \gB(\sH_{JK})$ and the map $\gR \ni A \mapsto A_{JK}\in  \gB(\sH_{JK})$ is a norm-preserving  

weakly-continuous (thus strongly-continuous) $^*$-isomorphism of real unital $^*$-algebras,\\

(ii) the map  $\cL_{\gR}(\sH)\ni P \mapsto P_{JK} \in \cL(\sH_{JK})$ is an isomorphism of complete 

orthocomplemented lattices.
\end{theorem}
\begin{proof}
(a) From $\gR'=\{aI\:|\: a\in\bR\}$ it immediately follows $\gR=\gR''=\gB(\sH)$ and thus $\cL_\gR(\sH)=\cL(\sH)$. \\
(b) (i) We know from Prop.\ref{complexidentif} that $\gR$ gives rise to a complex von Neumann algebra $\gR_J$, the latter von Neumann algebra being irreducible due to Prop. \ref{complexidentif2} since the former is irreducible. The complex version of 
 Schur's lemma implies that $\gR_J = \gB(\sH_J)$. The map $\gR\ni A\mapsto A_J \in\gR_J$ is a norm-preserving weakly-continuous, strongly-continuous    $^*$-isomorphism 
 of real unital $^*$-algebras in view of Prop. \ref{propHJop}. Let us pass to (ii). Prop. \ref{propHJop} easly implies that $\cL_{\gR}(\sH)\ni P \mapsto P_J \in \cL(\sH_J)$
  is an isomorphism of orthocomplemented lattices. The only  pair of properties to be proved concerns completeness of the involved lattices of orthogonal projectors 
   and are the following ones.  (1) Given a family $\{P_a\}_{a \in A} \subset \cL_{\gR}(\sH)$ such that $\sM_a := P_a(\sH)$, defining  $P= \inf_{a \in A}P_a$ -- in other words  
   $P$   is the orthogonal projector onto $\sM:= \cap_{a\in A}\sM_a$ -- it turns out that $P_J =  \inf_{a \in A}(P_a)_J$.  
 Regarding the fact that $P_J$ is well defined,  observe that $P \in \cL_\gR(\sH)$ because  $\cL_\gR(\sH)$ is complete and thus  $P$ commutes with 
 $J$ since $\cL_\gR(\sH) \subset \gR$.  Similarly, (2)
 Given a family $\{P_a\}_{a \in A} \subset \cL_{\gR}(\sH)$ such that $M_a := P_a(\sH)$ defining  $Q= \sup_{a \in A}P_a$ -- in other words $Q \in \gR$ 
 is the orthogonal projector onto $\sN := \overline{<\cup_{a\in A}\sM_a>}$ --  it turns out that $Q_J =  \sup_{a \in A}(P_a)_J$.
 Regarding the fact that $Q_J$ is well defined,  observe that $Q \in \cL_\gR(\sH)$ because  $\cL_\gR(\sH)$ is complete and thus  $Q$ commutes with $J$  since $\cL_\gR(\sH) \subset \gR$.
 To prove (1), observe that $P_J$ is the orthogonal projector onto $\sM \cap \sH_J$ because $x=P_Jx$  if and only if both $x\in \sH_J$ (because 
 $P_J$ is a projector in $\gB(\sH_J)$) and $x=Px$ which means $x \in \sM$.   We conclude that $P_J$ is the orthogonal projector onto
  $\sM \cap \sH_J = \left(\cap_{a\in A} \sM_a\right) \cap \sH_J =\cap_{a\in A} \left(\sM_a \cap \sH_J\right)$. The orthogonal projector
   onto the last space is $\inf_{a\in A} (P_a)_J$ by definition. We have obtained $P_J =  \inf_{a \in A}(P_a)_J$.
 Property (2) is an immediate consequence of (1) and De Morgan's rule, valid for any family  orthogonal projectors  $\{Q_b\}_{b \in B} \subset \gB(\sK)$ with $\sK$ real, complex or quaternionic,
 $\sup_{b\in B} Q_b =  \left(\inf_{b\in B}Q^{\perp}_b\right)^\perp$, where $Q^\perp := I-Q$ is the orthogonal projector onto $Q(\sK)^\perp$.\\
 (c) (i) Making use of Prop.\ref{lemmaextensreal} and Propositions  \ref{quaternidentif}, \ref{complexidentif2}
 we have that  $\gR \ni A \mapsto A_{JK}\in \gR_{JK}$ is a norm-preserving  weakly-continuous, strongly-continuous $^*$-isomorphism of real
 unital $^*$-algebras where $\gR_{JK}$ is an irreducible von Neumann algebra. Let us prove that $\gR_{JK}= \gB(\sH_{JK})$.
  Let $A\in\gB(\sH_{JK})$, then Prop.\ref{lemmaextensreal} assures that there exists a unique $\widetilde{A}\in\gB(\sH)$ such
   that $A=\widetilde{A}_{JK}$. Moreover $\widetilde{A}$ commutes with $J$ and $K$. Since $\gR'=\{aI+bJ+cK+dJK\}$, it immediately
    follows that $\widetilde{A}\in\gR''=\gR$, i.e., $A\in\gR_{JK}$. This means $\gR_{JK}\supset \gB(\sH_{JK})$ and thus  $\gR_{JK} = \gB(\sH_{JK})$
     because the converse inclusion is obvious. The proof of (ii) is essentially identical to the corresponding for (b), just noticing that now we have 
     to handle two (anticommuting) complex structures $J$ and $K$. 
\end{proof}

\section{Quaternionic Wigner  relativistic elementary systems (WRES)}
Within this section, as already done in \cite{MO1},  we introduce a first notion of {\em elementary system} with respect to
the relativistic symmetry adopting the famous framework introduced by Wigner. 
A relativistic elementary system in Wigner's view is a quantum system completely determined from its symmetry properties where the symmetry group is Poincar\'e one.
So it is determined by a faithful (with the {\em caveat} discussed below) continuous unitary representation of that group and that representation also fixes the class of observables of the
 system which coincide with the selfadjoint elements 
of the von Neumann algebra  generated by the representation itself. Elementariness of the said system is translated into the irreduciblility demand of the representation.
Later we will come back on these requirements improving this model as already done in \cite{MO1}.

\subsection{Faithfulness issues}
Actually  by {\em Poincar\'{e} group} we will    
indicate here the real simply-connected Lie group obtained as  the semi-direct product
$SL(2,\bC) \ltimes \bR^4$ with respect to the abelian normal subgroup $\bR^4$.
\begin{definition}
 $\cP$, called {\bf Poincar\'{e} group} in this paper, is the semi-direct product $SL(2,\bC) \ltimes \bR^4$ of the groups $SL(2,\bC)$ and $\bR^4$, with group product defined as  
$$(A, t)\cdot (A',t'):= (AA', t+\Lambda(A)t') \quad \mbox{for $A,A' \in SL(2,\bC)$ and $t,t' \in \bR^4$,}$$ 
where $\Lambda : SL(2,\bC) \to SO(1,3)\sp\uparrow$ is the standard covering homomorphism.
\end{definition}
\noindent  $\cP$ is more properly  known as the {\em universal covering}
 of the {\em proper orthochronous Poincar\'{e} group} as understood in relativity. $\cP$   actually enters all known  physical constructions and every representation 
 of the  proper orthochronous Poincar\'{e} group
 is also a representation of $\cP$. For this reason, we require a weaker {\em local-faithfulness} assumption for every  group representation of $\cP$
 we henceforth consider, i.e., {\em the representation is only supposed to be  injective  in a 
 neighbourhood of the neutral element of the group}, 
  because  only in a neighbourhood  
 of the identity  of the group elements 
$SL(2,\bC)$ and the proper orthochronous Lorentz groups $SO(1,3)\sp\uparrow$ are identical. \\ 
To corroborate our assumption, we remark that the {\em complex} strongly-continuous unitary 
irreducible representations of 
$SL(2,\bC) \ltimes \bR^4$ with physical meaning are all locally faithful: (1)
for positive squared mass with semi-integer spin they are faithful,  (2) for positive squared mass with integer spin they are faithful up to the sign of the $SL(2,\bC)$
 element, and (3) they are faithful up to the sign of the $SL(2,\bC)$ element for zero squared mass with non-trivial momentum representation. \\
From a general point of view, local faithfulness of a continuous representation
 $\cP \ni g \mapsto U_g$ is equivalent to the physically more natural requirement that 

{\em the representation of the subgroup $\bR^4$ of spacetime translations is non-trivial}.\\ More precisely, the following general result holds  whose proof stays in Appendix B.
\begin{proposition}\label{proplocfaithfulness} 
Let $f : SL(2,\bC) \ltimes \bR^4 \to G$ be a continuous group homomorphism to the topological Hausdorff group $G$. The following two facts are equivalent.

 {\bf (a)} $f$ is injective in a neighborhood of the unit element of $SL(2,\bC) \ltimes \bR^4$.

 {\bf (b)} The associated group homomorphism from the subgroup $\bR^4 \subset  SL(2,\bC) \ltimes \bR^4$ to $G$ $$\bR^4 \ni t \mapsto f((I,t))$$ 

is not trivial.\\
If (a) and (b) are true, then either $f$ is  injective or its kernel is  $\{(\pm I, 0)\}$.
\end{proposition}
\noindent Due to Proposition \ref{proplocfaithfulness}, when $G$ is the group of unitary operators in the (real, complex or quaternionic) Hilbert space  $\sH$ equipped with the (Hausdorff) strong operator 
topology, a unitary strongly-continuous representation $\cP \ni g \mapsto U_g \in G$ is locally faithful if and only if the associated representation of the group of spacetime translations is non-trivial. If $U$ 
is not faithfull the failure of injectivity is due to the sign of the elements in $SL(2,\bC)$ only.

\subsection{Wigner's approach extend to the quaternionic case}
\begin{definition}
{\em
A  real, complex, quaternionic {\bf Wigner  elementary  relativistic system} (WRES)  is a unitary strongly continuous  representation of the proper orthochronous Poincar\'{e} group $\cP$,
$$
U:\cP\ni g\mapsto U_g\in\gB(\sH)
$$
over the, respectively,  real, complex, or  quaternionic  separable Hilbert space $\sH$, and the representation  is  {\em locally faithful} (i.e., injective in a neighborhood of the unit element) and {\em irreducible}. \\ 
If $\gR_U$ is the von Neumann algebra generated by $U$, the {\bf observables} of the WRES
 are the selfadjoint operators $A$  {\bf affiliated} to $\gR_U$,  i.e., their  PVMs belong to $\cL_{\gR_U}(\sH)$.\\
$\cL_{\gR_U}(\sH)$ itself  is the set of {\bf elementary (YES-NO) observables} of the WRES.}
\end{definition}
\begin{remark}\label{remobservables}
{\em Every essentially self-adjoint operator of the form $u({\bf M})$ for some ${\bf M} \in E_{\gp}$ defines an observable 
$\overline{u({\bf M})}$ of the WRES. Indeed, every $u({\bf M})$ satisfies
$Bu({\bf M}) \subset u({\bf M}) B$ and thus $B\overline{u({\bf M})}\subset \overline{u({\bf M})}B$ for every $B \in \{U_g\}_{g \in \cP}' = \gR_U'$ as the 
reader can easily prove from Theorem \ref{thmuext}, so that the PVM of the selfadjoint operator $\overline{u({\bf M})}$
 commutes with $B$ due to Lemma \ref{commst} and therefore  it belongs to $\gR_U''= \gR_U$ as required for observables. We will generalize this example in Corollary \ref{CorollaryMT}.}
\end{remark}
\noindent As established in the previous section for the quaternionic case and in  \cite{MO1} for the real and complex case, the representation $U$ gives rise to a representation on
 its G\r{a}rding domain of the corresponding Lie algebra $\gp$
$$
u: \gp\ni {\bf A}\mapsto u({\bf A})\:.
$$
Let us fix a Minkowskian reference frame in Minkowski spacetime. Referring to that reference frame, for $i=1,2,3$,  ${\bf k}_i\in \gp$ denote the three generators of the 
\textit{boost} one-parameter subgroups along the three spatial axes, ${\bf l}_i\in\gp$ are the three generators of the \textit{spatial rotations} one-parameter subgroups
 about the three axes. Finally, for $\mu=0,1,2,3$ let ${\bf p}_\mu\in\gp$ be the four generators of the \textit{spacetime displacements} one-parameter subgroups along the four Minkowskian axes.
Next define the associated anti-selfadjoint generators
\begin{equation} \label{basispoinc}
\widetilde{K}_i:=\overline{u({\bf k}_i)},\ \widetilde{L}_i:=\overline{u({\bf l}_i)},\ \widetilde{P}_0:=\overline{u({\bf p}_0)},\ \widetilde{P}_i:=\overline{u({\bf p}_i)}\ \ i=1,2,3\:.
\end{equation}
All of these operators leave fixed  the G\r{a}rding domain $D_G^{(U)}$ which is a common core for all them.
If we change the initially fixed reference frame by means of a transformation $p \in \cP$, we obtain another set of generators related to the previous ones by the natural relations 
\begin{equation} \label{basispoinc2}
\widetilde{K}'_i=U_p \widetilde{K}_i U^{-1}_p ,\ \widetilde{L}'_i=U_p\widetilde{L}_i U_p^{-1},\ \widetilde{P}'_0 =U_p\widetilde{P}_0 U_p^{-1},\ \widetilde{P}'_i=U_p\widetilde{P}_i U_p^{-1}\ \ i=1,2,3\:.
\end{equation}
The said generators are {\em not} observables, since they are not selfadjoint and there is no trivial way to associate them with selfadjoint 
operators as instead it 
happens within the complex Hilbert space formulation where one has a standard imaginary unit at disposal. In turn, this problem makes it 
difficult to state a quantum version of Noether correspondence between generators of continuous symmetries and dynamically conserved 
quantities. For the moment we simply ignore these open issues. 

\noindent Consider the element ${\bf e}:=-{\bf p}_0^2+\sum_{i=1}^3{\bf p}^2_i\in E_\gg$ and define  the \textbf{Mass operator} as the symmetric operator
\begin{equation} \label{Moperator}
M_U^2:= u({\bf e})=\left.\left(-\widetilde{P}_0^2+\sum_{i=1}^3\widetilde{P}_i^2\right)\right|_{D_G^{(U)}}\:.
\end{equation}
By direct inspection one sees that ${\bf e}$ is symmetrice and commutes with all generators and thus,
for Proposition \ref{propCSAI}, $M_U^2$
 is essentially selfadjoint -- thus represents an observable of the WRES in view of Remark \ref{remobservables} -- and
has a trivial form $cI$ for some constant $c \in \bR$. 
Moreover it holds $U_p M_U^2 U_p^{-1} = M_U^2$ for every $p \in \cP$ and, in this sense, 
$M_U^2$ is {\bf Poincar\'e invariant}.
Indeed, exploiting the definition of Poincar\'{e} group and Stone Theorem, it is easy to see that $U_g \tilde{P}_{\mu}U_g^{-1}|_{D_G^{(U)}}=\sum_{\nu=0}^3 \Lambda_\mu^\nu \tilde{P}_\nu|_{D_G^{(U)}}$ where $\Lambda$ 
is the Lorentz matrix corresponding to the element $g\in\cP$ and this immediately yields the claim. \\
The real and complex cases were treated in \cite{MO1}, we now focus attention on the remaining quaternionic case taking advantage of the results already achieved in \cite{MO1}. 
A first technical result is the following one.
\begin{proposition}\label{WignerTHM}
Let $U:\mathcal{P}\ni g\mapsto U_g\in\gB(\sH)$ be a  WRES over the quaternionic Hilbert space $\sH$.
 If $\widetilde{P}_0=J_0|\widetilde{P}_0|$ is the polar decomposition of the time displacement generator and $M_U^2\ge 0$, then $J_0$ is a complex structure in $\sH$ with $J_0\in \gR_U\cap\gR_U'$.
\end{proposition}
\begin{proof}
 We need some preliminary results stated into some lemmata.
\begin{lemma}\label{lemma-a}
With the hypotheses of Proposition \ref{WignerTHM}, $U$ defines reducible (strongly-continuous locally-faithful) unitary representations both over the real Hilbert space
  $\sH_\bR$ and the complex Hilbert space $\sH_{\bC_j}$ as defined in Prop.\ref{propHR} and Prop.\ref{propHC} respectively.
\end{lemma}
\begin{proof}
 Suppose $U: \cP \ni g\mapsto U_g$ defines a WRES on a quaternionic Hilbert space $\sH$. It is clear that $U$ is a strongly-continuous locally-faithful unitary representation over 
 $\sH_\bR$ and  $\sH_{\bC_j}$ as defined in Prop.\ref{propHR} and Prop.\ref{propHC} respectively and defines two corresponding representations 
 called $U_\bR$ and $U_{\bC_j}$. Let us prove that the representations $U_{\bC_j}$ on $\sH_{\bC_j}$ cannot be irreducible.  First, remember that 
 $D_G^{(U)}=D_G^{(U_\bR)}=D_G^{(U_{\bC_j})}$ (Prop. \ref{propgarding}) and the anti-selfadjoint generators of the one-parameter subgroups of spacetime
  displacements defined with respect to $\sH$ or $\sH_\bR$ or $\sH_{\bC_j}$ are equal to each other (Prop. \ref{algebrarepr}). Hence, $M_U^2=M_{U_\bR}^2 = M_{U_{\bC_j}}^2$ 
   and, it being  symmetric, $M_U^2$ is  positive on $\sH$ if and only if  it is positive  on $\sH_\bR$ and, in turn, it happens if and only if it is positive  on $\sH_{\bC_j}$. Thanks to 
   Theorem \ref{PDT}, the polar decompositions $J_0|\widetilde{P}_0|$ of $\widetilde{P}_0$ defined with respect to the three considered spaces are identical. 
If the representation $U_{\bC_j}$ were irreducible,  due to (g) in  Theorem 4.3 in \cite{MO1},  we would have  $J_0=\pm \cJ$,  $\cJ$ being the {\em imaginary unit} of $\sH_{\bC_j}$.
 This is  impossible, since $J_0$ is quaternionic linear and so it commutes with $\cK$ if understood as a real operator on $\sH_\bR$, while we know that $\cJ\cK=-\cK\cJ$. This also 
 implies that $U_\bR$ is reducible: If it were not the case, then $U_{\bC_j}$ would be irreducible since  $\sH_{\bC_j} = (\sH_\bR)_\cJ$
where the right-hand side is the internal complexification of $\sH_\bR$ referred to the complex structure $\cJ$ as in Sect.2.5  in \cite{MO1}.\\
\end{proof}

\begin{lemma}\label{lemmaA}
With the hypotheses of Proposition \ref{WignerTHM}, there exists a projector $P\in\cL(\sH_{\bC_j})$ such that the orthogonal decomposition  $$\sH_{\bC_j}= \sH_{P} \oplus \sH_{P^\perp}$$ together with the following statements.\\

{\bf (i)} $\cJ(\sH_P) = \sH_{P}$ and $\cJ(\sH_{P^\perp}) = \sH_{P^\perp}$ while $\cK(\sH_P) = \sH_{P^\perp}$ and $\cK(\sH_{P^\perp}) = \sH_{P}$,\\

{\bf (ii)}  $\cK : (u,v) \mapsto (-A^{-1}v, Au)$ for every $u\in \sH_P$ and $v\in \sH_{P^\perp}$ where the map $A:= \cK|_{\sH_{P}}:  \sH_{P} \to  \sH_{P^\perp}$ is $\bR$-real isometric and surjective with inverse $A^{-1}=-\cK|_{\sH_{P^\perp}}$\\

{\bf (iii)} Each complex subspace  $\sH_{P}$ and $\sH_{P^\perp}$ is separately invariant under the action of $U$ and the maps $\cP \ni g \mapsto U|_{ \sH_{P}}$, $\cP \ni g \mapsto U|_{ \sH_{P}}$ 
are strongly-continuous irreducible locally-faithful unitary complex representations such that
\beq \label{intpAU}AU_g|_{\sH_P}=U_g|_{\sH_{P^\perp}}A \quad \forall g \in \cP\:.\eeq 
\end{lemma}
\begin{proof} Since $U_{\bC_j}$ is reducible as established in Lemma \ref{lemma-a}, there must exist an orthogonal projector  $P\in\gB(\sH_{\bC_j})\setminus\{0,I\}$ such that
 $PU_g=U_gP$ for every $g\in\cP$. Interpreting $P$ as an $\bR$-linear operator in $\gB(\sH_\bR)$, define the self-adjoint operator $P_\cK:=P-\cK P\cK\in\gB(\sH_\bR)$,
 it is easy to prove that this 
operator commutes with both $\cJ$ and $\cK$, i.e., it is $\bH$-linear,  and so $P_\cK\in\gB(\sH)$. Furthermore  $P_\cK U_g=U_gP_\cK$
 for all $g\in\cP$, because $U_g$ is quaternionic linear  and therefore commutes with $\cJ$ and $\cK$ and also $PU_g=U_gP$ by hypothesis. 
 Since $U$ is irreducible and $P_\cK$ is selfadjoint, Proposition \ref{SL} implies 
  $P_\cK=aI$ for some $a\in\bR$, which, multiplying by $\cK$ on the left, can be restated as   $\cK P+P\cK=a\cK$, or $\cK P=-P\cK+a\cK$. Since $PP=P$ this can
   be rewritten as $\cK P=-P\cK P+a\cK P$. Taking the adjoint on both the sides of this identity  produces  $-P\cK=P \cK P-aP\cK$. Both identities 
together yield    $[\cK,P]=a[\cK,P]$. The case  $a\neq 1$ is not permitted because  it implies  $[P,\cK]=0$, which, together with $[P,\cJ]=0$ ($P$ is complex
    linear by hypothesis), ensures that $P$ is actually quaternionic linear. Since  $U$ is  irreducible on $\sH$, we would get $P=0,I$, which was excluded {\em a priori}. 
It remains the case  $a=1$, that is $P-\cK P\cK=I$. 
Since $P\neq 0, I$ is an orthogonal projector on $\sH_{\bC_j}$, $\sH_{\bC_j}$ can be decomposed into an orthogonal direct sum $\sH_{\bC_j} = \sH_P\oplus \sH_{P^\perp}$, where
 $P^\perp:=I-P$ and both complex subspaces are nontrivial. Evidently $\cJ(\sH_P) = \sH_{P}$ and $\cJ(\sH_{P^\perp}) = \sH_{P^\perp}$ because each space are complex subspaces
  of $\cH_{\bC_j}$ and $\cJ$ is the complex structure used to construct $\cH_{\bC_j}$ out of $\sH$.
 Let us study the interplay of $\cK$ and that decomposition proving (i) and (ii).
 Multiplying $P-\cK P\cK=I$ by $\cK$ on the left, we have $\cK P=(I-P)\cK$, while a multiplication on the right yields $P\cK=\cK(I-P)$. 
  Notice that $\cK(\sH_P)=\cK P(\sH_{\bC_j})=P^\perp \cK(\sH_{\bC_j})=P^\perp(\sH_{\bC_j})=\sH_{P^\perp}$ and, 
  similarly, $\cK(\sH_{P^\perp})=\cK P^\perp (\sH_{\bC_j})=P\cK(\sH_{\bC_j})=P(\sH_{\bC_j})=\sH_P$. Hence, also 
  exploiting the facts that $\cK$ is isometric and that $\cK(-\cK)=I$,  the map $A:\sH_P\rightarrow \sH_{P^\perp}$ defined as $A:=\cK|_{\sH_P}$ is 
  a well defined bijective $\bR$-linear isometry with inverse given by the isometry $A^{-1}=-\cK|_{\sH_{P^\perp}}$. The action of $\cK$ can be written
  $\cK(u,v)=(-A^{-1}v,Au)$ for every $u\in\sH_P$ and $v\in\sH_{P^\perp}$.
  referring to the direct decomposition $\sH_{\bC_j} = \sH_P\oplus \sH_{P^\perp}$
  Indeed, if $x\in\sH_{\bC_j}$, then $\cK x=K(Px+P^\perp x)= \cK Px + \cK P^\perp x = A(Px)-A^{-1}(P^\perp x)$, where $A(Px)\in\sH_{P^\perp}$ and $-A^{-1}(P^\perp x)\in\sH_P$. 
  Let us pass to establish (iii).
Notice that $U_gP=PU_g$ ensures that $U_g(\sH_P)\subset \sH_P$ and $U_g(\sH_{P^\perp})\subset \sH_{P^\perp}$. Moreover the restrictions 
of $U_g$ to both the factors are clearly bijective isometries on the respective domains, thanks to the facts that $g$ is arbitrary, every $U_g$ is
isometric and $U_{g^{-1}}U_g=I$, so we can refine the inclusions above and write $U_g(\sH_P)=\sH_P$ and $U_g(\sH_{P^\perp})=\sH_{P^\perp}$. 
Notice that, since $U_gx=U_g(Px+P^\perp x)=U_g(Px)+U_g(P^\perp x)$, the action of $U$ on
$\sH_{\bC_j} = \sH_P\oplus \sH_{P^\perp}$   can be written $U_g(u,v)=(U_gu,U_gv)$ for every $u\in\sH_P$ and $v\in\sH_{P^\perp}$. We end up with
 two representations $g\mapsto U_g|_{\sH_P}$ and $g\mapsto U_g|_{\sH_{P^\perp}}$ which are clearly unitary and strongly-continuous on the respective 
 complex Hilbert spaces. Finally we want to prove (\ref{intpAU}).
Let $(u,v)\in\sH_{\bC_j}$,
 using the above identities we have
\begin{equation*}
\begin{split}
&\cK U_g(u,v)=\cK(U_g|_{\sH_P}u,U_g|_{\sH_{P^\perp}}v)=(-A^{-1}U_g|_{\sH_{P^\perp}}v,AU_g|_{\sH_P}u)\\
&U_g\cK(u,v)=U_g(-A^{-1}v,Au)=(-U_g|_{\sH_P}A^{-1}v,U_g|_{\sH_{P^\perp}}Au)
\end{split}
\end{equation*}
Since $\cK U_g=U_g\cK$ and $(u,v)$ is arbitrary, we get the thesis.\\
To conclude the proof of (iii), we intend to prove that $\cP \ni g\mapsto U_g|_{\sH_P}$ and $\cP \ni g\mapsto U_g|_{\sH_{P^\perp}}$ are irreducible
 representations  on the respective complex Hilbert spaces which also are locally faithful. We will establish  the first result  for $\sH_P$ only, the other case 
   being analogous. Suppose that $g\mapsto U_g|_{\sH_P}$ is not irreducible, i.e., there exists a projector $0\le Q\le P$ with $Q\neq 0,P$
    such that $QU_g|_{\sH_P}=U_g|_{\sH_P}Q$ for every $g\in\cP$. Since we are thinking of $Q$ as a projector defined on the whole $\sH_{\bC_j}$, the found identity  
    can be rephrased as $QU_g=U_gQ$. Indeed, on $\sH_P$ that identity  reduces to $QU_g|_{\sH_P}=U_g|_{\sH_P}Q$, while    both sides of  $QU_g=U_gQ$ vanish on $\sH_{P^\perp}$.
So we can repeat the analysis carried out so far using $Q$ in place  of $P$ and finding  $\sH=\sH_Q\oplus\sH_{Q^\perp}$ with $\cK(\sH_Q)=\sH_{Q^\perp}$. However, 
 this result  implies $\sH_{Q^\perp}=\cK(\sH_Q)\subset \cK(\sH_P)=\sH_{P^\perp}$ which can be restated as $Q^\perp\le P^\perp$, i.e., $P\le Q$, which was excluded {\em a priori}. 
 We conclude that $g\mapsto U_g|_{\sH_P}$ (and also $g\mapsto U_g|_{\sH_{P^\perp}}$) is irreducible. To conclude, we want eventually  to demonstrate  that the complex 
 strongly-continuous irreducible representations of  unitary operators $\cP \ni g\mapsto U_g|_{\sH_P}$ and $\cP \ni g\mapsto U_g|_{\sH_{P^\perp}}$ are also locally 
 faithful because  $\cP \ni g\mapsto U_g$ is.  Let $W_e\subset\cP$ be a  neighbourhood of the identity of $\cP$ where $\cP \ni g\mapsto U_g$ is faithful and suppose, 
 for example, that $U_g|_{\sH_P}=U_h|_{\sH_P}$ for some $g,h\in W_e$. Exploiting (\ref{intpAU}), we get $U_g|_{\sH_{P^\perp}}=AU_g|_{\sH_{P}}A^{-1}=AU_h|_{\sH_{P}}A^{-1}=U_h|_{\sH_{P^\perp}}$ 
 and so, with obious notation, $U_g=U_g|_{\sH_{P}}\oplus U_g|_{\sH_{P^\perp}}=U_h|_{\sH_{P}}\oplus U_h|_{\sH_{P^\perp}}=U_h$. Since $U$ is faithful on $W_e$ we have that $g=h$, proving 
 the local faithfulness of the restricted representations. 
\end{proof}
\noindent A subsequent lemma is in order.
\begin{lemma} Assume the hypotheses of Proposition \ref{WignerTHM} and refer to the
orthogonal decomposition $\sH = \sH_P \oplus \sH_{P^\perp}$ and 
 two strongly-continuous irreducible locally-faithful unitary complex representations  $\cP \ni g \mapsto U|_{ \sH_{P}}$, $\cP \ni g \mapsto U|_{ \sH_{P^\perp}}$ as in (iii) of Lemma \ref{lemmaA}.
  The following facts hold.\\
  
{\bf (iv)} The partial isometry $J_0$ of the polar decomposition  $\tilde{P}_0=J_0|\tilde{P}_0|$ of the  ($\bH$-linear) anti-selfadjoint  generator $\tilde{P}_0$ in $\sH$ of  time-displacements subgroup satisfies 
\beq
J_0 = J_P \oplus J_{P^\perp}
\eeq
where  $J_P$ and $J_{P^\perp}$ are the analogous operators for the pair of complex representations  $\cP \ni g \mapsto U_g|_{ \sH_{P}}$ and $\cP \ni g \mapsto U_g|_{ \sH_{P^\perp}}$. \\
%
%{\bf (v)} If $A: \sH_{P} \to \sH_{P^\perp}$ is the isometric surjective map as in Lemma \ref{lemmaA}, it holds
%$J_{P^\perp}=AJ_PA^{-1}$.\\
%

{\bf (v)} $J_0$ is a complex structure in the quaternionic Hilbert space $\sH$ commuting with the whole representation $U$ and thus $J_0 \in \gR_U'$ in particular.
\end{lemma}
\begin{proof} We simultaneously prove (iv) and (v). First of all we study the complex representations   $\cP \ni g \mapsto U|_{ \sH_{P}}$, $\cP \ni g \mapsto U|_{ \sH_{P^\perp}}$ focussing 
on  their anti-selfadjoint generators and  their mass operators  $M_{U_P}^2$ and $M_{U_{P^\perp}}^2$.
Let $\widetilde{P}_\mu$ the antiselfadjoint generator of $t\mapsto U_{\exp(t{\bf p}_\mu)}$ defined on $\sH$. As we know from the version of Stone theorem presented in 
Theorem \ref{stonetheorem},  this is in particular a well-defined linear operator over $\sH_\bR$ and $\sH_{\bC_j}$ and it coincides with the corresponding generators of
 $t\mapsto U_{\exp(t{\bf p}_\mu)}$ when reading the representation  on $\sH_\bR$ and $\sH_{\bC_j}$, respectively.
Due to Stone theorem,  $x\in D(\tilde{P}_\mu)$ if and only if there exists $\frac{d}{dt}\big|_0e^{t\tilde{P}_\mu}x$. As $e^{t\tilde{P}_\mu} (u,v)=(e^{t\tilde{P}_\mu}u,e^{t\tilde{P}_\mu}v)$ in view 
of the    orthogonal decomposition into complex subspaces $\sH \equiv \sH_{\bC_j} = \sH_P \oplus \sH_{P^\perp}$, we conclude for $x=(u,v)$ that 
\begin{equation*}
\exists \frac{d}{dt}\big|_0e^{t\widetilde{P}_\mu}x\in\sH\ \ \mbox{ if and only if }
\begin{cases}
&\exists \dfrac{d}{dt}\big|_0e^{t\widetilde{P}_\mu}u\in\sH_P  \\
& \qquad\qquad and \\
&\exists \dfrac{d}{dt}\big|_0e^{t\widetilde{P}_\mu}v\in\sH_{P^\perp}\:.
\end{cases}
\end{equation*}
As an immediate consequence,
 $P(D(\widetilde{P}_\mu))=D(\widetilde{P}_\mu)\cap \sH_P\subset D(\widetilde{P}_\mu)$ and $P^\perp(D(\widetilde{P}_\mu))=D(\widetilde{P}_\mu)\cap \sH_{P^\perp}\subset D(\widetilde{P}_\mu)$, 
 which in turn yields $D(\widetilde{P}_\mu)=P(D(\widetilde{P}_\mu))\oplus P^\perp(D(\widetilde{P}_\mu))$.
The generators of the one-parameter unitary subgroups $e^{t\widetilde{P}_\mu}|_{\sH_{P}}$ and $e^{t\widetilde{P}_\mu}|_{\sH_{P^\perp}}$ are clearly 
given by $\widetilde{P}_\mu|_{P(D(\widetilde{P}_\mu))}$ and $\widetilde{P}_\mu|_{P^\perp(D(\widetilde{P}_\mu))}$, respectively.
From $e^{tP_\mu}x=e^{t\widetilde{P}_\mu}u+e^{t\widetilde{P}_\mu}v$ and the above equivalence we have that $\widetilde{P}_\mu x=\widetilde{P}_\mu u + \widetilde{P}_\mu v$, that 
is $\widetilde{P}_\mu=\widetilde{P}_\mu|_{P(D(\widetilde{P}_\mu))}\oplus \widetilde{P}_\mu|_{P^\perp(D(\widetilde{P}_\mu))}$.
(In the following remember that, as already observed, $D_G^{(U)}=D_G^{(U_\bR)}=D_G^{(U_{\bC_j})}$ and  $0 \leq M_U^2=M_{U_\bR}^2= M^2_{U_{\bC_j}}$).
Next consider the G\r{a}rding domains defined with respect to the representations $U|_{\sH_P}$ and $U_{\sH_{P^\perp}}$ (hence subspaces of $\sH_P$ and $\sH_{P^\perp}$,
 respectively): $D_G^{(U_P)}$ and $D_G^{(U_{P^\perp})}$. We want to prove that they are both subspaces of $D_G^{(U)}$ and that $D_G^{(U)}=D_G^{(U_{P})}\oplus D_G^{(U_{P^\perp})} $ when $D_G^{(U)}$ is understood as a complex vector space.
The former property easily follows from the direct decomposition $\sH_{\bC_j}=\sH_P\oplus\sH_{P^\perp}$. Indeed a direct calculation shows that the definitions of the Garding vectors $x[f]$
 within $\sH_P (\sH_{P^\perp})$ or within $\sH_{\bC_j}$ coincide for $x\in\sH_P (\sH_{P^\perp})$. The equality $D_G^{(U_{\bC_j})}=D_G^{(U)}$   concludes the proof of this part.
To prove the latter, just notice that if $x[f]\in D_G^{(U)}$ with $x=(u,v)$, then $x[f]=u[f]+v[f]$ where $u[f]\in D_G^{(U_P)}$ and $v[f]\in D_G^{(U_{P^\perp})}$.
So, if we consider the  operators $M_{U_P}^2$ and $M_{U_{P^\perp}}^2$, defined respectively on $D_G^{(U_P)}$ and $D_G^{(U_{P^\perp})}$, using the above-proved
 decompositions for $\widetilde{P}_\mu$ we easily get $M_U^2=M_{U_P}^2\oplus M_{U_{P^\perp}}^2$. Finally,  $M_U^2\ge 0$ easily implies that $M_{U_P}^2\ge 0$
  and $M_{U_{P^\perp}}^2\ge 0$. Summarizing,  we end up with  two irreducible locally faithful strongly-continuous unitary complex representations of $\cP$: $U_P$
   on $\sH_P$ and $U_{P^\perp}$ on $\sH_{P^\perp}$ 
such that $M_{U_P}^2\ge 0$ and $M_{U_{P^\perp}}^2\ge 0$.
The respective  one-parameter subgroups of time displacement $t\mapsto U_{\exp{(t{\bf p}_0})}|_{\sH_P}$ and $t\mapsto U_{\exp{(t{\bf p}_0)}}|_{\sH_{P^\perp}}$ 
have generators $\widetilde{P}_0|_{P(D(\widetilde{P}_0))}$ and $\widetilde{P}_0|_{P^\perp(D(\widetilde{P}_0))}$, respectively. For each of them the polar decomposition
 theorem applies and gives $\widetilde{P}_0|_{P(D(\widetilde{P}_0))}=J_PX_P$ and $\widetilde{P}_0|_{P^\perp(D(\widetilde{P}_0))}=J_{P^\perp}X_{P^\perp}$. We notice for
  future convenience  that $\widetilde{P}_0 = (J_P \oplus J_{P^\perp}) (X_P \oplus X_{P^\perp})$ is the polar decomposition of $\widetilde{P}_0$ in $\sH_{\bC_j}$
as a consequence of the uniqueness property of the polar decomposition, the proof being elementary.
 We are in a position to apply (g) of Theorem 4.3 in \cite{MO1} establishing that  $J_P=\pm \cJ|_{\sH_P}$ and $J_{P^\perp}=\pm \cJ|_{\sH_{P^\perp}}$. Notice that, trivially,
  $\cJ|_{\sH_P}$ and $\cJ|_{\sH_{P^\perp}}$ are unitary operators 
 and they commute with the restrictions of $U_g$ to the respective subspaces.
In particular $J_P\oplus J_{P^\perp}$ commutes with $U_g=U_g|_{\sH_{P}}\oplus U_g|_{\sH_{P^\perp}}$. 
From  $J_P= \pm\cJ|_{\sH_P}$ and $J_{P^\perp}=\pm\cJ|_{\sH_{P^\perp}}$ we also have that $J_P\oplus J_{P^\perp}$ is an isometry such that $(J_P\oplus J_{P^\perp})^2=-I$: it  is 
a ($\bC$-linear) complex structure on $\sH_{\bC_j}$. 
%Furthermore, from $U_g|_{\sH_{P^\perp}}=AU_g|_{\sH_P}A^{-1}$ proved in the previous lemma, we obtain  in particular that
%$e^{t\widetilde{P}_0}|_{\sH_{P^\perp}}=Ae^{t\widetilde{P}_0}|_{\sH_P}A^{-1}$. Since $A$ is an isometric bijection, Stone theorem assures that $\widetilde{P}_0|_{P^\perp(D(\widetilde{P}_0))}=A\widetilde{P}_0|_{P(D(\widetilde{P}_0))}A^{-1}$. From
%$
%J_{P^\perp}X_{P^\perp}=\widetilde{P}_0|_{P^\perp(D(\widetilde{P}_0))}=A\widetilde{P}_0|_{P(D(\widetilde{P}_0))}A^{-1}=(AJ_PA^{-1})( AX_PA^{-1})
%$, 
%uniqueness of the polar decomposition assures that $J_{P^\perp}=AJ_PA^{-1}$ in particular. 
 Thinking of $\widetilde{P}_0$ as the generator of time displacements in $\sH$, since the polar decomposition is unique in $\sH$ and $\sH_{\bC_j}$ as established in (c) Theorem \ref{PDT}, 
$\widetilde{P}_0 = (J_P \oplus J_{P^\perp}) (X_P \oplus X_{P^\perp})$ is also the polar decomposition of $\widetilde{P}_0 = J_0 |\widetilde{P}_0|$ in $\sH$. Uniqueness 
entails  $J_0 = J_P \oplus J_{P^\perp}$.  In particular $J_P\oplus J_{P^\perp}$ must also be quaternionic linear because $J_0$ is.  Since  $J_P\oplus J_{P^\perp}$
satisfies $(J_P\oplus J_{P^\perp})(J_P\oplus J_{P^\perp}) = -I$
and $(J_P\oplus J_{P^\perp})^* = - (J_P\oplus J_{P^\perp})$ and the notions of adjoint in $\sH$ and $\sH_{\bC_j}$ coincide, $J_0$ is a complex structure in $\sH$. Moreover, 
since $J_P\oplus J_{P^\perp}$ commutes with $U_g=U_g|_{\sH_{P}}\oplus U_g|_{\sH_{P^\perp}}$,
 $J_0$  commutes with the whole representation $U$, i.e., $J_0\in\{U_g\:|\:g\in\cP\}' = \{U_g\:|\:g\in\cP\}'''=  \gR_U'$ so that $J_0 \in \gR_U'$.
\end{proof}
\noindent To conclude the proof of the main statement it is sufficient to prove that $J_0\in\gR_U$.
Let $B\in\gR_U'$, then in particular $Be^{t\widetilde{P}_0}=e^{t\widetilde{P}_0}B$ for every $t\in\bR$. Thanks to Proposition \ref{LemmaCOMM}, we immediately 
get $BJ_0=J_0B$. This means that $J_0\in\gR_U''=\gR_U$.
\end{proof}

\begin{remark}\label{remarkequiv}
	The reader interested in an alternative proof of Proposition \ref{WignerTHM} may consult  Theorem 9.2.12 of \cite{TesiMarco}, where the real case (see \cite{MO1}) and the current quaternionic setting are treated on an equal footing.
\end{remark}

\subsection{Structure of quaternionic a WRES and equivalence with a complex WRES}
 We are now in a position to state and prove the first main result of this work, establishing in particular that a quaternionic  WRES 
$\cP \ni p \mapsto U_p$ over the quaternionic Hilbert space $\sH$ is completely equivalent to a complex Hilbert space WRES
$\cP \ni p \mapsto U_p|_{\sH_{J_0}}$
 on a suitable complex Hilbert space $\sH_{J_0}$ constructed out of a complex structure $J_0$ in $\sH$
 according to  (b) Theorem \ref{quaternioni3theorem}.
Within this equivalent formulation, everything  is in agreement with the thesis of Sol\`er theorem,  the complex structure $J_0$
is related to the polar decomposition of the generator of time displacements and,
 independently from its physical meaning,  is unique up its sign and Poincar\'e invariant, i.e., 
 starting from another initial Minkowski frame to describe the symetries of Poincar\'e group we 
 would obtain the same equivalence of quaternionic-complex structures. The standard 
 Noether correspondence of {\em selfadjoint} generators of continuous symmetries and dynamically conserved
  quantities cen be stated also for the quaternionic case using $J_0$ as imaginary unit.

\begin{theorem}\label{poinccomplexstructure}
Consider a {\em quaternionic Wigner Relativistic Elementary System}  and adopt definitions (\ref{basispoinc}) and (\ref{Moperator}) (with respect to a given Minkowski reference frame).\\
Let $\widetilde{P}_0=J_0|\widetilde{P}_0|$ be the  polar decomposition of the anti selfadjoint generator  of the subgroup of temporal displacements. 
The following facts hold provided $M^2_U \geq 0$.\\

\noindent {\bf (a)} $J_0\in \gR_U$ and $J_0$ is a complex structure on $\sH$. \\

\noindent {\bf (b)} $J_0 \in  \gR'_U$,  in particular    $J_0U_g =U_gJ_0$ for all $g \in \cP$ and so  $J_0$ is Poincar\'e invariant (starting 
from another Minkowski reference frame to define the generators of the one parameter-subgroups of $\cP$ we would obtain the same $J_0$).\\

\noindent {\bf (c)}    $J_0u({\bf A}) =  u({\bf A})J_0$ for every ${\bf A}\in \gp$ in particular $J_0(D_G^{(U)})  = D_G^{(U)}$.\\

\noindent {\bf (d)}  If  $J_1$ is  a  complex structure on $\sH$ such that either $J_1 \in \gR_U'$
or $J_1 u({\bf A}) = u({\bf A}) J_1$ for every ${\bf A}\in\gp$ are valid, then $J_1=\pm J_0$. \\

\noindent {\bf (e)}   If ${\bf A} \in \gp$, then  $J_0\overline{u({\bf A})}=\overline{u({\bf A})}J_0$ and this operator is an observable of the WRES. \\

\noindent {\bf (f)} $U_{J_0}: \cP \ni g \mapsto U_{g}|_{\sH_{J_0}}$ is a complex WRES over the complex Hilbert space $\sH_{J_0}$ and furthermore

\begin{enumerate}
\item  the von Neumann algebra and 
lattice of elementary observables associated to this complex WRES   are the  full $\gB(\sH_{J_0})$ and   full $\cL(\sH_{J_0})$ respectivly;

\item  if $A$ is an observable of the initial quaternionic WRES, then  $A_{J_0}$ is an observable of the associated complex WRES and the map $A\to A_{J_0}$ is bijective
over the full set of densely defined selfadjoint operators in $\sH_{J_0}$ (the observables of the associated complex WRES) which preserves the (point, continuous) spectra.
\end{enumerate}

\noindent {\bf (g)}  $\gR_U$ is of quaternionic-complex type $\gR_U' = \{aI+bJ_0 \:|\: a,b \in \bR\}$ and,
referring to the complex Hilbert space $\sH_{J_0}$, the restriction map $\gR_U \ni A \mapsto A|_{\sH_{J_0}} \in \gB(\sH_{J_0})$
defines 

\begin{enumerate}
\item    a norm-preserving weakly-continuous real unital $^*$-algebra 
isomorphism from  $\gR_U$ onto the full $\gB(\sH_{J_0})$ which  maps $J_0$ to $jI$;

\item  an isomorphism of orthocomplemented complete lattices from  $\cL_{\gR_U}(\sH)$ onto  the full  lattice  $\cL(\sH_{J_0})$
 in agreement  with the thesis of Sol\`er theorem;
 
 \item an isomorphism of $\sigma$-complete Boolean lattices $P^{(A)} \ni P^{(A)}_E \mapsto  P^{(A|_{J_0})}_E \in P^{(A|_{J_0})}$ for every fixed $A$ (not necessarily bounded)
   observable of the quaternionic WRES, $A|_{J_0}$ being 
 the corresponding  observable of the complex WRES. 
\end{enumerate}

\end{theorem}

\begin{proof} (a) and (b). Everything  but the last statement in (b),  immediately arise from  Prop.\ref{WignerTHM}. 
The last sentence in (b) follows from
the fact that the penultimate identity in (\ref{basispoinc2}) and the uniqueness of the polar decomposition imply that $\widetilde{P}'_0 = (U_pJ_0U_p^{-1})\: (U_p \widetilde{P}_0 U_{p}^{-1})$ 
is the polar decomposition of $\widetilde{P}'_0 = U_p\widetilde{P}_0 U_p^{-1}$.
However, $U_pJ_0=J_0U_p$ entails  $U_pJ_0U_p^{-1}=J_0$.

(c)  Is a weaker case of Proposition \ref{gardintcompl} (a) and (d) since  $U_pJ_0=J_0U_p$.

(d) Since $U$ is irreducible, $\gR_U$ is irreducible and thus $\gR'_U$ must have one of the three  mutally exclusive forms listed
 in Theorem \ref{threecommutant}. The first case is impossible since $\gR_U'$ would be made of selfadjoint elements while
  $0\neq J_0 \in \gR_U'$ is 
 anti selfadjoint. The third case is similary forbidden because $\gR_U\cap \gR_U'$ would be made of selfadjoint elements while 
 $0\neq J_0 \in \gR_U \cap \gR_U'$ is 
 anti selfadjoint. We conclude that $\gR'_U = \{aI + bJ\:|\: a,b \in \bR\}$ for some complex structure $J$ determined up to its sign. Since 
 $J_0 \in \gR_U'$ is a complex structure it must be $J_0 = \pm J$ and $\gR'_U = \{aI + bJ_0\:|\: a,b \in \bR\}$. This argument also proves the last statement observing that 
 $J_1 u({\bf A})={\bf A}J_1$ for every ${\bf A} \in \gp$ is equivalent to $J_1 \in \{U_g\}_{g\in \cP}' = \gR_U'$ in view of (3) Theorem \ref{thmuext}.
 
 (e). $J_0u({\bf A})=u({\bf A})J_0$ is valid from (c) . Taking the closures remembering that $J_0 \in \gB(\sH)$, we have 
 $J_0\overline{u({\bf A})} = \overline{u({\bf A})}J_0$. Since $\overline{u({\bf A})}$ is antiselfadjoint 
 and  exploiting  $J_0\in \gB(\sH)$, we have 
 $(J_0\overline{u({\bf A})})^*= -\overline{u({\bf A})}(-J_0) = \overline{u({\bf A})}J_0 = J_0\overline{u({\bf A})}$ which is therefore self-adjoint.
  $J_0\overline{u({\bf A})}$ commutes with $J_0$ so that its PVM is included in $\gR_U$ as wanted.

  (f).  Due to the already established properties, we only have to prove that $\cP \ni g \mapsto U_{g}|_{H_{J_0}}$
   is irreducible and locally faithfull.
   The former property is  an immediate consequence of (a) Proposition \ref{complexidentif2}, 
   the latter  straightforwardly arises from locally faithfulness of $U$ taking advantage of  (a),(b) Prop.\ref{propHJop}. So 
  $\cP \ni p \mapsto U_{g}|_{H_{J_0}}$ is a complex WRES over $\sH_{J_0}$ as requested. Next, irreducibility and the complex version 
  of Schur's lemma eventually  yield 
  $\{ U_{g}|_{H_{J_0}} \:|\: g \in \cP\}''= \gB(\sH_{J_0})$, so that, in particular  $\cL_{\gR_{U_{J_0}}}(\sH_{J_0})=\cL(\sH_{J_0})$. The last statement immediately arises from 
  Proposition \ref{propHJop} (a)(b) and the statement at the end of  (c) therein, and Prop.\ref{operatorfunctioncompl} (d).

  (g). The identity  $\gR_U' = \{aI+bJ_0 \:|\: a,b \in \bR\}$ was established above proving (d), the rest of the thesis 
  easily arises from Theorem \ref{quaternioni3theorem} and Proposition \ref{operatorfunctioncompl} (a).
\end{proof}

\begin{corollary}\label{CorollaryMT} With the same hypotheses as for Theorem \ref{poinccomplexstructure} the following facts are valid.\\
{\bf (a)} For all ${\bf M}, {\bf N} \in E_{\gp}$ the following facts hold

\begin{enumerate}

\item   
$J_0(u({\bf M})+ J_0 u({\bf N}))= (u({\bf M})+ J_0 u({\bf N}))J_0$ and 
$J_0\overline{u({\bf M})+ J_0 u({\bf N})}= \overline{u({\bf M})+ J_0 u({\bf N})}J_0$;

\item $\overline{u({\bf M}) + J_0 u({\bf N}) }$ is an observable of the quaternionic WRES if and only if it is self-adjoint.
\end{enumerate}
{\bf (b)} The identity holds 
 $$D_G^{(U_{J_0})} = D_G^{(U)}\cap \sH_J$$
and  for all ${\bf M}, {\bf N} \in E_{\gp}$,
 
\begin{enumerate}
\item $u_{J_0}({\bf M}) = u({\bf M})|_{\sH_{J_0}\cap D_G^{(U)}}$ and  $u_{J_0}({\bf M})$ univocally determines $u({\bf M})$;

\item   $\overline{u_{J_0}({\bf M}) + j u_{J_0}({\bf N})} = \overline{u({\bf M}) + J_0 u({\bf N}) }|_{\sH_{J_0}}$ and the left-hand side is self-adjoint (i.e., an observable of the 
complex associated WRES) if and only if $ \overline{u({\bf M}) + J_0 u({\bf N}) }$ is an observable of original quaternionic WRES.
\end{enumerate}
\end{corollary}

\begin{proof} (a) We can apply Proposition \ref{gardintcompl} (d) since  $U_pJ_0=J_0U_p$ obtaining 
$J_0 (u({\bf M}) + J_0u({\bf N})) =(u({\bf M}) + J_0u({\bf N}))J_0$. Taking the closures remembering that $J_0 \in \gB(\sH)$, we have 
 $J_0 \overline{u({\bf M}) + J_0u({\bf N})} = \overline{u({\bf M}) + J_0u({\bf N})}J_0$. If $ \overline{u({\bf M}) + J_0u({\bf N})}$ is  selfadjoint, then its PVM satisfies the same 
 identity in view of (b) Lemma \ref{commst}. In turn, since $\gR_U'= \{aI + bJ_0\:|\:a,b \in \bR\}$, the PVM of  $\overline{u({\bf M}) + J_0u({\bf N})}$
 belongs to $\gR''_U = \gR_U$ as wanted.
 
 (b)  Everything easily follows from (a) above  and (c) and (d) in Prop.\ref{gardintcompl} also noticing that the observables of the complex WRES are all of
  selfadjoint operators in $\sH_{J_0}$. The fact that $u_{J_0}({\bf M})$ univocally determines $u({\bf M})$ is consequence of (b) 
 in  Prop.\ref{propHJop}.
\end{proof}

\section{Quaternionic relativistic elementary systems (qRES)} As done in \cite{MO1} we pass to present a more precise description of an elementary relativistic system.

\subsection{Elementary systems}
 We first focus on  the general notion of {\em elementary system} without referring to a group of symmetry.
Exactly as in \cite{MO1} we adopt the following general definition whose motivation is the same as presented in Sec. 5.1 of \cite{MO1}.
\begin{definition}{\rm
An \textbf{elementary system} is an irreducible von Neumann algebra $\gR$ over a separable quaternionic Hilbert space $\sH$. }
\end{definition}
\noindent An elementary system is therefore a quantum system such that $\gR'$ does not include non-trivial orthogonal projectors: they could be intepreted as elementary observables of 
another  external system independent from the system represented by the self-adjoint operators in $\gR$.
%%% CORRETTO 
If the center of  $\cL_\gR(\sH)$ is trivial, 
%%%% CORRETTO
 a sufficient and quite usual condition for having an elementary system is the existence in $\gR$ of a  {\em maximal set of compatible observables}
 as discussed in Sec. 5.1 of \cite{MO1}.\\
In the hypothesis that the elementary propositions of every quantum system are described by the orthogonal projectors of a quaternionic von 
Neumann algebra, it is clear that every symmetry of the system can be realized in terms of an automorphism of the lattice of projectors itself.  According to the discussion in Sec. 5.2 of \cite{MO1},  a  
{\em relativistic elementary system} is first of all  an elementary system supporting a continuous representation of Poincar\'e group which is irreducible with respect to that representation: there are no orthogonal projectors $P \in \gR$ 
that are fixed under the representation. Continuity must have a direct physical meaning. In fact, it is  induced by
seminorms defined by the states.
The second idea behind the notion of relativistic elementary system  is that the representation of the Poincar\'e group  must determine $\gR$ itself.  The idea is to view the representation of Poincar\'e group as a unitary 
representation (defined up to generalised phases in the centre of the algebra) and to assume that $\gR$ is the algebra generated by these unitary operators.
Thanks to Theorem \ref{quaternioni3theorem} and Gleason Theorem (including the quaternionc version proposed by Varadarajan \cite{V2}) and corrected in \cite{GleasonQuat} and following a
 procedure similar to the analog exploited in \cite{MO1}.  we can completely characterize both the states and the symmetries for such an elementary  system.
\begin{proposition}\label{gleason+wigner}
Let $\gR$ be an elementary system over the separable quaternionic Hilbert space $\sH$, then the following facts hold true.
\begin{enumerate}[(a)]
\item Assuming $\mbox{dim}\: \sH\neq  2$,  if $\mu:\cL_\gR(\sH)\rightarrow [0,1]$ is a state (i.e., a $\sigma$-additive probability measure), then there exists a unique positive, selfadjoint 
unit-trace, trace-class operator  $M\in \gR$ such that (see (3) in Sect. \ref{secGLEASON})
\begin{equation}
\mu(P)=tr(PMP)\ \mbox{ for every } P\in\cL_\gR(\sH)
\end{equation}
Moreover every positive, selfadjoint unit-trace, trace-class operator of $M \in \gR$ defines a state by means of the same relation.
\item If $h:\cL_\gR(\sH)\rightarrow \cL_\gR(\sH)$ is a symmetry of the system (i.e. an automorphism of $\sigma$-complete orthocomplemented lattices), then there
exists a $\bR$-linear surjective norm-preserving map $U:\sH\rightarrow\sH$ such that
\begin{equation}\label{symmetryrealization}
h(P)=UPU^{-1}\ \mbox{ for every } P\in\cL_\gR(\sH)
\end{equation}
and the following further facts are valid.
\begin{enumerate}[(i)]
\item In both the quaternionic-real and quaternionic-quaternionic cases, $U\in\gR$.
\item In the quaternionic-complex case we have two mutually exclusive  possibilties
\begin{itemize}
\item   $U$  belongs to $U\in\gR$,  

\item  $U$ is antilinear w.r.t. $j$ and $k$ -- thus it is linear with respect to $i$ -- and  anticommutes with $J$.  In this case  $U\not\in\gR$ but $U^2\in\gR$. 
\end{itemize}
\item Every real-linear surjective norm-preserving map $V: \sH \to \sH$ that satisfies (i) or (ii) depending on the case, satisfies (\ref{symmetryrealization}) in place of $U$ for the given $h$ if and only if 
$U^{-1}V\in \gR\cap \gR'$.
\end{enumerate}
Finally, every real-linear bijective norm-preserving map $U$ that satisfies (i) or (ii) depending on the case, defines a symmetry by means of (\ref{symmetryrealization}).
\end{enumerate}
\end{proposition}
\begin{proof}
We need  to prove a preliminary technical lemma.
\begin{lemma}\label{traceclass}
Referring to Theorems \ref{threecommutant} and \ref{quaternioni3theorem}, let $\gR$ be of  quaternionic-complex or quaternionic-quaternionic type.
 Then $A\in\gR$ is trace-class if and only if $A_J$, respectively $A_{JK}$, is trace-class. Moreover $tr(A_J)_{\sH_J}=tr(A)_\sH$ and $tr(A_J)_{\sH_{JK}}=tr(A)_\sH$.
\end{lemma}
\begin{proof}
We study the quaternionic-complex case, the remaining one being strictly analogous. 
Focusing on $|A| :=\sqrt{A^*A}$, proposition \ref{operatorfunctioncompl} guarantees that $|A|_J=(\sqrt{A^*A})_J=\sqrt{(A^*A)_J}=\sqrt{A_J^*A_J}=|A_J|$. If $\{f_k\}_{k\in\bN}$ is 
a Hilbert basis for $\sH_J$, it is also  a Hilbert basis also for $\sH$ (d) in view of  Prop.\ref{propHJ} and we have
\begin{equation}\label{traceclassval}
0 \leq \sum_{k=0}^\infty \langle f_k ||A_J|f_k\rangle = \sum_{k=0}^\infty \langle f_k ||A|_Jf_k\rangle=\sum_{k=0}^\infty \langle f_k ||A|f_k\rangle \leq +\infty \:, 
\end{equation}
which immediately implies  the first implication of the first statement. Now, suppose that $A$ is trace-class on $\sH$, then there exists a Hilbert basis $\{e_k\}_{k\in\bN}$ such that 
$\sum_{k=0}^\infty \langle e_k||A|e_k\rangle <\infty$. Once this is valid for one basis, it holds true for any Hilbert basis, in particular for any basis of $\sH_J$, thanks to (d) in view of  Prop.\ref{propHJ}. Using (\ref{traceclassval}) 
again we get the opposite implication. The second property  straightforwardly follows from the already exploited fact that  every Hilbert basis for $\sH_J$  is also a Hilbert basis of $\sH$ 
and the observation  that the Hermitian scalar product of $\sH_J$ is the restriction of the one of $\sH$. 
\end{proof}
\noindent Let us come back to the main thesis.

(a)  Consider a state $\mu:\cL_\gR(\sH)\rightarrow [0,1]$.  The thesis is true when $\gR$ is of quaternionic-real type since 
$\cL_\gR(\sH) = \cL(\sH)$ for  Theorem \ref{quaternioni3theorem}
and we can apply  
 Gleason Theorem as stated in (3) of Section \ref{secGLEASON}.  
Let us pass to  the quaternionic-complex case, the only we will consider,  
the quaternionic-real  case being strictly similar using the real version of Gleason theorem. 
Thanks to Theorem \ref{quaternioni3theorem}, the map $\mu$ can be reinterpreted as a $\sigma$-additive probability 
measure $\mu_J$ on the {\em whole} $\cL(\sH_J)$ such that  $\mu_J(P_J):=\mu(P)$. Since we are dealing with the full lattice of orthogonal projectors of $\sH_J$,  the complex version of 
 Gleason theorem (which is valid for real, complex and quaternionic Hilbert spaces see \cite{V2}])  guarantees the existence of a
 positive, selfadjoint, unit-trace, trace-class operator $T_J$ on $\sH_J$ such that $\mu_J(P_J)=tr(T_JP_J)_{\sH_J}$ for every $P_J\in\cL(\sH_J)$. 
Thanks to Lemma \ref{traceclass}, the extension $T$ of $T_J$, which belongs to $\gR$ thanks to the identification of the latter with $\gB(\sH_J)$, is still of trace-class and unit-trace. Moreover it is clearly positive and self-adjoint thanks 
to Proposition \ref{operatorfunctioncompl}. So, let $P\in\cL_\gR(\sH)$. The operator $TP$ is still of trace-class and commutes 
with $J$, hence $(TP)_J=T_JP_J$. Finally $tr(TP)_\sH=tr((TP)_J)_{\sH_J}=tr((T_JP_J)_{\sH_J}=\mu_J(P_J)=\mu(P)$. 
Conversely,  if $T\in\gR$ is positive, selfadjoint, unit-trace, trace-class operator,  the map $\mu(P):=tr(TP)$ for $P\in\cL_\gR(\sH)$ is a well-defined $\sigma$-continuous probability measure as the
 reader can immediately prove.

(b) Let us prove the initial statement together with (i) and (ii).  With the same argument exploited to prove (a), taking the identifications on $\gR$ as stated in Theorem
 \ref{quaternioni3theorem}  into account, we can reinterpret $h$ as
 a lattice automorphism on $\cL(\sH),\cL(\sH_J)$ or $\cL(\sH_{JK})$, respectively. We will denote it as $h,h_J$ and $h_{JK}$, respectively. We are in a position to apply Theorems 4.27 and 4.28 in \cite{V2}.
  Sticking to the quaternionic-real and quaternionic-quaternionic cases there always exists a unitary operator $U_0\in\gB(\sH)$ or $\gB(\sH_{JK})$, respectively, such 
  that $h(P)=U_0PU_0^{-1}$ or $h_{JK}(P_{JK})=U_0P_{JK}U_0^{-1}$ for every orthogonal projector. In the quaternionic-complex case, instead, there exists a either unitary or 
  antiunitary operator $U_0$ on $\sH_J$ such that $h_J(P_J)=U_0P_JU_0^{-1}$.\\
   In the quaternionic-real case  $U_0 \in \gR$, it being $\gR=\gB(\sH)$ and thus $U:=U_0$ is the wanted operator. In the quaternionic-quaternionic
   case, $U_0$ can  uniquely be  extended to a unitary operator $U\in\gR$ thanks to the isomorphism of  $\gR$ and $\gB(\sH_{JK})$  proved to exist in (c) Theorem \ref{quaternioni3theorem}.
 Since $h_{JK}(P_{JK})=h(P)_{JK}$, it is easy to 
   see that $h(P)= UPU^{-1}$ which proves that $U$ is the operator we searched for. In the quaternionic-complex case, when  $U_0$
is unitary,  reasoning as in the quaternionic-quaternionic case, we find a uniquely defined unitary extension $U\in\gR$ of $U_0$ such that $h(P)=UPU^{-1}$ so that $U$ is the operator we were looking for. 
Eventually suppose that, instead,  $U_0$ is antiunitary, hence a surjective, norm-preserving antilinear map on $\sH_J$. 
 Decompose every $x\in \sH$ as  $x=x_1+x_2k\in\sH$ for $x_1,x_2\in\sH_J$ and define 
$Ux:=U_0x_1-U_0x_2k$.
We want to prove that $U$ is the operator we are looking for.
By direct inspection one sees that $U$ is real-linear, norm-preserving and surjective, with inverse $U^{-1}x=U_0^{-1}x_1-(U_0^{-1}x_2)k$. 
Moreover $U(xj)=U(x_1j-(x_2j)k)=U_0(x_1j)+U_0(x_2j)k=-(U_0x_1)j+(U_0x_2)kj=-[U_0x_1-(U_0x_2)k]j=-(Ux)j$, hence $U$ is antilinear with respect to 
$j$. Similarly we find $U(xk)=U(-x_2+x_1k)=-U_0x_2-(U_0x_1)k=-[U_0x_1-(U_0x_2)k]k=-(Ux)k$. As a consequence we also have  $U(xjk)=(Ux)jk$. The same properties can be established  for $U^{-1}$.
Let us now prove that $U$  anticommutes with $J$. Let $x\in\sH$ as above, then $UJx=U(x_1j+(x_2j)k)=U_0(x_1j)-(U_0(x_2j))k=-(U_0x_1)j+(U_0x_2)jk=-J(U_0x_1)+(J(U_0x_2))k=-J(U_0x_1-(U_0x_2)k)=-JUx$.
Next  notice that $UUx=U_0^2x_1+(U_0^2x_2)k$. Since $U_0^2\in\gB(\sH_J)$ it can be uniquely extended to a unitary operator $W\in\gR$ and  by construction it holds $U^2=W$ as wanted. 
Finally, $U$ implements $h$ as required. Indeed,  take $P\in\cL_\gR(\sH)$, then by definition it follows
\begin{equation*}
\begin{split}
&UPU^{-1}x=U_0P_JU_0^{-1}x_1+(U_0P_JU_0^{-1}x_2)k=h(P)_Jx_1+(h(P)_Jx_2)k=\\
&=\widetilde{h(P)_J}x=h(P)x
\end{split}
\end{equation*}
Proof of (i) and (ii) are completed.\\
Let us prove (iii). 
Suppose first that we are in the quaternionic-real or quaternionic-quaternionic case and let $V\in\gR$ such that $VPV^{-1}=UPU^{-1}$ for every
 projector $P\in\cL_\gR(\sH)$. This can be rewritten as $(U^{-1}V)P(U^{-1}V)^{-1}=P$ for every  $P$. In the quaternionic-real case, Theorem 4.27 in   \cite{V2} 
and the $\bH$-linearity of $UV^{-1}$  guarantee that $U^{-1}V=aI$ for some $a\in\bR$ and the proof of (iii) ends here for this case. In the quaternionic-quaternionic case, since both $U^{-1}$ 
  and $V$ commutes with $J,K$ ($U^{-1}=U^*\in\gR$) we have $(U^{-1}V)_{JK}P_{JK}(U^{-1}V)_{JK}^{-1}=P_{JK}$ for every $P_{JK}\in \sH_{JK}$. Again, Theorem 4.27 in  \cite{V2} 
  assures that $(U^{-1}V)_{JK}=aI_{JK}$ for some $a\in\bR$. Uniqueness of the extension from $\sH_{JK}$ to $\sH$ leads to $U^{-1}V=aI$ for some $a\in\bR$ concluding the proof of (ii)
in this case, too. Let us pass to the remaining quaternionic-complex case. If $U,V\in\gR$, then we can repeat the argument exploited  for the
 quaternionic-quaternionic case finding  $U^{-1}V=aI+bJ$ for some $a,b\in\bR$ concluding the proof.
Suppose conversely that $U,V$ are instead real-linear norm-preserving surjective maps which are antilinear w.r.t $j$ and $k$ and anticommute with $J$. These 
properties hold also for $U^{-1}$ as a direct calculations shows. 
The function $U^{-1}V$ turns easily out to be a linear bounded operator which commutes with $J$, hence we can consider $(U^{-1}V)_J\in\gB(\sH_J)$. 
Moreover, since $(U^{-1}V)_JP_J(U^{-1}V)_J^{-1}=P_J$, Theorem 4.27 in  \cite{V2}  guarantees that $(U^{-1}V)_J=cI_J$ for some $c\in\bC$. Again uniqueness of the 
extension from $\sH_J$ to $\sH$ yields $U^{-1}V=aI+bJ$ for some $a,b\in\bR$. \\
Let us pass to the final statement. It is easy to prove in the cases (i) and (ii) with $U\in\gR$. So we only focus attention on the quaternionic-complex case with  
$U$ which is  real-linear, norm-preserving surjective,  antilinear with respect to $j$ and $k$ and anticommutes with $J$. Notice that the same properties hold true for $U^{-1}$. The only difficult step  is to establish
 that $h(P):=UPU^{-1}$ is an orthogonal projector of $\gR$ if $P \in \cL_\gR(\sH)$. The operator  $h(P)$ is clearly quaternionic linear, belongs to $\gB(\sH)$, 
 and is idempotent. Moreover it commutes with $J$, 
hence $h(P)\in\gR''=\gR$ for every $P\in\cL_\gR(\sH)$ because $\gR' = \{aI+bJ\:|\; a,b \in \bR\}$. To establish that $h(P) \in \cL_\gR(\sH)$ it remains to prove that $h(P)$ is selfadjoint. 
To prove it,  let us initially restrict  ourselves to $\sH_J$. Since $U$ anticommutes with $J$ and it is antilinear with respect to $j$, if  $x\in\sH_J$  it must hold $JUx=-UJx=-U(xj)=(Ux)j$. As the same holds for $U^{-1}$
 we immediately see that $U$ is a well-defined antilinear norm-preserving surjective operator on $\sH_J$. The standard complex polarization identity implies  $\langle Ux|Uy\rangle=\langle y|x\rangle$
 for every $x,y\in\sH_J$. The same property 
is valid for $U^{-1}$. Hence for $x,y\in\sH_J$, using the established property and selfadjointness of $P$,
\begin{equation*}
\langle x|UPU^{-1}y\rangle=\langle PU^{-1}y| U^{-1}x\rangle =\langle U^{-1}y| PU^{-1}x\rangle  = \langle UPU^{-1}x|y\rangle\:,
\end{equation*}
so that $((UPU^{-1})_J)^*=(UPU^{-1})_J$. This identity lifts to the whole space $\sH$ in view of Proposition \ref{propHJop}, giving $h(P)^*=h(P)$ and thus $h(P) \in \cL_\gR(\sH)$ as required.
\end{proof}
\subsection{Relativistic elementary systems in quaternionic Hilbert spaces}
 We are now in a position to state the general definition of {\em quaternionic elementary relativistic system} that, exactly as the analogous definition in \cite{MO1} for real and complex systems, includes the
notion of elementary system, the presence of an irreducible continuous representation of Poincar\'e group and the requirement that them representaion itself  defines the algebra of observables of the system. 

\begin{definition}\label{qRES}{\em
A \textbf{quaternionic relativistic elementary system} (qRES) is a couple $(\gR,h)$ where $\gR$ is an irreducible von Neumann algebra over the quaternionic Hilbert
 space $\sH$ and $h:\cP\ni g\mapsto h_g\in \mbox{Aut}(\cL_\gR(\sH))$ is a locally faithful representation of the Poincar\'{e} group satisfying the following requirements
\begin{enumerate}[(a)]
\item $h$ is irreducible, in the sense that $h_g(P)=P$ for all $g\in\cP$ implies either $P=0$ or $P=I$
\item $h$ is continuous, in the sense that the map $\cP\ni g\mapsto \mu(h_g(P))$ is continuous for every fixed $P\in\cL_\gR(\sH)$ and quantum state $\mu$
\item $h$ defines the observables of the system. That is, accordingly to Proposition \ref{gleason+wigner} and representing $h$ in terms of unitary operators $U_g\in \gR$ 
(defined up to unitary factors in the center $\gZ_\gR:= \gR \cap \gR'$),
$h_g(P) = U_gPU_g^{-1}$ for $g \in \cP$ and $P \in \cL_\gR(\sH)$, it must be
\begin{equation*}
\left(\{U_g\:|\: g\in\cP\}\cup \gZ_\gR\right)''\supset \cL_\gR(\sH)
\end{equation*}
\end{enumerate}
}
\end{definition}
\begin{remark}\label{remarklinearity} $\null$\\
{\em 
{\bf (1)} Local faithfulness requirement of $h:\cP\ni g\mapsto h_g\in \mbox{Aut}(\cL_\gR(\sH))$ is equivalent to a physically more meaningful  requirement:

  {\em The representation of spacetime translations subgroup   $\bR^4 \ni t \mapsto h_{(I,t)}$ is not trivial}.\\
 This equivalence immediately arises from Proposition \ref{proplocfaithfulness}
when equipping $\mbox{Aut}(\cL_\gR(\sH))$ with the topology induced by all off the seminorms $p_{P,\mu}$ for every $P\in \cL_\gR(\sH)$ and every state $\mu$ on $\cL_\gR(\sH)$, defined as 
$p_{P,\mu}(h) := |\mu(h(P))|$, since the resulting topological group is Hausdorff as the reader  may prove easily. As a byproduct of Proposition \ref{proplocfaithfulness}, if $h$ is not injective, this is only due to the sign of the $SL(2,\bC)$ part of its argument. \\
{\bf (2)} In (c), we have explicitly 
assumed that $U_g \in \gR$ is always valid,
excluding the case of $U_g$ antilinear w.r.t $j,k$ and anticommuting with $J$  in the quaternionic-complex case (corresponding to an {\em anti unitary} operator in the
 complex Hilbert space $\sH_J$). The reason is the following (the analogous disussion in \cite{MO1} contained a trivial mistake we correct here).  From the polar decomposition of $SL(2,\bC)$ one sees 
that every $g \in \cP = SL(2,\bC) \ltimes \bR^4$ can 
 be always decomposed into a product of this kind $g= ttbbrr$ where $r$ is a spatial rotation, $b$ a boost, and $t$ a four-translation, hence the real-linear surjective norm-preserving operator
 $U_g =  U_t^2U_r^2 U_b^2$ generates the symmetry $h_g$. It is now clear that, whether or not $U_t$, $U_r$ and $U_b$ are $\bH$-linear and commute with $J$, 
  their squares are $\bH$-linear  and commute with $J$, and so $U_g$ would be linear and commutes with $J$ in every case. }
\end{remark}

\noindent The  map $\cP\ni g\mapsto U_g$ introduced in  Definition \ref{qRES} (c) is not, in general, a group representation since we may have  $U_g U_h=\Omega(g,h)U_{gh}$
 for  operators $\Omega(g,h)\in\gU(\gZ_\gR)$ where $\gU(\gZ_\gR)$ henceforth denotes  the set of unitary operators in the center of $\gR$. In particular $U_e= \Omega(e,e)$ putting $g=h=e$
in the identity above.
Such a map  $\cP\ni g\mapsto U_g$ is known as a {\bf projective unitary representation} of $\cP$, while the function $\Omega:\cP\times\cP\rightarrow \gU(\gZ_\gR)$ is 
said to be the {\bf multiplier function} of the representation. 
\begin{remark}
{\em The structure of $\gZ_\gR$ implies  the following algebraic identifications for a qRES. We have that   $\gU(\gZ_\gR)=\bZ_2 I$ -- the multiplicators are signs -- in the 
quaternionic-real and quaternionic-quaternionic cases and  $\gU(\gZ_\gR)=U(1) I$ -- the multiplicators are complex phases --  in the quaternionic-complex case.}
\end{remark}
\noindent The associativity property of the operator multiplication easily gives the {\bf cocycle}-property,
\begin{equation}\label{multprop}
\Omega(r,s)\Omega(rs,t)=\Omega(r,st)\Omega(s,t)\quad \mbox{for all $r,s,t\in\cP$}\:.
\end{equation} 
For  any function $\chi:\cP\rightarrow \gR(\gZ_\gR)$ the map $\cP \ni g\mapsto \chi(g)U_g$ is still a projective representation associated with the same representation
 $h$ of $\cP \ni g\mapsto U_g$, whose multiplier is now given by $$\Omega_\chi(g,h)=\chi(g)\chi(h)\chi(gh)^{-1}\Omega(g,h)\quad \mbox{for all $g,h \in \cP$.}$$
A natural question then concerns the possibility of getting rid of the multipliers by finding a function $\chi$ such that $\Omega_\chi=I$ in order to end up with a proper 
unitary representation from a given projective unitary representation.
A positive answer can be given for all of the three cases. 

\begin{proposition}\label{bargmann}
Let $\gR$ and $h$ respectively  be the von Neumann algebra and  the Poincar\'{e} representation  of a qRES
as in Definition \ref{qRES}. The following facts hold.\\

\noindent {\bf (a)} There exists a locally faithful  strongly-continuous unitary representation $\mathcal{P}\ni g\mapsto U_g\in\gR$  on $\sH$ such that $h_g(P)=U_g P U_g^{-1}$ 
for every $g\in\mathcal{P}$ and every $P \in \cL_\gR(\sH)$.\\

\noindent{\bf (b)} $\mathcal{P}\ni g\mapsto U_g\in\gR$   is irreducible if understood respectively on $\sH$, $\sH_J$ or $\sH_{JK}$ according to the three cases of Theorem \ref{quaternioni3theorem}.
\end{proposition}
\begin{proof}
\noindent We simultaneously  prove (a) and (b). We already know that $h_g(\cdot) =V_g\cdot V_g^*$ for some unitary operator $V_g\in\gR$. By the continuity
 hypothesis on $h_g$ and (a) of Proposition \ref{gleason+wigner}, reinterpreting everything within the corresponding space as established by Theorem \ref{quaternioni3theorem}, we see that the map
\begin{equation*}
\cP \ni g\mapsto tr(P_\psi h_g(P_\phi))=|\langle\psi|V_g \phi\rangle|^2 \mbox{ for every } 
\begin{cases}
&\psi,\phi\in\sH\ \mbox{ if } \gR'\equiv \bR\\
&\psi,\phi\in\sH_J\ \mbox{ if } \gR'\equiv \bC\\
&\psi,\phi\in\sH_{JK}\ \mbox{ if } \gR'\equiv \bH\\
\end{cases}
\end{equation*}
is continuous, where $P_\chi\in\gR$ is the projector obtained by lifting to $\sH$ the projector on the one-dimensional subspace generated by the generic vector $\chi$ within the space $\sH,\sH_J$
 or $\sH_{JK}$, depending on the case (the equality in the equation derives from the last statement of Lemma \ref{traceclass}).
Let us focus on the quaternionic-complex case first. Thanks to the above remark, following  the analysis contained in the well-known paper \cite{Bargmann}, we
  get a strongly-continuous unitary representation $\cP \ni g\mapsto U_g$ on $\sH_J$ such that $U_g=\chi_g (V_g)_J$ for some $\chi_g\in U(1)$, hence 
  generating $h_J$. More precisely $U_g$ can be uniquely extended to a unique unitary operator $\tilde{U}_g\in\gR$ such that $\tilde{U}_g=(a_gI+b_gJ)V_g$. 
  The map $g\mapsto \tilde{U}_g$ generates $h$ and is still faithful as $h$ is.
Notice that since $\|\tilde{U}_gx-x\|^2=\|U_gx_1-x_1\|^2+\|U_gx_2-x_2\|^2$ the representation $\cP \ni g\mapsto \tilde{U}_g$ is also strongly continuous on $\sH$. 
Irreducibility of $U$ on $\sH_J$ follows from the following argument. 
Since the family $\gU:=\{U_g,\:|\:g\in\cP\}$ is closed under  Hermitian conjugation, thanks to Remark \ref{remirr} we need only to prove that $(\gU_J)'\cap\cL(\sH_J)=\{0,I_J\}$, 
but this is a direct consequence of the irreducibility of $h$. Indeed, if $P_J$ is a complex projector commuting with every $U_g$ then $(h_g)_J(P_J)=U_gP_JU_g^*=P_J$ for 
every $g\in\cP$ and thus  $P=0$ or $P=I$.
Let us next focus on the quaternionic-real case. Thanks to the analysis of \cite{Em} we can always find a strongly-continuous unitary representation $\cP \ni g\mapsto U_g$ on $\sH$
  such that $U_g=\chi_g V_g$ for some $\chi_g\in \bZ_2$, hence generating $h$. Again, it is straightforward to prove the irreducibiity and faithfulness of this representation on $\sH$.\\
Let us conclude the proof discussing the quaternionic-quaternionic case.
We affirm that we may always choose a representative $\cP \ni g\mapsto U_g$ on $\sH_{JK}$ equivalent to $g\mapsto (V_g)_{JK}$ such that $U_e=I_{JK}$, it is strongly 
continuous over an open neighbourhood of the identity $A_e$ and its multiplier $(g,h)\mapsto \Omega(g,h)$ is continuous over $A'_e\times A'_e$ with $A'_e\subset A_e$,  a  
smaller open neighbourhood of $e$ which can always be assumed to be connected ($\cP$ is a Lie group and as such it is locally connected).
The proof of this fact can be found within the proof of Proposition 12.38 in \cite{M2} which is valid both for complex and real Hilbert spaces since there is no distinctive role played by the imaginary unit. 
Since $\Omega(g,h) \in \{\pm I\}$ which is not connected if equipped with the topology induced by $\bR$ and $\Omega(e,e)=I$, the 
continuity of $\Omega$ guarantees that $\Omega(g,h)=I$ for every $g,h\in A'_e$. In other words  $U_gU_h=U_{gh}$ for every $g,h\in A'_e$. 
As the  group $\gU(\sH_{JK})$ of all unitary operators over $\sH_{JK}$ is a topological group with respect to the strong operator topology, 
the function $\cP \ni g\mapsto U_g$ is then a local topological-group homomorphism as in Definition B, Chapter 8, Par.47 of \cite{P}.
 Since, as established in \cite{connect}, $\gU(\sH_{JK})$  is connected if $\dim \sH_{JK}$ is not finite and $\cP$
is a simply connected Lie group, we can apply  Theorem 63 \cite{P} proving that there exists a strongly-continuous unitary representation 
$\cP \ni g\mapsto W_g \in \gU(\sH_{JK})$ such that $W_g=U_g$ on some open neighborhood of the identity $A_e''\subset A_e'$. 
If $\dim(\sH) = n < +\infty$, then $\gU(\sH_{JK})$ can be identified with the topological group $O(n)$. Its open subgroup $SO(n)$ is the connected 
component including the identity element $I$. In this situation, we can restrict ourselves to deal with a smaller initial open set $A'_e \cap B$ where $B$ is 
the pre-image through the map $U$ (which is continuous on $A'_e$) of an open set including $I$ and completely included in $SO(n)$. As $SO(n)$ is connected, 
we can finally exploit the same procedure as in the infinite dimensional case,  proving that there exists a strongly-continuous 
unitary representation $\cP \ni g\mapsto W_g \in \gU(\sH_{JK})$ such that $W_g=U_g$ on some open neighbourhood of the identity element $A_e''\subset A_e'\cap B$.
To conclude, we observe that since  the Lie group $\cP$ is connected, a standard result guarantees that every $g\in\cP$ can be written 
as $g=g_1\cdots g_n$ for some $g_1,\dots,g_n\in A_e''$. So, $W_g=W_{g_1}\cdots W_{g_n}=U_{g_1}\cdots U_{g_n}$ and
$h_g = h_{g_1}\circ \cdots \circ h_{g_n}$,
from which, dealing with the extensions $\widetilde{W}_g\in\gR$, it easily follows $h_g=\widetilde{W}\cdot \widetilde{W}_g^*$ for every $g\in\cP$. Again, 
faithfulness and irreducibility of the representation  $\cP \ni g\mapsto \tilde{W}_g$ immediately arises form the same properties of $h$.
\end{proof}
\noindent The following  technical lemma is useful in the proof of another  main  result of this work.
\begin{lemma}\label{commutantmaximal}
Let $\sH$ be an either real Hilbert space or quaternionic Hilbert space, with dimension striclty greater than one if quaternionic, then $\cL(\sH)'=\bR I$.
 In particular $\cL(\sH)''=\gB(\sH)$ and $\gB(\sH)'=\bR I$.
\end{lemma}
\begin{proof}
 $\cL(\sH)'=\bR I$ was  established in \cite{MO1} Lemma 5.16 for $\sH$ quaternionic and in \cite{MO1}  Lemma 5.13 for $\sH$ real. The final statements immediately arise from  $\cL(\sH)'=\bR I$.
\end{proof}
\subsection{Reduction to the complex Hilbert space case}
We are ready to state and prove  the second  main result  of this work.

\begin{theorem}\label{secondmain}
Let $(\gR,h)$ be a qRES on a quaternionic Hilbert space $\sH$. Let $\cP\ni g\mapsto U_g\in\gR$ be a corresponding
 locally faithful strongly-continuous unitary representation of $\cP$ on $\sH$ generating $h$ as in Proposition \ref{bargmann}. \\
 If the associated mass operator $M_U^2$ satisfies $M_U^2\ge 0$, then the following facts are valid.\\

\noindent {\bf (a)} $\gR$ turns out to be of quaternionic-complex type with commutant generated by the complex structure $J$.\\ 

\noindent {\bf (b)  }$U$ results to be irreducible  and $\gR =\gR_U$ defining  a quanternionic WRES with $M_U^2\ge 0$. \\

\noindent As a consequence of Theorem \ref{poinccomplexstructure},  the quaternionic relativistic elementary system can be equivalently described as a complex  Wigner relativistic
  elementary system in the complex Hilbert space $\sH_J$. In particular  $J$, up to sign, coincides to the Poincar\'e invariant complex structure arising from the polar
   decomposition of the subgroup of $U$ of the temporal translations.
\end{theorem}
\begin{proof} The final statements immediately arise from (b), $M_U^2 \geq 0$ and (a), so we prove (a) and (b).\\
(a) According to Theorem \ref{threecommutant}, $\gR$ can be of three mutually exclusive types. We prove that the quaternionc-real and the quaternionic-complex types
 are not possible and thus $\gR$ must be 
of quaternionic-complex case.
Let us start by assuming that  $\gR$ is of  quaternionic-real  type.  We affirm that $\sH$ cannot have dimension $1$. If it were the case, suppose for simplicity
 that $\sH=\bH$, the representation $U$ could be seen as a strongly-continuous locally faithful  unitary  representation on the $2$-dimensional {\em complex}
  Hilbert space $\sH_\bC=\bC^2$. Thus  $U$ would include a strongly-continuous  locally faithful  unitary  representation $V$ of $SL(2,\bC)$ on the $2$-dimensional
   complex Hilbert space $\bC^2$. This is not possible since the continuous finite-dimensional  unitary  representations of $SL(2,\bC)$  are completely reducible and
    the irreducible ones are the trivial representations only \cite{knapp}. In other words, the initial representation $U$ would be the trivial representation against the local faithfulness hypothesis.
So suppose that $\sH$ has  dimension $> 1$.
We have that $g\mapsto U_g$ is irreducible on $\sH$ and $\gR=\gB(\sH)$, in particular $\cL_\gR(\sH)=\cL(\sH)$.
Notice that $\gZ_\gR=\bR I$, hence the physical assumption on the observables of the system reduces to $\cL(\sH)\subset\{U_g\:|\: g\in\cP\}''$. Lemma \ref{commutantmaximal}
 guarantees that $\{U_g\:|\: g\in\cP\}''=\gB(\sH)$. Since $U$ is irreducible, Theorem \ref{WignerTHM} applies, ensuring the existence of a unitary antiselfadjoint operator $J_0$
  that commutes with the whole representation. 
So, we have $J_0\in\{U_g\:|\: g\in\cP\}'=\gB(\sH)'=\bR I$, which is clearly impossible because $J_0$ is not selfadjoint and not vanishing.
Let us prove that $\gR$ of quaternionic-quaternionic type is similarly impossible. We know by Proposition \ref{bargmann} that $g\mapsto U_g$ is irreducible if
 understood on $\sH_{JK}$ and that $\gR=\gB(\sH_{JK})$ under the action of the restriction map, 
in particular $\cL_\gR(\sH)$ is isomorphic to $\cL(\sH_{JK})$ under the action of the restriction map.
Again,  $\gZ_\gR=\bR I$ so that the assumption $\cL_\gR(\sH)\subset (\{U_g\:|\:g\in\cP\}\cup \gZ_\gR)''$ simplifies  to  $\cL_\gR(\sH)\subset \{U_g\:|\:g\in\cP\}''$. 
Thanks to Lemma \ref{Lemmagenertorsrestrictoinrea}, the anti-selfadjoint generators of $\bR \ni t\mapsto U_{\exp(t{\bf p}_\mu)}$ commutes with $J,K$ and their restriction to $\sH_{JK}$ give 
the generators of the one-parameter subgroups $t\mapsto (U_{\exp(t{\bf p}_\mu)})_{JK}$, while Proposition \ref{gardingtreal} ensures that $D_G^{(U_{JK})}=D_G^{(U)}\cap\sH_{JK}$. 
This proves that $M_{U_{JK}}^2=(M_{U}^2)_{JK}$. Hence, since the positivity of operators is preserved when moving from $\sH$ to $\sH_{JK}$, we get $M_{U_{JK}}^2\ge 0$. 
Theorem 4.3 of  \cite{MO1}   for real Hilbert spaces applies, ensuring the existence of a real-linear unitary antiselfadjoint operator $J_0\in\gB(\sH_{JK})$ such that $J_0 (U_g)_{JK}=(U_g)_{JK}J_0$ 
 for every $g\in\cP$. Thanks to Proposition \ref{lemmaextensreal}, $J_0$ uniquely extends to a quaternionic-linear unitary antiselfadjoint operator $\tilde{J}_0$ on $\sH$
  which commutes with $J$ and $K$ and the relation $J_0 (U_g)_{JK}=(U_g)_{JK}J_0$ similarly extends to $\tilde{J}_0U_g=U_g\tilde{J}_0$ for all $g\in\cP$. Hence 
   $\tilde{J}_0\in\{U_g\:|\:g\in\cP\}'\subset\cL_\gR(\sH)'$. To conclude, since $\cL_\gR(\sH)$ is ismorphic to $\cL(\sH_{JK})$ via the restriction map from $\sH$ to $\sH_{JK}$ and the 
   commutativity of operators is preserved
    when restricting operators on $\sH$ to $\sH_{JK}$ and  when extending operators on $\sH_{JK}$ to operators on $\sH$,
     we have that $J_0\in\cL(\sH_{JK})'$. Lemma \ref{commutantmaximal} eventually guarantees that $\cL(\sH_{JK})'=\bR I$, which is clearly impossible because $J_0$ is not self-adjoint and not vanishing.\\
     (b) The thesis is true if $U$ is irreducible. Indeed, under this hypothesis Theorem  \ref{poinccomplexstructure} implies that there is an up-to-sign unique complex
      structure $J_0$ commuting with $U$ and this is obtained from the polar decomposition of the
     generator of temporal translations.
      As $\{U_g \:|\: g \in \cP\} \subset \gR$ in view of (b) Proposition \ref{gleason+wigner} and (b) Remark \ref{remarklinearity} , and since $\gR' = \{aI+bJ \:|\: a,b \in \bR\}$ for (a), it must be $J=\pm J_0$.  
 We may henceforth assume $J=J_0$ redefining $J$ up a sign if necessary.      Theorem  \ref{poinccomplexstructure} also proves that  the action on operators 
  of the restriction map from $\sH$ to $\sH_{J_0}$
  makes $\gR_U$  isomorphic to $\gB(\sH_{J_{0}})$. We know form (a) that  the same  action on operators of the restriction map from $\sH$ to $\sH_J$
  makes $\gR$  isomorphic to $\gB(\sH_J)$. Since $\sH_J= \sH_{J_0}$,  so that $\gB(\sH_J) = \gB(\sH_{J_{0}})$ %, and since  $A=B$ if $A|_{\sH_J}= B|_{\sH_J}$
%  for operators $A,B\in \gB(\sH)$ commuting with $J$ ((b) Proposition \ref{propHJop})
, we conclude that $\gR= \gR_U$.\\
   To conclude the proof of (b) is  enough establishing that $\cP \ni g \mapsto U_g$ is irreducible. The proof of this fact is identical to the proof of Proposition 5.17 in  \cite{MO1}  just replacing real-complex for 
   quaternionic-complex, interpreting the real Hilbert space $\sH$ and the relevant subspaces $\sH_P$, $\sH_P^\perp$ appearing therein as our quaternionic Hilbert space and 
   corresponding quaternionic subspaces. Finally,  Theorem 4.3 in  \cite{MO1}  appearing in the proof of  Proposition 5.17 in  \cite{MO1}  has to be understood here as Theorem \ref{poinccomplexstructure}.
   \end{proof}

\section{Conclusions} This works proves that elementary relativistic quantum systems (generalizing Wigner's original ideas) described in a quaternionic Hilbert space, can be equivalently 
described in complex Hilbert spaces provided a natural spectral condition is satisfied for the observable corresponding to the squared mass of the system. The final picture removes any
 redundancy of the theory as every selfadjoint operator in the final complex Hilbert space represents an observable, the complex structure is Poincar\'e invariant, and the standard 
relation between continuous symmetries and conserved quantities is restored.  \\
Though this result seems  physically important,  it is worth stressing that it only concerns a very peculiar notion of physical system corresponding to an elementary relativistic particle.   
From the mathematical side these systems, in complex formulation, are represented by irreducible  von Neumann algebra of observables of type $I$. This restriction excludes physically 
important systems where other types of von Neumann algebras take place (Quantum Field Theory, finite temperature extended systems)  or 
simply the presence of a gauge group (quarks).\\
It is not obvious whether or not,  referring to a larger class of physical systems necessarily different from elementary relativistic particles, a real or quaternionic
 formulation of quantum theory  may have some advantage. First of all one should construct a theory of classification of quaternionic von Neumann algebra analogous to the classical
 one for complex (and real) algebras focusing on the interplay of the lattice of orthogonal projectors of a von Neumann algebra and the von Neumann algebra itself. Some very elementary  
steps towards  this classification are  the double commutant theorem (Theorem \ref{DCT}) and Proposition \ref{2.12}  we have established in this work.\\
Another interesting issue deserving a closer investigation  is the fact that Poincar\'e symmetry makes sense only in the absence of gravitation, i.e., in the framework of
 Special Relativity. This restriction leaves open the possibility that quaternionic formulations cannot be excluded when dealing with quantum field theory in (classical) curved spacetime.
All these problems will be investigated elsewhere.\\

\noindent {\bf Appendix}
\appendix

\section{Proof of some propositions}\label{AppProof}
\noindent{\bf Proof of Theorem \ref{stonetheorem}}. 
If viewing  $t\mapsto U_t$ as a one-parameter subgroup of unitary operators in $\sH_\bR$, Stone Theorem for real Hilbert spaces 
(e.g., see Thm 2.9 in \cite{MO1}) guarantees the existence of a unique antiself-adjoint operator $A$ in $\sH_\bR$ satisfying both the requirements of (\ref{stoneproperty}) with respect to the structure of $\sH_\bR$.
 If we manage to prove that $A$ commutes with $\cJ$ and $\cK$ we get an antiselfadjoint operator on $\sH$ 
 satisfying (\ref{stoneproperty}) on $\sH$  (notice that the derivative in (\ref{stoneproperty}) does not depend on the algebra of scalars of the Hilbert space, 
 as the norms of $\sH_\bR$ and $\sH$ coincide). For $x\in D(A)$,
\begin{equation*}
\lim_{t\to 0}\frac{U_t \cJ x-\cJ x}{t}=\lim_{t\to 0}\cJ\frac{U_t x-x}{t}=\cJ\lim_{t\to 0}\frac{U_t x-x}{t}=\cJ Ax
\end{equation*}
hence $\cJ x\in D(A)$ and $A\cJ x=\cJ Ax$. This means $\cJ A\subset A\cJ$ and so $\cJ A=A\cJ$, $\cJ$ being unitary and antiselfadjoint. 
The same argument applies to $\cK$. $A$ is then a quaternionic linear antiselfadjont operator satisfying (\ref{stoneproperty}). Now  we can focus on the one-parameter
 unitary subgroup $\bR \ni t\mapsto e^{tA}$ in $\sH$.  So, let $x,y\in D(A)$, then
\begin{equation*}
\begin{split}
&\left.\frac{d}{dt}\right|_{t=0} \langle x|U_{-t} e^{tA} y\rangle =\left.\frac{d}{dt}\right|_{t=0}\langle U_{t}x|e^{tA}y\rangle = \left\langle \left.\frac{d}{dt}\right|_{t=0}
 U_{t}x\bigg|y\right\rangle+\left\langle x\bigg|\left.\frac{d}{dt}\right|_{t=0}e^{tA}y\right\rangle=\\
&= \langle Ax|y\rangle+\langle x|Ay\rangle=0
\end{split}
\end{equation*}
This proves that $t\mapsto \langle x|U_{-t} e^{tA} y\rangle $ is constant and thus  $\langle x|U_{-t} e^{tA} y\rangle=\langle x|y\rangle$. Since 
$D(A)$ is dense and $U_{-t} e^{tA}\in\gB(\sH)$, we have $U_{-t} e^{tA}=I$, i.e. $U_t= e^{tA}$. Now, suppose that there exists another antiselfadjoint 
operator $B$ on $\sH$ such that $U_t=e^{tB}$. Using the discussion above for the one-parameter subgroup of the kind $e^{tC}$, we immediately
 see that $A=B$.  We have also proved, {\em en passant}, that $A$ coincides with the generator of $U$ with respect to $\sH_\bR$.  With the same 
 argument used for $\sH_\bR$, noticing that  $\sH_{\bC_j}=(\sH_\bR)_\cJ$ and $A$ commutes with $\cJ$, it is easy to prove that 
$A$ is also the antiselfadjoint of $U$ intepreted as one-parameter unitary group in $\sH_{\bC_j}$. $\Box$\\

\noindent{\bf Proof of Theorem \ref{PDT}}.
(a) Let us view $A$ as a $\bR$-linear operator in $\sH_\bR$ with $D(A)$  intepreted  as a real subspace of $\sH_\bR$. With this interpretation,
 $A$ is still densely defined closed $\bR$-linear operator ((e )Prop \ref{propHR})  and the operator $A^*A$ with this natural domain is densely defined, positive and self-adjoint
  on $\sH_\bR$ ((a) Thm 2.18 in \cite{MO1}). Again, Prop \ref{propHR} guarantees that $A^*A$ (with the
   same domain as in the real case) is a densely defined $\bH$-linear positive and selfadjoint operator in $\sH$. The proof of (a) is completed. \\
(b) Still interpreting $A$ as a $\bR$-linear operator in $\sH_\bR$, the polar decomposition theorem on real Hilbert spaces (e.g., Thm 2.18 in \cite{MO1}),
yields $A= UP$ for a unique pair of $\bR$-linear operators $U,P$ that satisfy properties (i)-(iv) and (1)-(3) in $\sH_\bR$. If we manage to prove that $U,P$ 
are $\bH$-linear these properties will be proved valid also in $\sH$. Indeed self-adjointness and positivity are preserved when moving from $\sH_\bR$ to
 $\sH$ giving (ii); the norms of the two involved Hilbert spaces are equal to each other and so point (iii) immediately follows; point (i) and (iv) are just
  function properties, not depending on the scalar field.   Finally, properties (1)-(3) immediately derive from the analogous properties valid in $\sH_\bR$ 
  (In particular, Lemma \ref{commst} implies that $\sqrt{A^*A}$ is the same operator in $\sH$, $\sH_\bR$ and $\sH_{\bC_j}$).
So, we need to prove that $U$ and $P$ commute with $\cJ,\cK$. Take $\cJ$ and  work in $\sH_\bR$. Since $A$ is $\bH$-linear, it commutes with $\cJ$, 
hence $A=\cJ^*A\cJ=\cJ^*(UP)\cJ=(\cJ^*U\cJ)(\cJ^*P\cJ)$. Since  $\cJ$ is unitary and anti selfadjoint, we see that $\cJ^*U\cJ$ and $\cJ^*P\cJ$ satisfy (i)-(iv) and so, 
by uniqueness of the polar decomposition in $\sH_\bR$, we have $\cJ^*U\cJ=U$ and $\cJ^*P\cJ=P$.  The same argument applies to $\cK$ proving that $U$ and $P$ are $\bH$-linear as wanted.
To complete the proof of (b), it is enough to prove that  the polar decomposition of $A$ in   $\sH$ is unique. Suppose $A=VQ$ on $\sH$
 for some $V,Q$ satisfying (i)-(iv).  Using an argument similar to that exploited above, we see that $A=VQ$ is a polar decomposition of $A$ in  $\sH_\bR$ and so, by  uniqueness, $V=U$ and $Q=P$.\\
(c)  The fact that the polar decompositions of $A$ in $\sH$ and $\sH_\bR$ coincide has been just established. A strictly analogous argument proves the
 same result for the polar decompositions of $A$ in $\sH_\bR$ and $\sH_{\bC_j}$. $\Box$\\

\noindent{\bf Proof of Proposition \ref{operatorfunctioncompl}}.
(a) If $JA\subset AJ$,  Lemma \ref{commst} implies $P^{(A)}_EJ=JP^{(A)}_E$ for every Borelian $E\subset \bR$ and Prop. \ref{propHJop} easily 
implies that $\mathcal{B}(\bR)\ni  E\mapsto (P_E^{(A)})_J \in \gB(\sH_J)$ is a complex PVM.
Let $x\in\sH$ be generic. 
Since the Hermitian scalar product of $\sH_J$ is the restriction of $\langle\cdot|\cdot\rangle$ to $\sH_J$ and the self-adjoint operator $A_J$
 is the analogous restriction of the self-adjoint operator $A$, 
Lemma \ref{lemmast} guarantees that $\mathcal{B}(\bR)\ni  E\mapsto (P_E^{(A)})_J \in \gB(\sH_J)$ is the PVM associated with $A_J$ in $\sH_J$. 
The proof of the last statement is now trivial.\\
(b) If $x\in\sH$, using notations as in Lemma \ref{lemmast},  $\mu_{Jx}^{P^{(A)}}(E)=\langle Jx|P^{(A)}_E Jx\rangle= \langle x|J^*P^{(A)}_EJx\rangle=\langle x|P^{(A)}_Ex\rangle=\mu^{P^{(A)}}_x(E)$.
 If, in particular,  $x\in D(f(A))$, we have $\int_\bR|f(\lambda)|^2\, d\mu^{P^{(A)}}_{Jx}(\lambda)=\int_\bR|f(\lambda)|^2\, d\mu^{P^{(A)}}_{x}(\lambda)<\infty$, which means $Jx\in D(f(A))$ and 
\begin{equation*}
\langle x|J^*f(A)Jx\rangle=\langle Jx|f(A)Jx\rangle=\int_\bR f\, d\mu_{Jx}^{P^{(A)}}= \int_\bR f\, d\mu_x^{P^{(A)}} = \langle x|f(A)x\rangle\:.
\end{equation*}
As a consequence,  $\langle x|(J^*f(A)J-f(A))x\rangle=0$ if $x\in D(f(A))\subset D(f(A)J)$.  Since both $f(A)$ and $J^*f(A)J$ are selfadjoint 
operators, $J^*f(A)J-f(A)$ is symmetric and so the vanishing of the scalar product ensures that $J^*f(A)J=f(A)$ on $D(f(A))$, i.e. $Jf(A)\subset f(A)J$. Since $J$ is 
unitary and antiselfadjoint, then $Jf(A)=f(A)J$ as wanted. Let us conclude the proof of (b). As a byproduct of the first part of (b), Prop.\ref{propHJop} implies that
 the restriction $f(A)_J$  a well-defined self adjoint operator of $\sH_J$. Let us focus on the operator $f(A_J)$ defined  on $\sH_J$ through the complex spectral Theorem.
$x\in D(f(A_J))$ iff $\int_\bR|f(\lambda)|^2\, d\mu_x^{P^{(A_J)}}<\infty$. Since for (a) $\mu_x^{P^{(A_J)}}=\mu_x^{P^{(A)}_J}$ when $x\in\sH_J$, we have that $x\in D(f(A)_J)\subset D(f(A))$ and furthermore
\begin{equation*}
\langle x|f(A_J)x\rangle = \int_\bR f\, d\mu_x^{P^{(A_J)}}=\int_\bR f\, d\mu_x^{P_J^{(A)}}=\langle x|f(A)x\rangle=\langle x|f(A)_Jx\rangle
\end{equation*} 
therefore  $\langle x|(f(A_J)-f(A)_J )x \rangle=0$ for  $x\in D(f(A_J))\subset D(f(A)_J)$. As $f(A_J)-f(A)_J$ is symmetric on $D(f(A_J))$, it is easy to prove that, $f(A_J)-f(A)_J=0$ on $D(f(A_J))$, i.e. $f(A_J)\subset f(A)_J$. 
Taking the adjoint we also have $f(A_J)^*\supset (f(A)_J)^*$.
Since both operators  are selfadjoint, it must be $f(A_J)= f(A)_J$.\\
(c) From the uniqueness of the polar decomposition for closed operators (see \cite{M2} for the complex case) it is easy to see that a selfadjoint operator $X$ is positive iff $X=|X|$. Now, thanks to (b) we
 have $|A|_J=(\sqrt{A^*A})_J=\sqrt{(A^*A)_J}=\sqrt{A^*_JA_J}=|A_J|$ which gives $A\ge 0$ iff $A=|A|$ iff $A_J=|A|_J$ iff $A_J=|A_J|$ iff $A_J\ge 0$, concluding the proof of (c). \\
(d) As we know (Theorem 4.8 \cite{GMP1}(b)), if $A=A^*$ its spherical spectrum is completely included in $\bR$  and there is not residual spectrum. Furthermore, as
 in the complex and real Hilbert space cases, $\sigma_S(A)$ is the support of $P^{(A)}$, namely, 
the complement of the largest open set $O\subset \bR$ such that $P^{(A)}_O=0$  (Theorem 6.6 (b) \cite{GMP2}).
Due to the last statement in (a), this is equivalent to saying that  $\sigma_S(A)$ is the complement of the largest open set $O\subset \bR$ such that $P^{(A_J)}_O=0$ so that $\sigma_S(A)=\sigma(A_J)$.
We also know that, exactly as in the real and complex Hilbert space cases, (Theorem 6.6 (d) \cite{GMP2})  $\lambda \in \sigma_{pS}(A)$ iff $P^{(A)}_{\{\lambda\}}\neq 0$. 
Due to the last statement in (a), this is equivalent to 
$P^{(A_J)}_{\{\lambda\}}\neq 0$, which is equivalent to   $\lambda \in \sigma_{p}(A_J)$. Since, for self-adjoint operators in quaternionic Hilbert spaces we have
 $\sigma_{cS}(A)= \sigma_S(A)\setminus \sigma_{cS}(A)$,  exaclty as in 
real and complex cases, where
$\sigma_{c}(A_J)= \sigma(A_J)\setminus \sigma_{p}(A_J)$, we also have $\sigma_{cS}(A)=\sigma_c(A_J)$
This concludes the proof.  $\Box$\\

\noindent{\bf Proof of Proposition \ref{lemmaextensreal}}.
(a). It is sufficient to observe that $JX\subset XJ$ and $KX\subset XK$ imply  that the  vectors  of the real subspace $D(X) \cap \sH_{JK}$ 
are mapped to $\sH_{JK}$ by  $X$ and thus $X_{JK} : D(X) \cap \sH_{JK} \to \sH_{JK}$ is a well defined $\bR$-linear operator in $\sH_{JK}$.\\
(b). Take $A$ as in the hypothesis, referring to the decomposition of Prop.\ref{realdecomp}, consider  the quaternionic linear 
subspace of $\sH$ defined as  $D(\widetilde{A}):=D(A)\oplus D(A)i\oplus D(A)j\oplus D(A)k$. From (\ref{defL}),  $L_q(D(\widetilde{A}))\subset D(\widetilde{A})$ 
and $D(A)=D(\widetilde{A})\cap\sH_{JK}$. Now, if $x=x_1+x_2i+x_3j+x_4k\in D(\widetilde{A})$, define
\begin{equation}\label{extension}
\widetilde{A}x:=Ax_1+(Ax_2)i+(Ax_3)j+(Ax_4)k
\end{equation}
which is a $\bR$-linear operator on $\sH$ extending $A$ and also a (right) quaternionic linear operator as it can be proved by direct inspection.
Finally one easily proves that  $J\widetilde{A}\subset\widetilde{A}J$ and $K\widetilde{A}\subset \widetilde{A}K$ using the given definition of $\widetilde{A}$.  Let us pass to the uniqueness property.
Suppose there exists a $\bH$-linear operator $B$ in $\sH$ which extends $A$ and such that $L_q(D(B))\subset D(B)$ and $D(A)=D(B)\cap \sH_{JK}$.  From this equality we see that $D(\tilde{A})\subset D(B)$ 
as $D(B)$ is a quaternionic linear subspace. Now, take any $x\in D(B)$ and decompose it as 
$x=x_1+x_2i + x_3j+x_4k$ with $x_r \in \sH_{JK}$ for $r=1,2,3,4$ according to Prop.\ref{realdecomp}. So,
\begin{equation*}
\begin{split}
D(B)\ni (L_jx)j&=(-x_3+ x_4i+ x_1j-x_2k)j=-x_1 + x_2i-x_3j+x_4k=\\
&=(-x_1-x_3j)+(x_4+x_2j)k
\end{split}
\end{equation*}
Hence, we get 
\begin{equation}\label{eq1}
\frac{1}{2}(x-(L_jx)j)=x_1+x_3j\in D(B)\mbox{ and }-\frac{1}{2}(x+(L_jx)j)k=x_4+x_2j\in D(B)
\end{equation}
Similarly, 
\begin{equation}\label{eq2}
-(L_k(x_1+x_3j))k=x_1-x_3j\in D(B) \mbox{ and }-(L_k(x_4+x_2j))k=x_4-x_2j\in D(B)
\end{equation}
Combining (\ref{eq1}) and (\ref{eq2}) together we can easily see that $x_r\in D(A)$ for every $r=1,2,3,4$. As a consequence, $D(B)\subset D(A)\oplus D(A)i\oplus D(A)j\oplus D(A)k=D(\tilde{A})$. 
Summing up we have $D(B)=D(\tilde{A})$ and $B=\widetilde{A}$ from the $\bH$-linearity of $B$. \\
(c) Let us pass to prove the properties (i)-(viii).
(i) holds  since  $J\widetilde{A}\subset\widetilde{A}J$ and $K\widetilde{A}\subset \widetilde{A}K$ are valid, (\ref{defL})  implies   
$L_q\widetilde{A} \subset \widetilde{A}L_q$. Exploiting the properties of $L$ it is straightforward to see that it actually holds $L_q\widetilde{A} = \widetilde{A}L_q$. (ii) is evidently valid 
from $D(\tilde{A}):=D(A)\oplus D(A)i\oplus D(A)j\oplus D(A)k$ and   (\ref{extension}).
(iv)  Let $B$ be as in the hypotheses, so $x=x_1+x_2i+x_3j+x_4k\in D(\widetilde{AB})$ iff $x_r\in D(AB)$  for $r= 1,2,3,4$, which is equivalent to saying   $x_r\in D(B)$ and $Bx_r\in D(A)$
for $r= 1,2,3,4$, which is equivalent to requiring 
  $x\in D(\widetilde{B})$ and $\widetilde{B}x\in D(\widetilde{A})$, which is finally equivalent to saying  $x\in D(\widetilde{A}\widetilde{B})$. 
  Therefore $D(\widetilde{AB})= D(\widetilde{A}\widetilde{B})$.
Since both $\widetilde{AB}$ and $\widetilde{A}\widetilde{B}$ extend $AB$ to  $\sH$,  commutes with $J$ and $K$ and 
$D(AB)= D(\widetilde{AB}) \cap \sH_{JK} =  D(\widetilde{A}\widetilde{B}) \cap \sH_{JK}$,  uniqueness of such an 
extension gives $\widetilde{AB}=\widetilde{A}\widetilde{B}$. The proofs of (iii) and (v) are analogous.
Let us pass to (vi). Suppose that $D(A) =\sH_{JK}$ and let $x\in\sH$
be decomposed as $x=x_1+x_2i+x_3j+x_4k$ with $x_r\in \sH_{JK}$ for $r=1,2,3,4$.
Exploiting (\ref{extension}) and Prop.\ref{realdecomp} we have $D(\widetilde{A})= \sH$ and
$\|\widetilde{A}\|\le \|A\| \leq +\infty$ because
$$
\|\widetilde{A}x\|^2=\sum_{n=1}^4\|Ax_n\|^2\le \|A\|^2\sum_{n=1}^4\|x_n\|^2=\|A\|^2\|x\|^2
$$
On the other hand, since $\widetilde{A}u=Au$ for  $0 \neq u\in \sH_{JK} \subset \sH$,  it must be $\|A\|\leq  \|\tilde{A}\| \leq +\infty$ and thus 
$\|\tilde{A}\| =\|A\| \leq +\infty$.  In particular $A\in \gB(\sH_{JK})$ immediately implies $\widetilde{A} \in \gB(\sH)$.
 The converse is trivially true as, if $\widetilde{A} \in \gB(\sH)$ then $D(A)= \gB(\sH)\cap \gB(\sH_{JK}) = \gB(\sH_{JK})$ and furthermore using the argument 
 above, $A \in \gB(\sH_{JK})$.  The proof of (vii) immediately follows from $D(\tilde{A}):=D(A)\oplus D(A)i\oplus D(A)j\oplus D(A)k$
 and (b) in Prop.\ref{realdecomp}. Let us pass to (viii).  Notice that $\widetilde{A^*}\subset (\tilde{A})^*$. indeed,
assume that   $x\in D(\widetilde{A^*})$ and $y\in D(\widetilde{A})$, then it holds with obvious notation
\begin{equation*}
\begin{split}
\langle x|\widetilde{A}y\rangle &= \sum_{\alpha=1}^4\sum_{\beta=1}^4\langle x_\alpha\imath_\alpha|(Ay_\beta)\imath_\beta\rangle=\sum_{\alpha=1}^4\sum_{\beta=1}^4 \overline{\imath_\alpha}\langle x_\alpha|Ay_\beta\rangle\imath_\beta=\\
&=\sum_{\alpha=1}^4\sum_{\beta=1}^4 \overline{\imath_\alpha}\langle A^*x_\alpha|y_\beta\rangle\imath_\beta = \langle \widetilde{A^*}x|y\rangle
\end{split}
\end{equation*}
As a consequence $\widetilde{A^*}\subset (\tilde{A})^*$. Let us prove the converse inclusion $\widetilde{A^*}\supset (\tilde{A})^*$ to conclude. It is enough
 establishing  that $D((\widetilde{A})^*)\subset D(\widetilde{A^*})$. To this end, suppose $x\in D((\widetilde{A})^*)$. By definition this means that there exists $z\in \sH$ 
 such that $\langle x|\widetilde{A}y\rangle=\langle z|y\rangle$ for every $y\in D(\widetilde{A})$. Take in particular  $y=y_1+0i+0j+0k\in D(A)\subset D(\widetilde{A})$, then
\begin{equation*}
\sum_{\alpha=1}^4\overline{\imath_\alpha}\langle x_\alpha|Ay_1\rangle=\langle x|\tilde{A}y\rangle=\langle z|y\rangle = \sum_{\alpha=1}^4\overline{\imath_\alpha}\langle z_\alpha|y_1\rangle\:,
\end{equation*}
where $\langle x_\alpha|Ay_1\rangle,\langle z_\alpha|y_1\rangle\in\bR$ since the restriction of $\langle\cdot|\cdot\rangle$ to $\sH_{JK}$ is real-valued. This implies
 that $\langle x_\alpha|Ay_1\rangle=\langle z_\alpha|y_1\rangle$ for every $\alpha$. Since $y_1\in D(A)$ is generic this means that $x_\alpha\in D(A^*)$ for every 
 $\alpha=1,2,3,4$, i.e., $x\in D(\widetilde{A^*})$ as wanted concluding the proof of (viii).
We prove (ix) concluding the proof. Suppose that  $\widetilde{A}$ is closable but $A$ is not, then there exits $\{a_n\}_{n\in \bN}\subset D(A)$ such that $a_n\to 0$ as $n\to +\infty$ and $A a_n\to y\neq 0$. 
Since  $A \subset \widetilde{A}$, we have $\widetilde{A}a_n\to y\neq 0$, which is impossible. Now, on the contrary, suppose that $A$ is closable and $\widetilde{A}$ is not. It must therefore  be  
$D(\widetilde{A})\ni x_n\to 0$ and $\widetilde{A}x_n \to z\neq 0$.  Proposition \ref{realdecomp} easily entails $(x_n)_r\rightarrow 0$  
and $A(x_n)_{r_0} \to z_{r_0}\neq 0 $ for at least one $r_0=1,2,3,4$,  which is impossible since $A$ is assumed to be closable. The fact that $\overline{\widetilde{A}}=\widetilde{\overline{A}}$ 
is now an easy consequence of the definition of $\widetilde{A}$, (b) in Proposition \ref{realdecomp} and the fact that, for general operators in real Banach spaces,
$x \in D(\overline{B})$ if and only if there exists a sequence $D(B) \ni x_n \to x$ such that  $Bx_n \to y_x$ for some $y_x$, where $\overline{B}x := y_x$.\\
The proof of (d) is elementary and follows form the definitions. $\Box$\\

\noindent {\bf Proof of Proposition \ref{gardintcompl}}.\\
(a)  $J(D_G^{(U)})\subset  D_G^{(U)}$  immediately arises from boundedness of $J$, $JU_g=U_gJ$ and the definition of $D_G^{(U)}$. As a consequence 
 $J(J(D_G^{(U)}))\subset  J(D_G^{(U)})$ so that $D_G^{(U)} \subset J(D_G^{(U)})$ because $JJ=-I$ and $D_G^{(U)}$ is a subspace. Summing up, $J(D_G^{(U)})=  D_G^{(U)}$.\\
(b)  Since  $U_gJ=JU_g$, every $U_g$ admits  $\sH_J$ as invariant space giving rise to a unitary operator for Prop.\ref{propHJop} and  (b) is manifestly true. \\ (c) If $x\in\sH_J$, 
define $x(f)$ as the only vector of $\sH_J$ such that $\langle y|x(f)\rangle=\int_G f(g) \langle y|U_gx\rangle\, dg$ for all $y\in\sH_J$. Let $x[f]$ be the only vector of $\sH$ 
such that $\langle z|x[f]\rangle=\int_G f(g)\langle z|U_g x\rangle\, dg$ for all $z\in\sH$. We can decompose  $z=z_1+z_2k$ with $z_1,z_2 \in \sH_J$ due to (b) Prop.\ref{propHJ}, then
$$
\langle z|x(f)\rangle=\langle z_1|x(f)\rangle - k\langle z_2|x(f)\rangle k=\int_G f(g) \langle z_1|U_gx\rangle\, dg -k\left( \int_G f(g) \langle z_2|U_gx\rangle\, dg\right)= \: 
$$
$$
= \int_G f(g) \left(\langle z_1|U_gx\rangle - k\langle z_2|U_gx\rangle k\right) \, dg=  \int_G f(g) \langle z|U_gx\rangle\, dg
= \langle z|x[f]\rangle\:.
$$
Since this holds for any $z\in\sH$, it must be $x(f)=x[f]$. $D_G^{(U_J)}$ is made of $\bC_j$ complex linear combinations of vectors $x(f) (= x[f])$. As $D_G^U$ is 
closed with respect generic quaternionic linear combinations, $D_G^{(U_J)}\subset D_G^{(U)}$.  In particular $D_G^{(U_J)}\subset D_G^{(U)}\cap\sH_J$.
Let us prove the converse inclusion.
Take $\sum_{\alpha=1}^n  x_\alpha[f_\alpha]c_\alpha\in D_G^{(U)}\cap \sH_J$ for some $c_\alpha\in\bH, x_\alpha\in\sH$ and real-valued functions $f_\alpha\in C^\infty_0(G)$.  Thanks to 
Remark \ref{gardingrealcombin} $x[f]q=(xq)[f]$ and thus we can henceforth  suppose that $c_k=1$. In view of (b) Prop.\ref{propHJ},
 $x_\alpha=u_\alpha+v_\alpha k$ for some $u_\alpha,v_\alpha\in\sH_J$, then it is easy to see that $x_\alpha[f_\alpha]=u_\alpha[f_\alpha]+v_\alpha[f_\alpha]k$, where 
 $u_\alpha[f_\alpha],v_\alpha[f_\alpha]\in D_G^{(U_J)}\subset \sH_J$ (see discussion above) and so
\begin{equation*}
\sum_{\alpha}x_\alpha[f_\alpha]=\sum_{\alpha}u_\alpha[f_\alpha]+\left(\sum_\alpha v_\alpha[f_\alpha]\right)k\:.
\end{equation*}
Since $\sum_{\alpha}x_\alpha[f_\alpha]\in\sH_J$ it must be $\sum_\alpha v_\alpha[f_\alpha]=0$ and $\sum_{\alpha}x_\alpha[f_\alpha]=\sum_{\alpha}u_\alpha[f_\alpha]\in D_G^{(U_J)}$. We
 have found that $D_G^{(U)}\cap \sH_J \subset D_G^{(U_J)}$ concluding the proof of (c).
(d) First assume ${\bf M}={\bf A} \in \gg$. Consider $u({\bf A})=\tilde{A}|_{D_G^{(U)}}$ where $\widetilde{A}:D(\tilde{A})\rightarrow \sH$ is the anti-selfadjoint generator of
 $t\mapsto U_{\exp(t{\bf A})}$ on $\sH$. Thanks to lemma \ref{Lemmagenertorsrestrictoin} we have two important facts. First, it holds $\widetilde{A}J=J\widetilde{A}$ 
 which, together with $J(D_G^{(U)})= D_G^{(U)}$, yields $u({\bf A})J=Ju({\bf A})$ and so  we can consider $u({\bf A})_J$ which is clearly given
  by $u({\bf A})_J=\widetilde{A}|_{D_G^{(U)}\cap \sH_J}=\widetilde{A}|_{D_G^{(U_J)}}$. Second, it ensures that  $\widetilde{A}_J$ is the generator of $t\mapsto (U_{\exp(t{\bf A})})_J$ on $\sH_J$, 
and so, by definition, the map $u_J$ is necessarily given by $u_J({\bf A})=\widetilde{A}_J|_{D_G^{(U_J)}}=\widetilde{A}|_{D_G^{(U_J)}}$. 
Putting the two conclusions together we get the thesis for ${\bf M}={\bf A} \in \gg$. The extension to a generic element of $E_\gg$ is trivial making use of (\ref{defugen}). $\Box$\\

\noindent {\bf Proof of Proposition \ref{proplocfaithfulness}}. \\
Let us prove that (a) implies  (b). Assuming (a) is true, there is a neighborhood of $\{(I,0)\}$ where $f$ is injective, therefore the map $\bR^4 \ni t \mapsto f((I,t)) \in G$ cannot be the constant function 
always attaining the neutral element  $e$ of $G$, so (b) is valid.\\
Let us now prove that (b) implies  (a).  $Ker(f)$ is a normal subgroup of 
$SL(2,\bC)\ltimes \bR^4$ and is also closed because $f$ is continuous and, if $e\in G$ is the unit elements, $\{e\}$ is closed because $G$ is Hausdorff. Since $Ker(f)$ is a closed subgroup of a Lie group, it must be a
Lie subgroup due to Cartan theorem. As we prove below, the normal Lie subgroups of $SL(2,\bC)\ltimes \bR^4$ are (1) $SL(2,\bC)\ltimes \bR^4$ itself, (2) $\{(\pm I, t) \:|\: t \in \bR^4\}$ and its subgroups (3) 
$\bR^4 \equiv \{(I, t) \:|\: t \in \bR^4\}$, (4) $\{(\pm I, 0)\}$, (5) $\{(I,0)\}$. Since, by hypothesis 
$\bR^4 \ni t \mapsto f((I,t)) \in G$ is injective, $Ker(f)$ cannot be as in the cases  (1),(2),(3). Only (4) and (5) are possible. In both cases $f$ is injective at least  in a neighbourhood of $(I,0)$ so (a) is valid.\\
Regarding the last statement in the thesis of Proposition \ref{proplocfaithfulness}, observe that holding (b), from the above analysis we conclude that  $f$ is not injective if and only if
$Ker(f) = \{(\pm I, 0)\}$.  \\
To conclude we have to prove that the normal Lie subgroups of $SL(2,\bC)\ltimes \bR^4$ are those mentioned. More precisely (2), (3) and (4), since (1) and (5) are trivial. 
First one has to identify the connected normal Lie subgroups. They
correspond
to ideals of the Lie algebra invariant under the adjoint action. In the
Lie algebra,
the adjoint action on $\bR^4$ is irreducible and on the quotient we have the
adjoint action
of the real simple Lie algebra $sl(2,\bC)$ which is also irreducible.
Hence there is only one
non-trivial ideal corresponding to the translation group. This shows that
closed normal subgroups are either discrete or their identity component
consists of the translation group.
Let us start with the discrete case. Discrete normal subgroups of
connected Lie groups are central.
And it is easy to figure out the centre, which is $\{\pm I\}$ in $SL(2,\bC)$; so
this is the only discrete normal
subgroup.
If $N$ is a Lie normal subgroup whose identity component consists of
translations,
 then its quotient by $\bR^4$ is a discrete subgroup of $SL(2,\bC)$, hence
either trivial or the centre.
This argument produces precisely $3$ non-trivial closed normal subgroups, i.e.,  (2),(3) and (4)
and these are all.
$\Box$

\subsection*{\bf Acknowledgements}. The authors are grateful to Professor K.-H. Neeb for useful  discussions and technical suggestions.

\end{document}